\documentclass[11pt]{article}
\usepackage[T1]{fontenc}
\usepackage[utf8]{inputenc}
\usepackage{amsmath,amssymb,euscript}
\usepackage{bbm}
\usepackage{graphicx}
\usepackage{epic}
\usepackage{epsfig}
\usepackage{pstricks}
\usepackage{psfrag}
 \usepackage{curves}
 \usepackage{mathrsfs}
 \usepackage{comment}
\usepackage{placeins} 

\usepackage[authoryear]{natbib}
\usepackage[citecolor=blue]{hyperref}
\usepackage{lscape}

\newtheorem{thm}{Theorem}[section]
\newtheorem{lem}[thm]{Lemma}

\newtheorem{definition}[thm]{Definition}
\newtheorem{prop}[thm]{Proposition}

\newenvironment{proof}{\noindent {\bf Proof \phantom{9}}}
{\hfill $\square$ \vspace{0.25cm}}

\def\be{\begin{eqnarray}}
\def\ee{\end{eqnarray}}
\def\ben{\begin{eqnarray*}}
\def\een{\end{eqnarray*}}

\numberwithin{equation}{section}
\numberwithin{figure}{section}

\def\be{\begin{eqnarray}}
\def\ee{\end{eqnarray}}

\def\me{\medskip\noindent}

\newcommand{\Co}{\mathcal{C}}

\newcommand{\Var}{\mbox{Var}}
\newcommand{\card}{\mbox{Card}}

\newcommand{\supp}{\mbox{supp}}
\newcommand{\TnMRCA}{T_n^{\mbox{{\scriptsize MRCA}}}}

\def\N{\mathbb{N}}
\def\P{\mathbb{P}}
\def\R{\mathbb{R}}
\def\E{\mathbb{E}}
\def\X{\mathcal{X}}
\def\U{\mathcal{U}}

\def\ind{{\mathchoice {\rm 1\mskip-4mu l} {\rm 1\mskip-4mu l}
{\rm 1\mskip-4.5mu l} {\rm 1\mskip-5mu l}}}
\def\eg{\textit{e.g.} }
\def\ie{\textit{i.e.} }

\title{Inference with selection, varying population size and evolving population structure: Application of ABC to a forward-backward coalescent process with interactions}

\author{
Clotilde Lepers\thanks{Universit\'e de Paris, IAME, INSERM, F-75108 Paris, France; E-mail: \texttt{clotilde.lepers@inserm.fr}},\qquad
Sylvain Billiard\thanks{Univ. Lille, CNRS, UMR 8198 - Evo-Eco-Paleo, F-59000 Lille, France; E-mail: \texttt{sylvain.billiard@univ-lille.fr}}, \qquad
Matthieu Porte\thanks{CMAP UMR CNRS 7641, Ecole Polytechnique, route de Saclay, 91128 Palaiseau Cedex-France},\qquad
Sylvie M\'el\'eard\thanks{CMAP UMR CNRS 7641, Ecole Polytechnique, route de Saclay, 91128 Palaiseau Cedex-France; E-mail: \texttt{sylvie.meleard@polytechnique.edu}}, \qquad
Viet Chi Tran\thanks{LAMA, Univ Gustave Eiffel, Univ Paris Est Creteil, CNRS, F-77454 Marne-la-Vall\'ee, France; E-mail: \texttt{chi.tran@univ-eiffel.fr}},
}

\date{\today}

\begin{document}

\maketitle

\textit{Word count}:6337.\\

\textit{Running title}: Inference using a forward-backward coalescent with competitive interactions

\begin{abstract}
  Genetic data are often used to infer demographic history and changes or detect genes under selection. Inferential methods are commonly based on models making various strong assumptions: demography and population structures are supposed \textit{a priori} known, the evolution of the genetic composition of a population does not affect demography nor population structure, and there is no selection nor interaction between and within genetic strains. In this paper, we present a stochastic birth-death model with competitive interactions and  asexual reproduction. We develop an inferential procedure for ecological, demographic and genetic parameters. We first show how genetic diversity and genealogies are related to birth and death rates, and to how individuals compete within and between strains. {This leads us to propose an original model of phylogenies, with trait structure and interactions, that allows multiple merging}. Second, we develop an Approximate Bayesian Computation framework to use our model for analyzing genetic data. We apply our procedure to simulated data from a toy model, and to real data by analyzing the genetic diversity of microsatellites on Y-chromosomes sampled from Central Asia human populations in order to test whether different social organizations show significantly different fertility.
\end{abstract}



\bigskip


\section{Introduction}

Demographic, spatial or genetic structures affect genetic diversity because they determines genetic flows between lineages, relationships between individuals, and coalescent rates \citep{charlesworthetal03}.  In turn, genetic polymorphism within and between taxa is commonly used for estimating population structures \citep{goldsteinandchikhi, mulleretal} or demographic changes \citep{beichmanetal}, to infer population history, migration patterns, or to search for genes under selection \citep{stephan}. These methods are mostly based either on the site frequency spectrum, the identity per state or descent, or on summary statistics in an Approximate Bayesian Computation (ABC) framework \citep{beaumontzhangbalding}.\\
Statistical testing and model selection are generally performed under simplifying assumptions which allow computations of quantities such as the likelihood of a model, in particular under neutrality.For instance, under the Wright-Fisher model, the population size is supposed deterministic: it is known at any given time and independent of the composition of the population, \ie it is supposed that the mechanisms underlying the variations of the population size are extrinsic and without noise. Individuals thus compete for space but the carrying capacity of the environment does not change because of the evolution of the population itself, or because of extrinsic or intrinsic stochasticity. In birth-death models, population size can vary but populations can grow indefinitely because individuals do not interact. In addition, the Wright-Fisher and birth-death models are most often supposed neutral when used for demographic inference, \ie the reproduction and survival rates do not depend on the genetic lineage \cite[but see a recent birth-death model without interactions where rates can depend on mutations,][]{rasmussenandstadler}. \\
Yet, the assumptions of neutrality, extrinsic control of population size or non-interacting individuals are certainly often violated. For instance, genealogies of the seasonal influenza virus show important departure from neutrality which might suggest that selection and interaction between lineages are important enough to significantly affect evolution and the shapes of the phylogenetic trees \citep{bedfordetal,strelkowalassing}. Reproduction rates and carrying capacities have also been shown to depend on strains in the domesticated yeasts \citep{sporetal}, and the ecological literature contains many cases where competitive interactions vary among strains or species \citep{gallieni2017}. Finally, not explicitly including competition in spatially structured population leads to biological inconsistencies in population genetics models \citep{felsenstein75}. Developing {models and} inference methods which relax such hypotheses is thus a contemporaneous challenge, in order to improve our knowledge of the history and ecological features of species and populations. As emphasized by \cite{frostpybusgogviboudbonhoefferbedford}, this challenge is particularly important for the analysis of phylodynamics in clonal species such as viruses.\\
Some of these assumptions have been already relaxed. For instance, \cite{rasmussenandstadler} developed a model where reproductive and death rates can differ between lineages which can emerge because of spontaneous mutations. They applied their method on Ebola and influenza viruses in order to have estimate of fitness effects of mutations from phylodynamics. Indeed, variation of death and birth rates between lineages can affect viruses phylogenies, which can be detected and used to infer the effect of mutations. However, they supposed no interaction between lineages, discarding a possible effect of competition between viruses strains. \\
In this paper, we present a model and an inference method which allow the relaxation of several of these assumptions. First, in Section \ref{sec:micro}, we recall the stochastic process describing the eco-evolution of a structured population with ecological feedbacks \cite[introduced in][]{billiardferrieremeleardtran}. This model takes into account: i) A trait structure that can affect birth, death and competitive rates. The traits, which evolves because of mutations and selection, are seen as proxies for the species, taxa or strains; ii) Explicit competitive interactions between and within lineages; iii) Varying population sizes depending on the genetic composition of the population, \ie the carrying capacity depends on the ecological properties of existing strains (their birth, death and competitive rates). The model assumes that reproduction is asexual, that mutations affecting fitness are rare, and that neutral mutation follows an intermediate timescale between reproduction and death rates (the ecological timescale) and the rate at which mutations affecting fitness appear (the evolutionary timescale). {Second, in Section \ref{sec:def-FBcoal}, a new forward-backward coalescent process is proposed to describe the phylogenies in such a population. The forward step accounts for interactions, demography and evolution of trait structures, defining the skeleton on which the phylogenies of sampled individuals can be reconstructed in the backward step. Phylogenies of structured populations have been previously already modeled in nested coalescent models, \textit{e.g.} \citep{blancasbenitezduchampslambertsirijegousse,blancasbenitezguflerkliemtranwakolbinger,duchamps,verduausterlitzetal}, but in our case interactions within and between lineages, ecological feedbacks between selection and population size, and multiple coalescence mergers are taken into account. Contrarily to $\Lambda$-coalescent models proposed in the literature \citep{donnellykurtz_99,pitman,sagitov}, multiple merging here are not due to sweepstakes reproductive successes but they appear as a consequence of natural selection via mutation-competition and timescales.} Third, in Section \ref{sec:ABC}, we develop an ABC framework in order to estimate the parameters of the model from genetic diversity data. We show how ecological parameters such as individual birth and death rates, and competitive abilities can be estimated. Finally, we applied our inferential procedure on the one hand on simulated data from a eco-evolutionary toy model, and on the other hand, on genetic data from Y-chromosomes sampled in Central Asia human populations \citep{chaixetal,heyeretal} in order to test whether different social organizations can be associated with difference in fertility.

\section{The forward-backward coalescent model}\label{sec:fundations}

In the current work, we extend the population model developed in \cite{billiardferrieremeleardtran}  \citep[following][]{metzgeritzmeszenajacobsheerwaarden,champagnat06,champagnatmeleard} to include phylogenies and develop a statistical ABC procedure  that we apply on simulated and real datasets. The eco-evolution of a structured population with ecological feedbacks is described by a stochastic process. The population is structured by traits, considered as proxies for species, taxa or strains. These traits can affect birth, death and competitive rates, and new traits are generated by mutations. Explicit competitive interactions are modeled between individuals of the population with intensities depending on the traits, inducing varying population sizes depending on the genetic composition of the population. Also, a marker structure is added. Markers are assumed neutral in the sense that they have no impact on fecundity, survival or competition. They are introduced in the model to measure the neutral diversity and allow the reconstruction of the phylogenies. The model assumes asexual reproduction and complete linkage between traits and markers, and that the population evolves following three timescales. First, the ecological timescale: birth and death rates occur at a fast rate. Second, marker mutations arise slightly slower than the ecological timescale. Finally, mutations on the trait under selection occur at the slowest timescale. This reflects for instance that a large proportion of a genome is not composed of traits under selection. This happens for example in the influenza virus which shows a large diversity within seasons despite a very rapid evolution and adaptation \citep{neherbedford}.

\me Before precisely describing the application of the model to infer demographic and genetic parameters within an ABC framework, we summarize hereafter the main features and outcomes of the model.

\subsection{Genetic diversity in an eco-evolutionary dynamics with three timescales: The substitution Fleming-Viot process (SFVP)}\label{sec:micro}

We assume a population of clonal individuals characterized, on the one hand, by a trait $x\in {\cal X}\subset \mathbb{R}^d$, which affects the demographic processes  such as birth, death and competitive interactions between individuals and, on the other hand, by a vector of genetic markers $u\in {\cal U}\subset \mathbb{R}^q$ , supposed neutral (\emph{i.e.} $u$ does not affect the demographic process). Individuals with trait $x$ give birth at rate $b(x)$ and $d(x)$ is their intrinsic death rate. The competitive interactions between individuals with traits $x$ and $y$ add an effect $C(x,y)$ on the individual death rate. When the population is large,  the evolution of the population can be decomposed into the succession of invasions of favorable mutations on the trait $x$, because ecological processes are very fast, and the population jumps from one state to another.  The neutral marker also evolves between each adaptive jump, at a faster timescale that is compensated by mutations of small effect. Since the ecological parameters change after each adaptive jump on trait $x$ (the birth rate, death rate and the population size change), the evolution of the neutral marker also changes. Hence, even if the marker is neutral, its own evolution depends on the state of the population at a given time, especially on the competitive interactions $C(x,y)$ between individuals with traits $x$ and $y$. Overall, the joint eco-evolutionary dynamics of the neutral marker and the selected traits can be approximated by the so-called Substitution Flewing-Viot Process \cite[SFVP,][see  Appendix \ref{sec:maths} for details]{billiardferrieremeleardtran}.\\

\textit{Distribution of the trait $x$ between two adaptive jumps.} At the ecological timescale, when the population is large, $p$ strains with traits $x_1,\dots x_p$ can coexist. Between two adaptive jumps, the trait distribution in the population remains almost constant. Indeed, the size of subpopulations can vary but are expected to stay close to their equilibria $\widehat{n}(x_1; x_1,\dots, x_p),\dots \widehat{n}(x_p; x_1,\dots, x_p),$ given by the following competitive Lotka-Volterra system of ordinary differential equations (ODE) {that approximates the evolution in the ecological timescale}:
\begin{align}
 \frac{dn_t(x_j)}{dt}=\Big(b(x_j)-d(x_j)-\sum_{\ell=1}^p C(x_j,x_\ell)n_t(x_\ell)\Big)n_t(x_j),\ j\in \{1,\dots , p\}, \label{eq:lotka-volterra}
\end{align}
where {$n_t(x)$}  can be seen as the density of individuals of strain with trait $x$. 
The equilibrium $\widehat{n}(x_i ; x_1,\dots ,x_p)$ of the population of the strain with trait $x_i$ depends on the whole trait structure of the population which is in turn defined entirely by the set of traits present in the population (the arguments of $\widehat{n}$ given after the semicolon).\\

\textit{Change of the distribution of the trait $x$ during an adaptive jump.} In the timescale of trait mutations and in the population composed of $p$ strains with traits $x_1,\dots x_p$ and respective sizes $\ \widehat{n}(x_1 ; x_1,\dots, x_p), \dots, \widehat{n}(x_p ; x_1,\dots, x_p)$, when a mutation on trait $x_i$ occurs at time $t$,  a new strain is introduced with trait $x_i+h$ where $h$ is drawn in a distribution $m(x_i,h)dh$ (mutations on trait $x$ are not necessarily small, \ie selection can be strong). Whether the mutant strain invades or not the population depends on its invasion fitness defined by
\begin{equation}
 f(y ; x_1,\dots, x_p)= b(y)-d(y)- \sum_{j=1}^p \widehat{n}(x_j ; x_1,\dots, x_p) C(y,x_j) \label{def:fitness}
 \end{equation}
 \citep{metzgeritzmeszenajacobsheerwaarden,champagnat06,champagnatferrieremeleard}. The mutant strain invades with probability $\frac{[f(x_i+h ; x_1,\dots, x_p)]_+}{b(x_i+h)}$, in which case the population jumps to a new state given by the solution of the Lotka-Volterra ODE system (Eq. \ref{eq:lotka-volterra}) updated with the introduction of the mutant strain $(\widehat{n}(x_1 ; x_1,\dots, x_p, x_i+h),\dots \widehat{n}(x_i+h ; x_1, \dots, x_{p}, x_i+h))$. {In the new equilibrium, some former traits $x_1,\dots ,x_p$ may be lost.} The evolution of the trait can thus be described by a Polymorphic Evolution Sequence (PES), \ie the succession of the adaptive jumps of the population from one state to another \citep{champagnatmeleard2011}. For a visual abstract of the PES, see Fig. \ref{Fig:PES} in Appendix.

\bigskip
\textit{Evolution of the neutral marker.}
When the mutant strain with trait $x=x_i+h$ invades the population, {say at time 0}, an adaptive jump occurs. {Let us denote by $u$ the marker of the first mutant individual $(x,u)$.} Initially, the distribution of the neutral marker within strain $i$ and trait $x$, is thus composed of a single individual with marker $u$. 
The evolution of the marker distribution within this strain is given by $F_{t}^{u}(x,dv)$, the distribution at time $t$ of the marker values within the strain with trait $x$ given the initial value $u$. This distribution changes with time depending on the supposed mutation kernel on the marker, on the birth and death rates of individuals with trait $x$, and on the competitive interactions $C(x,y)$ with all the other individuals of any trait value $y \in \{x_1,\dots, x_p, x_i+h\}$. Between two adaptive jumps, assuming small marker mutations  but not necessarily small trait mutations, how the distribution $F_{t}^{u}(x,dv)$ evolves with time is given by the following stochastic differential equation \cite[see][]{billiardferrieremeleardtran} (derivation details and a more general form are given in  Appendix \ref{sec:maths})
\begin{equation}\label{eq:FV}
\int_{\mathcal{U}} \phi(v) F^{u}_{t}(x,dv) =   \phi(u)  +  b(x) \int_{0}^t \bigg(\int_{\mathcal{U}} \Delta \phi(v) F^{u}_{s}(x,dv)\bigg) ds+M^x_{t}(\phi).
\end{equation}
 The left side of the equation can be seen as the expectation of the distribution of the marker value at time $t$, where $\phi$ is a test function (supposed twice differentiable on $\mathcal{U}$). Different choices of functions $\phi$ will provide descriptors of the distribution $F_{t}^u$ (for example $\phi(v)=v$ gives the mean of the distribution).  The right side of the equation tells what is the expected form of the distribution. The first term on the right side gives the initial conditions: the first mutant with trait $x$ has a marker value $u$, hence the initial condition for the distribution is $\phi(u)$. The second term on the right side integrates the changes of the distribution which are only due to mutations on the marker between time $0$ (the invasion time of $x$) and $t$. Since mutation only occurs at birth, the rate at which $F$ changes with mutation is proportional to the birth rate $b(x)$. Within the integral, $\Delta \phi(v)$ is the Laplacian of the function $\phi$ which gives the rate of change of $F$ in all the dimensions of the marker values (which depends on the assumptions made on the mutation kernel and can be generalized, see Appendix \ref{sec:maths}). The last term $M_t^x(\phi)$ on the right side gives the changes of $F$ which are due to the ecological processes, \ie the fluctuations due to the birth and death of the individuals with trait $x$. $M_t^x(\phi)$ is a martingale \ie a square integrable random variable with mean 0 and variance
\begin{align}\label{eq:crochet-PMB-FV}
\Var(M^x_t(\phi))= \frac{2b(x)} {\widehat{n}(x ; x_1,\dots, x_p, x_i+h)}
\int_{0}^t  \E\bigg[\bigg(\int_{\mathcal{U}} \phi^2(v) F^{u}_{s}(x,dv) -   \Big(\int_{\mathcal{U}} \phi(v) F^{u}_{s}(x,dv)\Big)^2\bigg) \bigg] \ ds.
\end{align}
The fraction in the right hand side (r.h.s.) of Eq. \ref{eq:crochet-PMB-FV} corresponds to the demographic variance $\,2b(x)\,$ divided by the effective population size
\begin{equation}
\label{def:Ne}
N_e(x)=\widehat{n}(x ; x_1,\dots, x_p, x_i+h).
\end{equation}The population effective size, which partially governs the evolution of the diversity at the neutral marker, depends on the trait value $x$, but also on the whole trait distribution $x_1,\dots, x_p, x_i+h$. In particular, it means that the variance in the neutral diversity within the strain with trait $x$ depends on the competitive interactions of the latter with all the other strains.


 \subsection{Genealogies in a forward-backward coalescent with competitive interactions}\label{sec:def-FBcoal}

Genealogies are piecewise-defined and constructed by dividing time between intervals separating adaptive jumps of the PES, following a forward-backward coalescent process. Since the evolution of trait  $x$  depends on the current distribution of the traits in the population, the PES tree is constructed forward in time where the successive adaptive jump times are denoted by $(T_k)_{k\in \{1,\dots J\}}$, with $T_0=0$ and $J$ is the number of jumps that occurred before time $t$.                                                                                                                                  \me During the PES, a subpopulation with trait $x_i$ has its own coalescent rate on the markers which depends on its reproductive rate $b(x_i)$ and on the distribution of the traits in the whole population (Eq. \ref{def:Ne}). Genealogies are thus expected to be different among the different strains and between different adaptive jumps of the PES. Between adaptive jumps, since under our assumptions trait $x$ distribution and population size are supposed fixed, within-strains genealogies can be constructed backward in time. Given the PES during the time interval $[T_k,T_{k+1})$ ($k\in \{0,\dots J-1\}$)and the trait distribution $\{x_1,\dots x_p\}$,  the genealogy of the individuals within the strain with trait $x_i$ is obtained by simulating a Kingman coalescent with coalescence rate $\frac{2b(x_i)}{\widehat{n}\big(x_i ; x_1,\dots, x_p\big)}$ (Eq. \ref{eq:crochet-PMB-FV}).  When an adaptive jump occurs at time $T_k$, all lineages in the subpopulation of strain with trait $x_i$ instantaneously coalesce because a single mutant is always at the origin of a new strain during the PES. Note that coalescence is instantaneous under the assumptions underlying the PES, \ie at the timescale governing the evolution of the trait, the transition to fixation of the mutant trait is negligible. The allelic state at the marker is determined given the previously constructed genealogy, depending on the mutational model considered.

A more formal definition of the coalescent and associated proofs are given in App. \ref{sec:phylo}. A simulation algorithm for the construction of genealogies under our model is given in App. \ref{sec:simulations}.\\

\section{ABC inference in an eco-evolutionary framework}\label{sec:ABC}

We showed in the previous sections that the genetic structure of a sample of $n$ individuals can be related to the parameters of our eco-evolutionary model. We now aim at using this framework to infer genealogies, ecological and genetic parameters from genetic and/or phenotypic data sampled in a population at time $t$. In other words, given a dataset containing the genotype at the marker $u$ and the genotype or phenotype at the trait $x$ for the $n$ sampled individuals, we want to infer the parameters of the model: birth, death and competitive interaction rates, mutation rates, etc. Since we have only a partial information on the population ($n$ individuals are sampled and {possible extinct lineages are unobserved}), the likelihood of a model given the data has no tractable form. Given a possible genealogy of the $n$ individuals, an infinite number of continuous genealogical trees could be obtained from the model. The likelihood of each tree depends on the number and the traits of the different subpopulations (strains) during the history of the population, including the unobserved and extinct ones. Because summing over all possible unobserved data (number of unobserved and extinct lineages with their traits and adaptive jump times) is not feasible in practice, we have to make inference without likelihood computations.\\

{An alternative to likelihood-based inference methods is given by the Approximate Bayesian Computation (ABC) \citep{beaumontzhangbalding,beaumontcornuetmarinrobert}, which relies on repeated simulations of the forward-backward coalescent trees (Section \ref{sec:def-FBcoal}).} In the following, we briefly give a general presentation of the application of the ABC method to our model. We then apply the method to simulations of a toy model (the Dieckmann-Doebeli model) and to real data (genetic data on microsatellites on the Y chromosomes of human populations from Central Asia, with their social and geographic structures).

\subsection{ABC estimation of the ecological parameters based on the genealogical tree}
The dataset denoted $\mathbf{z}$ contains the genotype and/or phenotype on the trait $x$ and the marker $u$ for each of the $n$ sampled individuals. The trait $x$ can be geographic locations, species or strain identity, size, color, genotypes or anything that affect the ecological parameters and fitness. The marker $u$ can also be genotypic or phenotypic measures, discrete or continuous, qualitative or quantitative, but with no effect on fitness (the marker is supposed neutral). Our goal is to use the dataset $\mathbf{z}$ to estimate the parameters of the model denoted $\theta$ (in our case, birth and death rate, competition kernel, mutation probabilities and kernel) using an ABC approach. To do so, the following procedure is repeated a large number of times:
\begin{enumerate}
\item[$1^{st}$ step.] A parameter set $\theta_i$ is drawn in a prior distribution $\pi(d\theta)$;
\item[$2^{nd}$ step.] A PES and its neutral nested genealogies of the $n$ sampled individuals are simulated in each model associated with the parameters $\theta_i$;
\item[$3^{rd}$ step.] A set of summary statistics $S_i$ is computed from the data simulated under $\theta_i$, for each $i$.
\end{enumerate}
The posterior distribution of the model is then approximated by comparing, for each simulation $i$, the simulated summary statistics $S_i$ to the ones from the real dataset {and by computing for each parameter $\theta_i$ a weight $W_i$ that defines the approximated posterior distribution (see Formula \ref{def:posterior} in Appendix)}. Three categories of summary statistics have been used, each associated with a different aspect of the genealogical tree (the complete list of summary statistics is given in the Appendix \ref{sec:summarystatistics}):
 \begin{itemize}
 \item[-] The trait distribution describing the strains diversity and their abundances (\eg number of strains, the mean and variance of strains abundance, ...);
 \item[-] The marker distribution in the sampled population describing the neutral diversity within each sampled strain (\eg the M-index, $F_{st}$, Nei genetic distances,...);
 \item[-] The shape of the genealogy (\eg most recent common ancestor, length of external branches, number of cherries, ...).
 \end{itemize}

Depending on the dataset and the information that is available for a given population, four scenarios can be encountered:
\begin{itemize}
\item[Scenario 1.] {\bf Complete information}: The evolutionary history of the trait and the genealogies, populations and subpopulations abundances, values of the sampled individuals on the trait $x$ and the marker $u$. This situation certainly never occurs but it is a reference which allows to evaluate the expected ABC estimation in a perfect situation where all information is available. This situation can also include cases where independent information can be added such as fossil records;
\item[Scenario 2.] {\bf Population information}: Total population abundance, values of the trait $x$ and marker $u$ of the sampled individuals. The estimations given with those statistics represent the estimations one could expect with a complete knowledge of the present population;
\item[Scenario 3.] {\bf Sample information}: The number of sampled sub-populations, the values of the trait $x$ and the marker $u$ of the sampled individuals;
\item[Scenario 4.] {\bf Partial sample information}: Only the number of sampled sub-populations and the values of the marker $u$ of the sampled individuals.
\end{itemize}
The four situations will be compared regarding the quality of the ABC estimations of the model parameters.

\subsection{Application 1: Inference of the parameters in the Dieckmann-Doebeli model}\label{sec:Dieckmann-Doebeli}

In this section, we applied the ABC statistical procedure on the traits distribution and their phylogenies generated by a simple eco-evolutionary model \citep{roughgarden, dieckmanndoebeli, champagnatferrieremeleard}. The birth rate of an individual with trait $x$ is $b(x)=\exp(-x^2/2\sigma_b^2)$, the individual natural death rate is constant $d(x)=d_C$, and the competition between two individuals with traits $x$ and $y$ is $C(x,y)=\eta_c\,\exp(-(x-y)^2/2\sigma^2_c)$, $\sigma_c>0$ . The trait space is chosen to be $\X=[-1,1]$. The effect of a mutation on the trait $x$ is randomly drawn in a Gaussian mutation kernel with mean 0 and variance $\sigma_m^2$ (values outside $\X$ are excluded). The probability of mutation is $p$. The markers are assumed to be a vector of 10 microsatellites, each of them mutating with the same rate $q$. When a microsatellite mutates, we increase or decrease its value by 1 with equal probability.\\
The distribution of the phylogenies depends on the parameter $\theta=(p,q ,\sigma_b,\sigma_c,\sigma_m,d_c, \eta_c,t_{\sc{sim}})$, where $t_{\sc{sim}}$ is the duration of the PES ($t_{\sc{sim}}$ is not known \textit{a priori} and must be considered as a nuisance parameter).

\subsubsection{Posterior distribution and parameters estimation}
We ran $N=400\ 000$ simulations with identical prior distributions and scaling parameter $K=1000$ (see details in App. \ref{sec:ABC-recap}). Chosen parameter sets and prior distributions are given in Appendix \ref{sec:simulations}. We randomly chose four simulations runs among the $N$ simulations as \textit{pseudo-data} sets (these sets are named $A$, $B$, $C$ and $D$, see App. Table \ref{Table:data_parameters} and Fig. \ref{Fig:PES_line3}). All other simulations runs were used for the parameters estimates. Fig \ref{Fig:PriorPost_line3} shows the posterior distribution for one of the the pseudo-dataset (see App. \ref{sec:posterior-distrib} for full results). Our results show that estimates based on all statistics (Scenario 1, blue distribution) are not always the most accurate, suggesting that some of the descriptive statistics introduce noise and worsen estimate accuracy. However, the descriptive statistics providing knowledge about how population is trait-structured do not belong to this group and importantly improves estimation when available (compare orange \textit{vs.} red posterior distributions).

The impact of the number of microsatellites on the quality of the estimation is tested for the first pseudo-dataset $A$ (see App. Tab. \ref{Table:data_parameters}) with the number of microsatellites varying from 10 to 100. A sensitivity analysis is shown in App. Fig. \ref{Fig:nb_microsat_all}: the results are quite robust to this number. For some parameters such as $t_{sim}$, better precision is achieved with increased number of microsatellites, and for other parameters such as $q$ or $p$, the impact of the number of microsatellites is more visible under Scenario 4 when we should rely a lot on the information brought by the microsatellites.

\begin{figure}[ht!]
\begin{center}
\begin{tabular}{ccc}
A & \includegraphics[width=7.5cm]{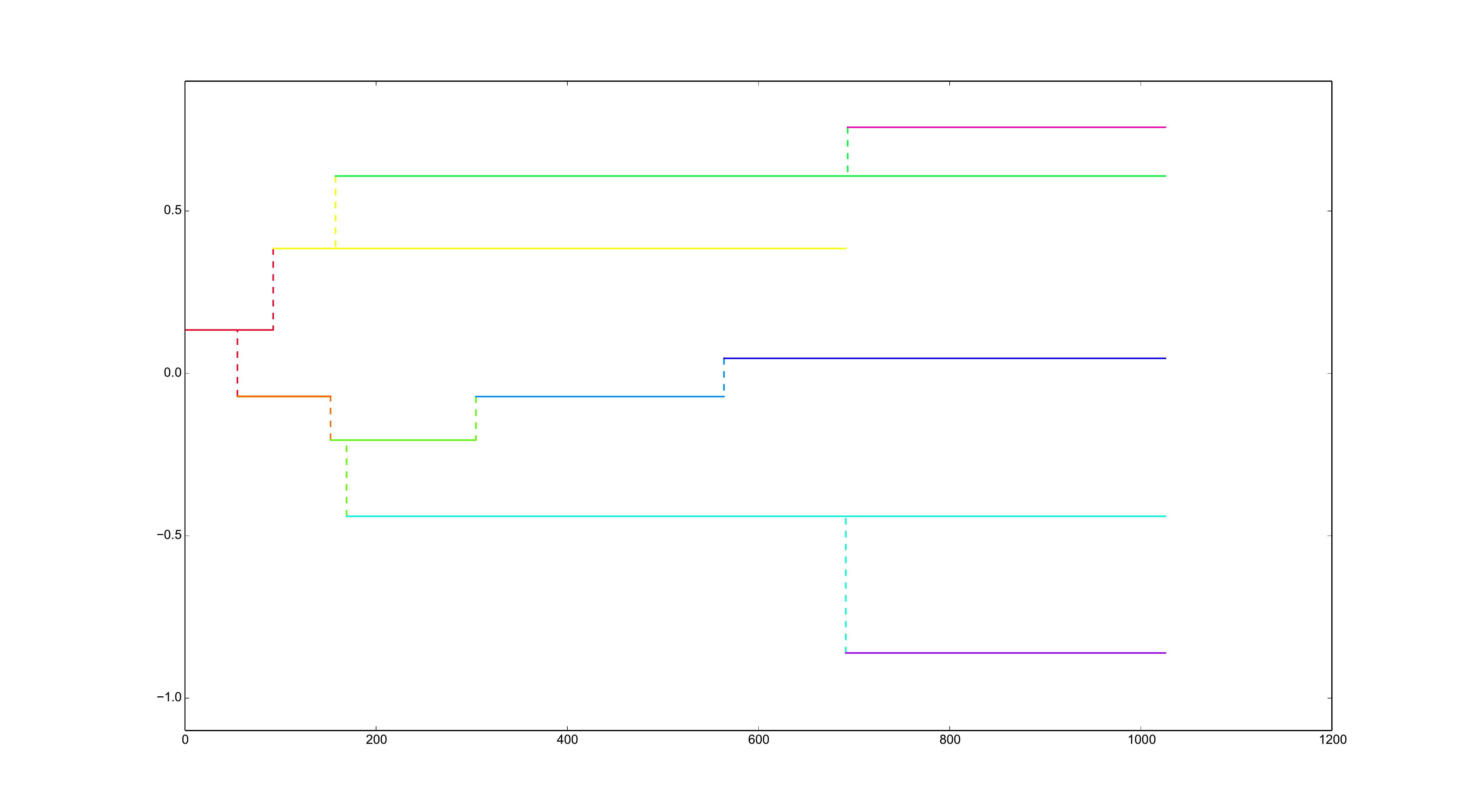} & \includegraphics[width=7.5cm]{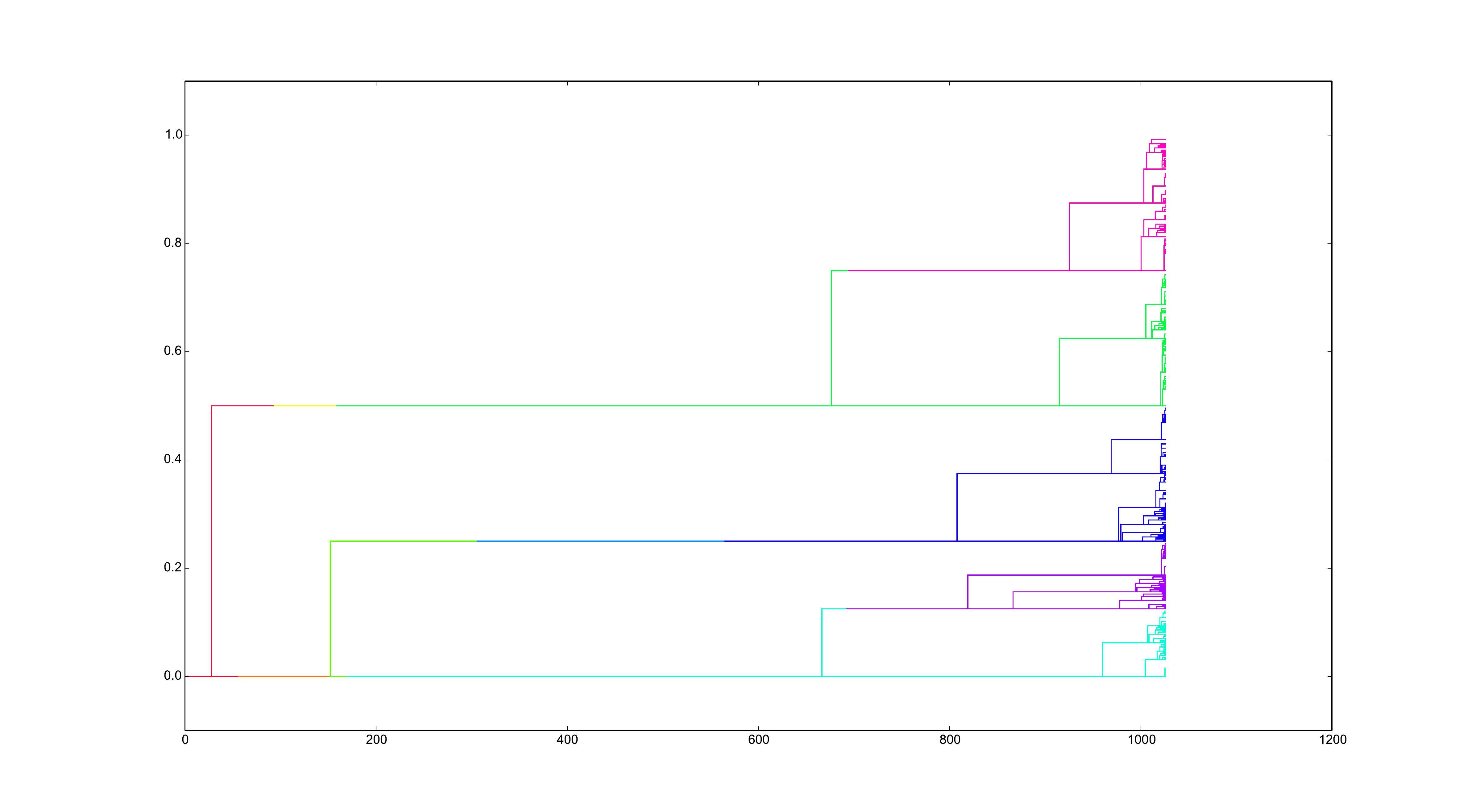}\\
B & \includegraphics[width=7.5cm]{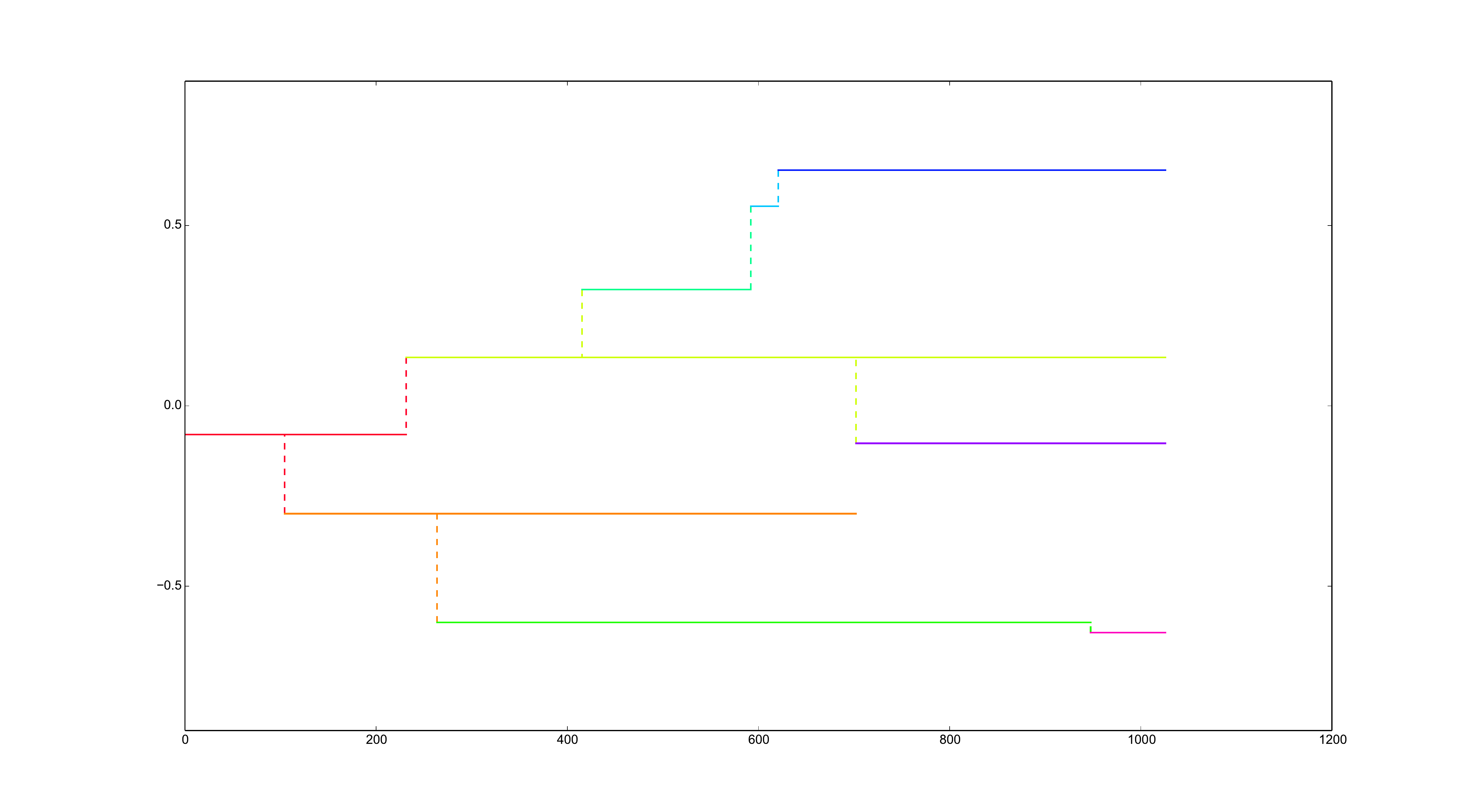} & \includegraphics[width=7.5cm]{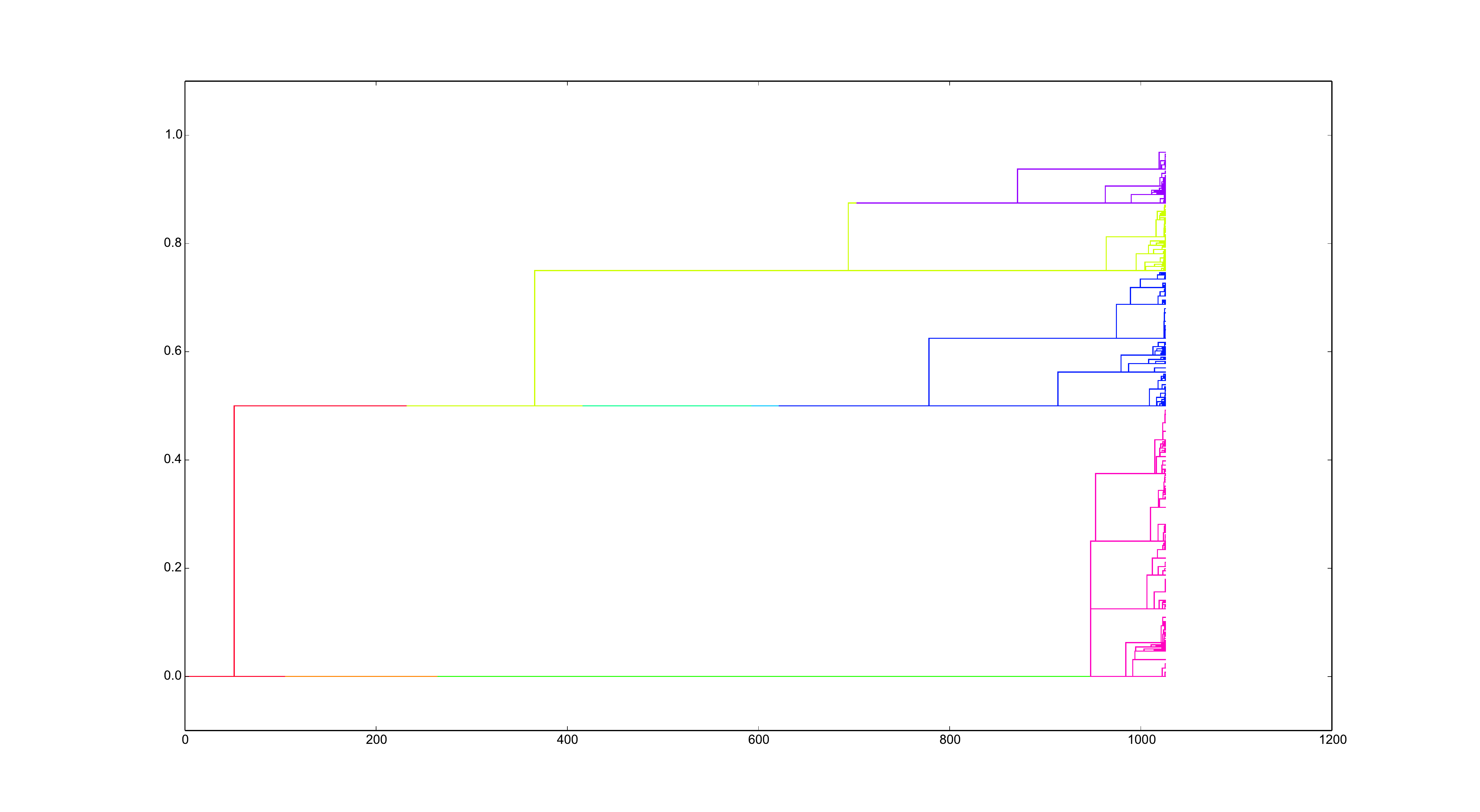}\\
C & \includegraphics[width=7.5cm]{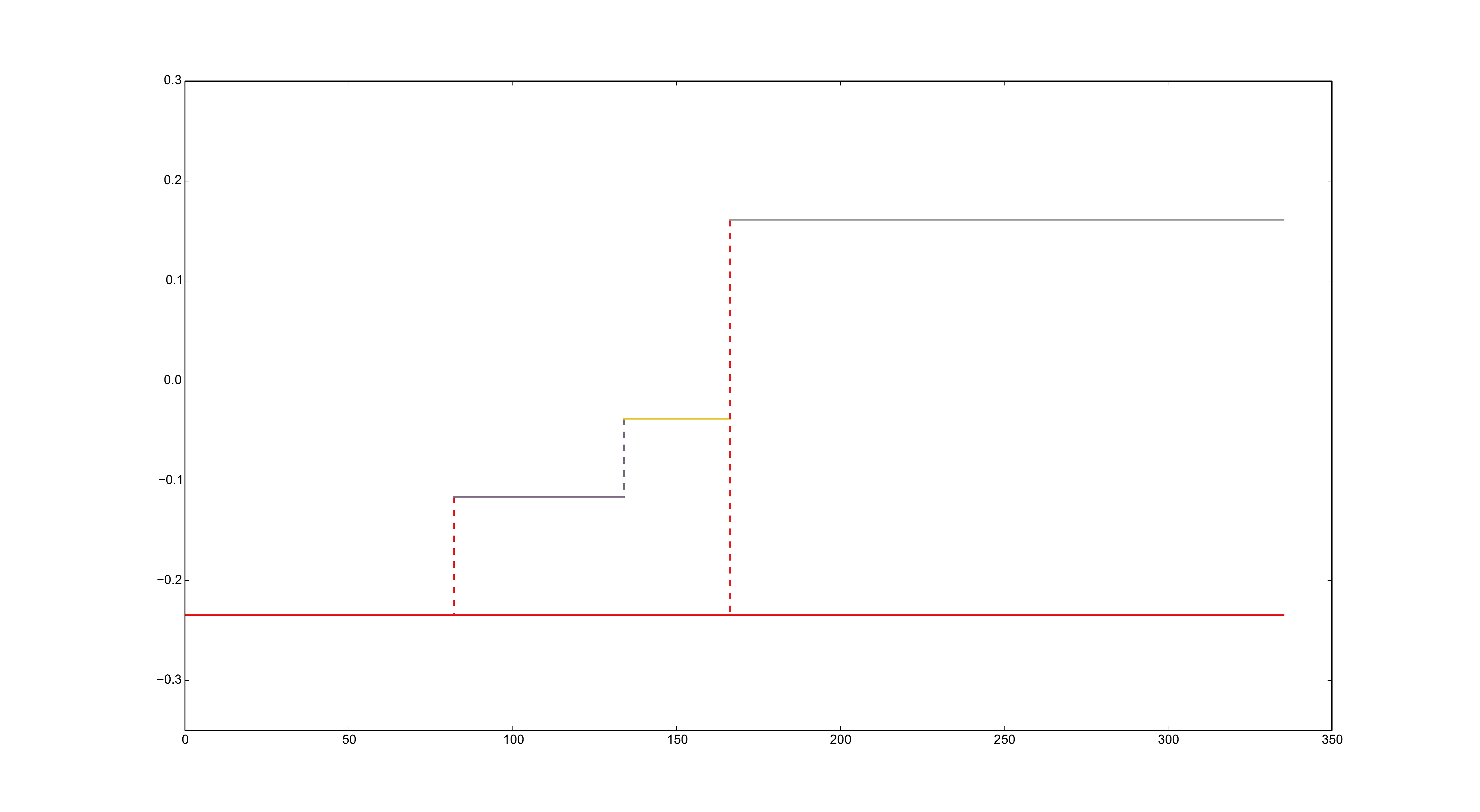} & \includegraphics[width=7.5cm]{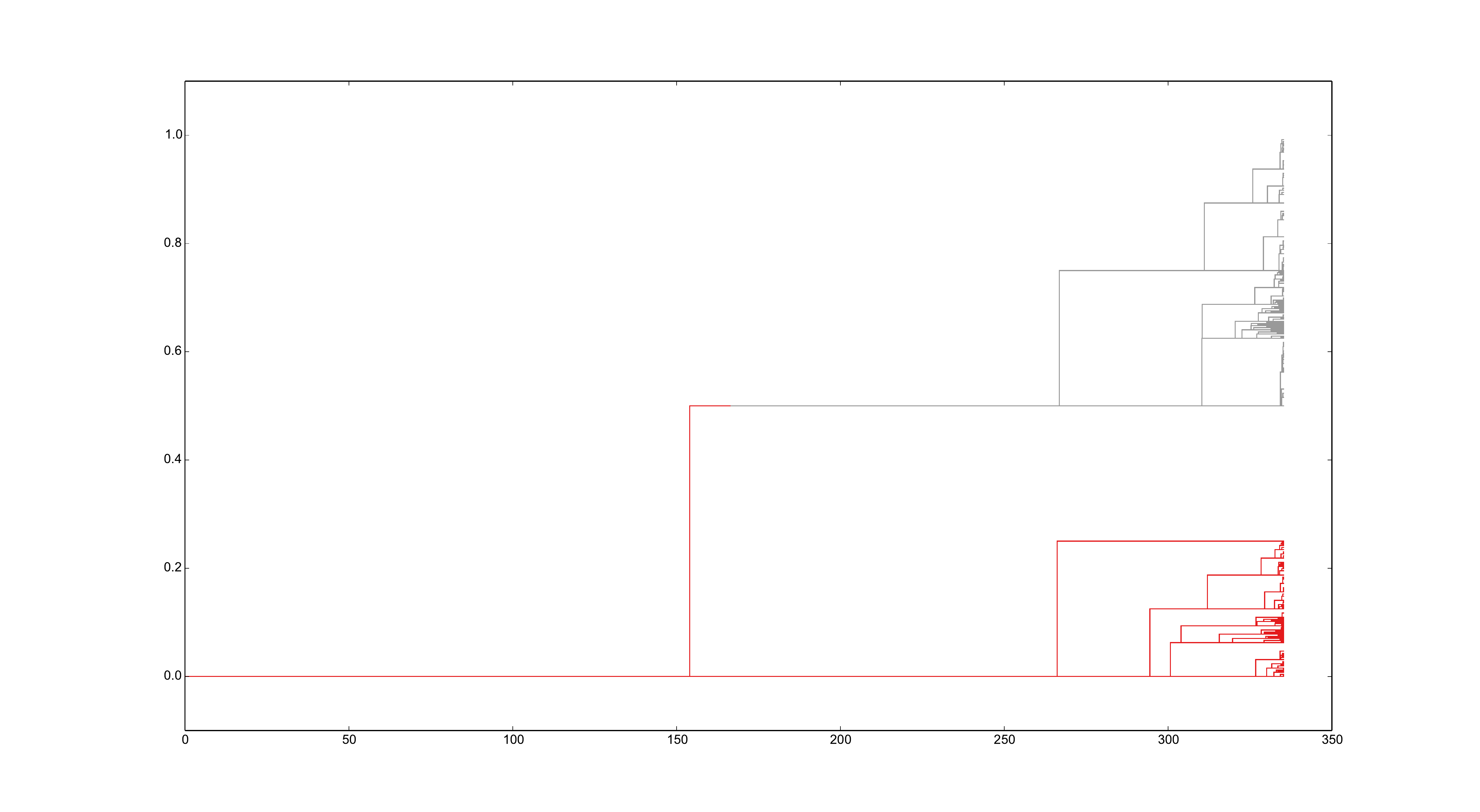}\\
D & \includegraphics[width=7.5cm]{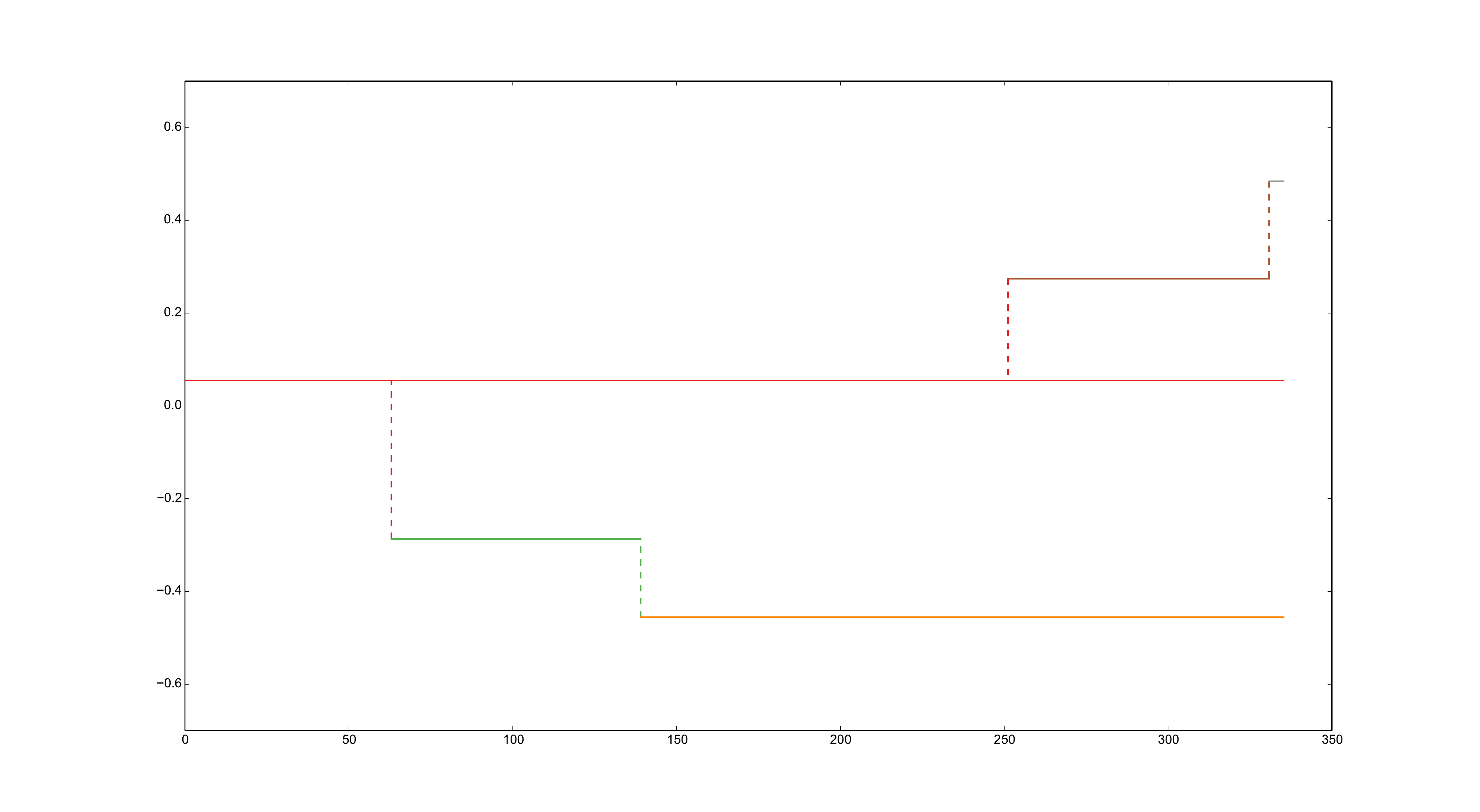} & \includegraphics[width=7.5cm]{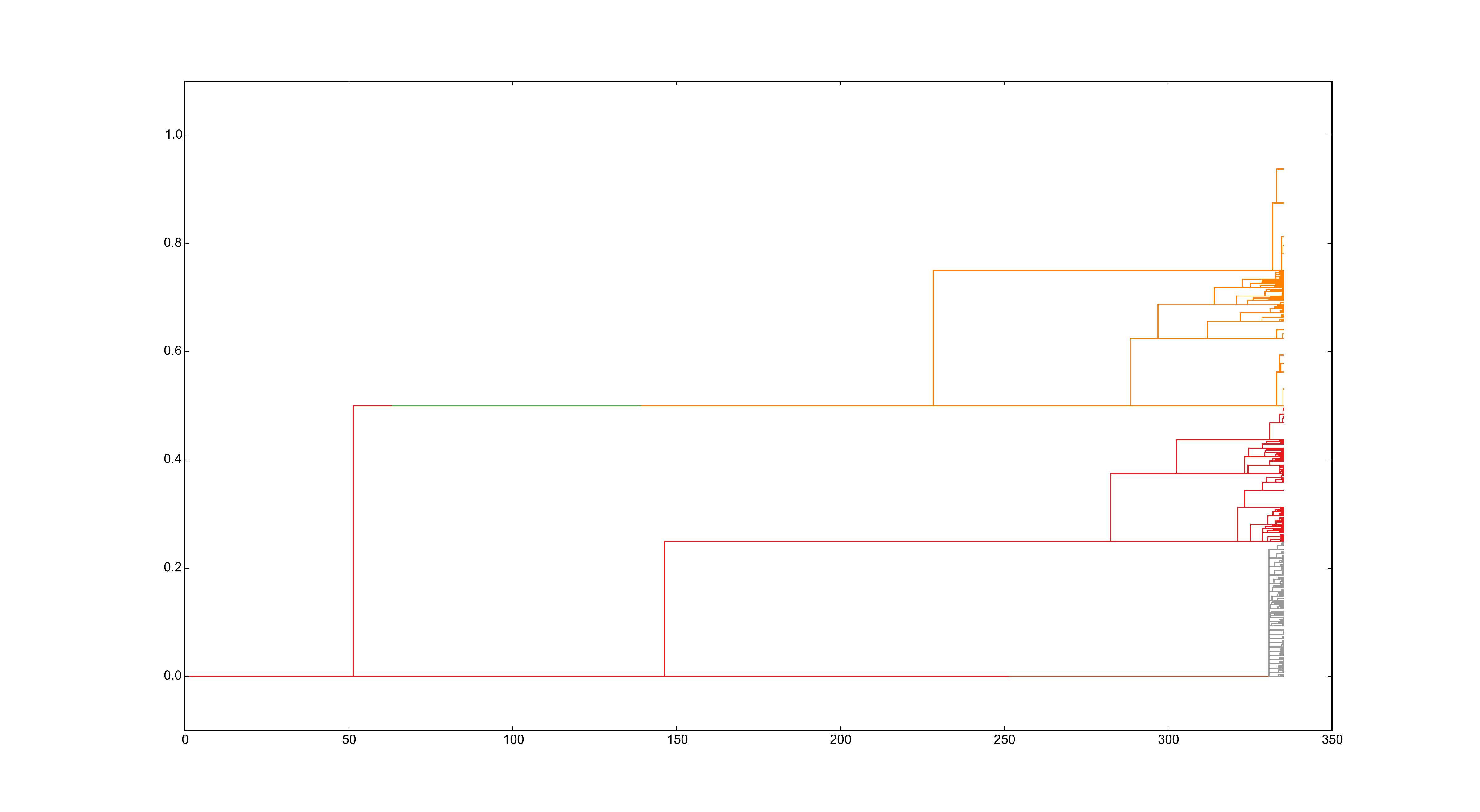}\\
 & (a) & (b)
\end{tabular}
\caption{{\footnotesize \textit{Dynamics of (a) the trait $x$ and (b) the neutral marker $u$ of the four pseudo-data sets $A$, $B$, $C$ and $D$ randomly sampled among $N=400,000$ simulations runs of the Doebeli-Dieckmann's model (Parameter sets are given in App. Table \ref{Table:data_parameters}). Figures show the Substitution Fleming-Viot Process (SFVP) and the nested phylogenetic tree of $n$ individuals sampled at the final time of the simulation. (a): The trait $x$ follows a Polymorphic Evolutionary Substitution (PES) process introduced in \cite{champagnatmeleard}. (b): The genealogies of the marker $u$ follow a forward-backward coalescent process nested in the PES tree as described in Section \ref{sec:def-FBcoal}. The colors refer to the lineage to which one individual belong shown in (a).}}}\label{Fig:PES_line3}
\end{center}\end{figure}

\clearpage
\begin{figure}[ht!]
\begin{center}
\includegraphics[width=14cm]{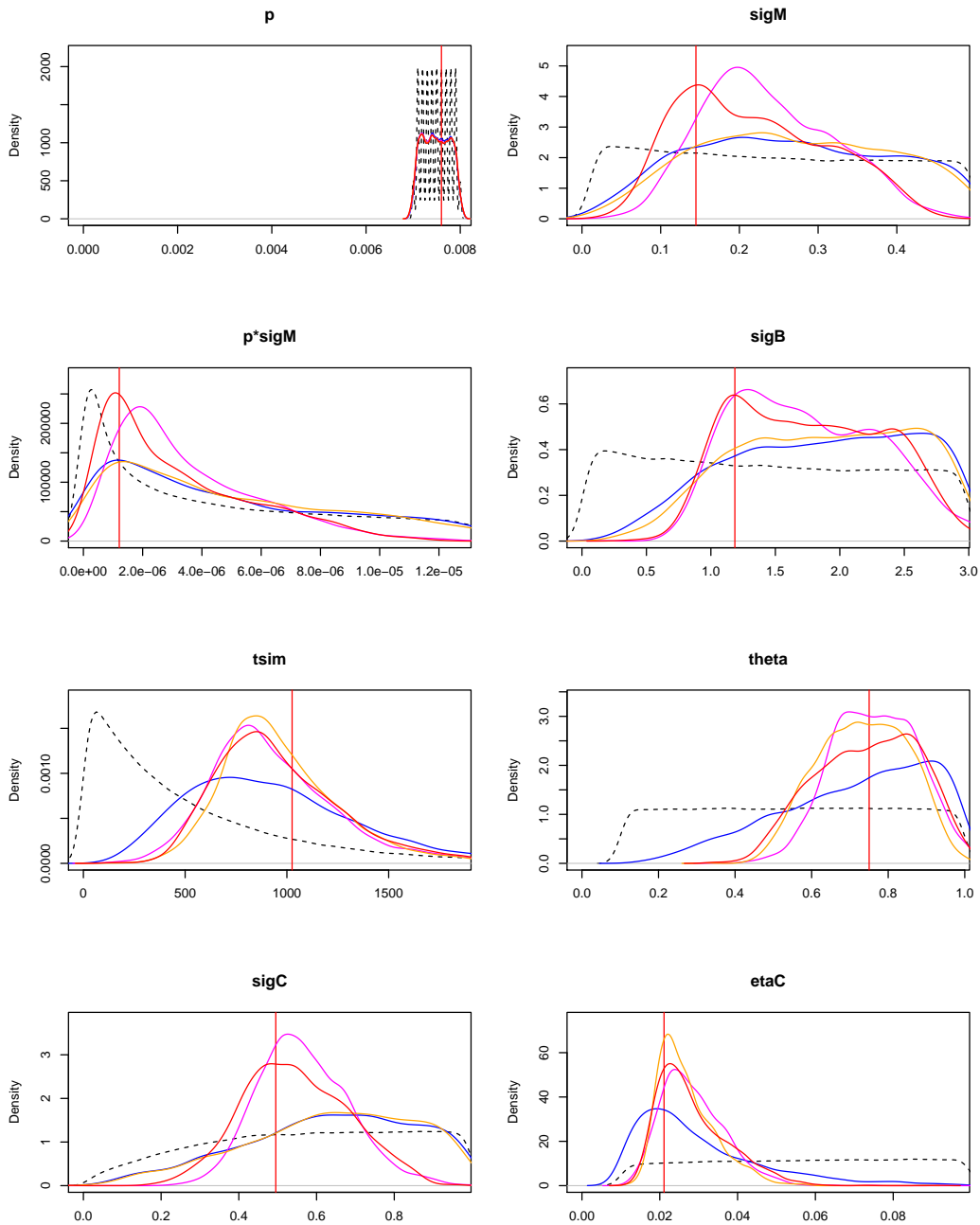}
\caption{\textit{\footnotesize Prior and posterior distributions (pseudo-data set $A$ in Fig. \ref{Fig:PES_line3}). Black dashed curve: prior distribution; Vertical red line: true value. The different colors correspond to different scenario regarding which data are available or not: Blue, Scenario 1 (All descriptive statistics are available); Pink, Scenario 2 (data from the totality of the population); Red, Scenario 3 (data from a sample of the population); Orange, Scenario 4 (data from a sample of the population, the traits $x$ is not known). Results for other pseudo-data sets are given in App.\ref{sec:posterior-distrib}}}\label{Fig:PriorPost_line3}
\end{center}
\end{figure}

\clearpage

\subsubsection{Discrepancy with Kingman's coalescent}\label{sec:docoalescenttreessignificantlydifferfromKingman}

After a correct renormalization, Kingman's coalescent are generally considered as a good approximation of coalescent trees, even in structured populations. However, in our model, the population structure itself can evolve, demographic rates can vary with time, and subpopulations can interact with each other, which might strongly affect the topology of the coalescent trees and its branches length. In this section, our aim is to evaluate to what extent the Kingman's coalescent is a good approximation or not of the genealogies generated by the Doebeli-Dieckmann's model. In case of a significant discrepancy, we further determined the properties of the trees which show important differences between both models, and then we identified and evaluated the type and extent of errors that one would expect when using Kingman's coalescents for inference without taking into account the evolution of population structure.\\

%


We considered statistics commonly used to test the neutrality of the phylogenies of $n$ sampled individuals \citep{fuli}: the number of cherries $C_n$, i.e. the number of internal nodes of the tree having two leaves as descendants, the length of external branches $L_n$, i.e. edges of the phylogenetic tree admitting one of the $n$ leaves as extremity, and the time $\TnMRCA$ to the most recent common ancestor (MRCA). The distributions of the normalized $C_n$ and $L_n$ and the distribution of $\TnMRCA$ for the forward-backward Doebeli-Dieckmann's coalescent and the Kingman's coalescent are compared. {For Kingman's coalescent, asymptotic normality has been established for $C_n$ and $L_n$ \cite[see][]{blumfrancois2005-Sackin,jansonkersting}.} The distribution of $\TnMRCA$ for the Kingman coalescent is computed by using the fact that the trees are binary with exponential durations between each coalescence. Neutrality tests conditionally on the number of lineages $m$ at the time of sampling are performed using the behavior of these statistics under the null assumption $H_0$ that the phylogenies correspond to a Kingman's coalescent. For each $m$, we chose as pseudo-data one of the simulations of our model with $m$ species at the final time, and we performed normality tests for $C_n$ and $L_n$, and an adequation test for the expected distribution under Kingman for $\TnMRCA$. This was repeated 100 times for each value of $m\in \{1,\dots 10\}$ (details given in App. \ref{sec:neutrality-test}).         \\

Fig. \ref{Fig:neutral-test} shows the distributions of the \textit{a posteriori} p-values for the normality tests for $L_n$ and $C_n$. The coalescent trees significantly differ from Kingman's coalescent trees regarding the external branch length $L_n$ (Fig. \ref{Fig:neutral-test}(a)), while the number of cherries $C_n$  is not always  significantly different (the p-values have a median close to 0.05, Fig. \ref{Fig:neutral-test}(b)). Finally, Fig. \ref{Fig:neutral-test}(c) shows the distribution of the time to the MRCA depending on the number of lineages $m$. A mean comparison test shows that the mean of the $\TnMRCA$s obtained from the simulations of our forward-backward coalescent significantly differs from the expected MRCA time under a Kingman's coalescent (see App. \eqref{test:comparaison-moyennes}). Hence, our results show that coalescent tree topologies generated under a Doebelli-Dieckmann's model are expected to be significantly different from a Kingman's coalescent. \\

\begin{figure}[ht!]
\begin{center}
\begin{tabular}{ccc}
\includegraphics[width=5cm,height=5cm, trim= 0 1cm 0cm 2cm, clip=true]{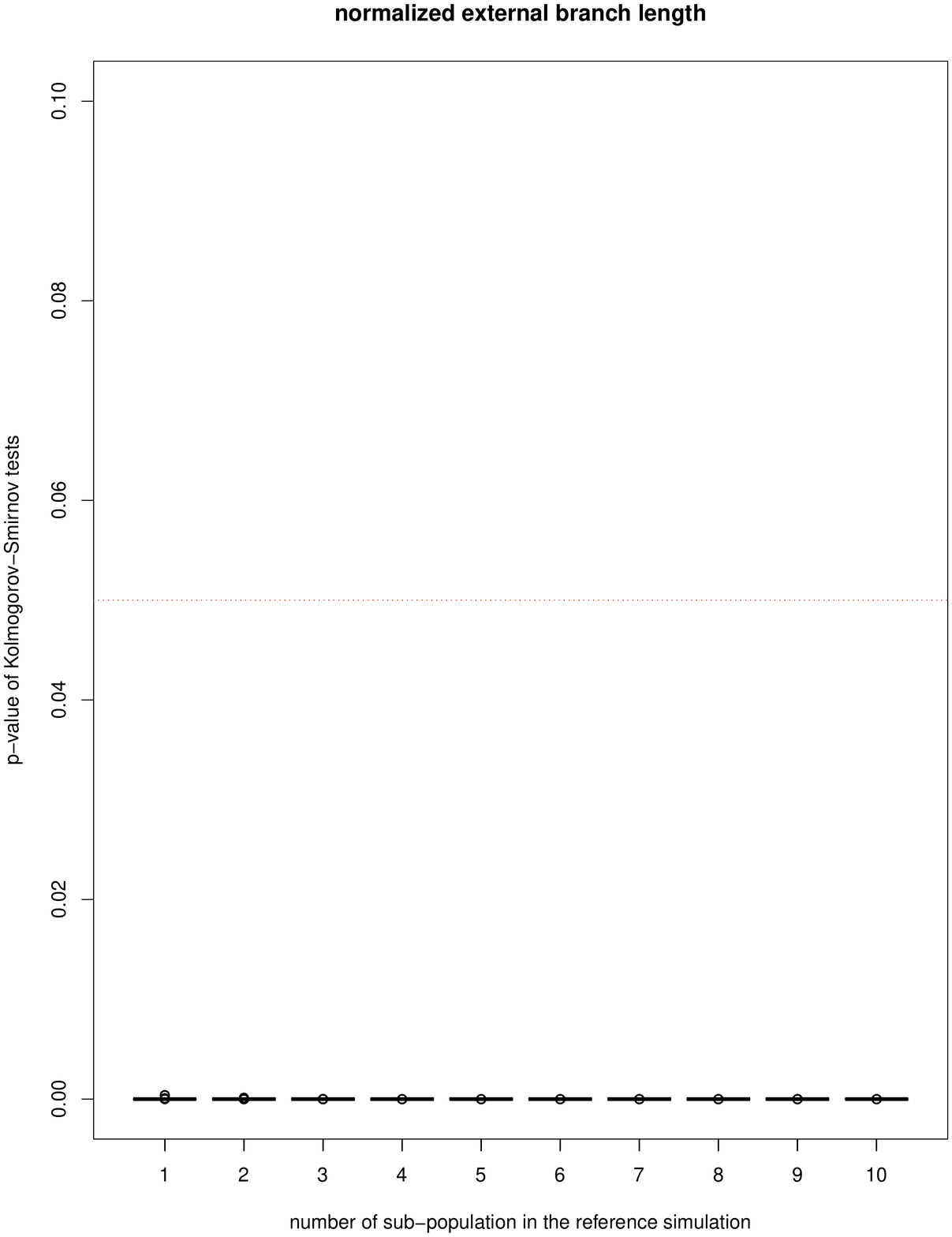}
&
\includegraphics[width=5cm,height=5cm, trim= 0 1cm 0cm 2cm, clip=true]{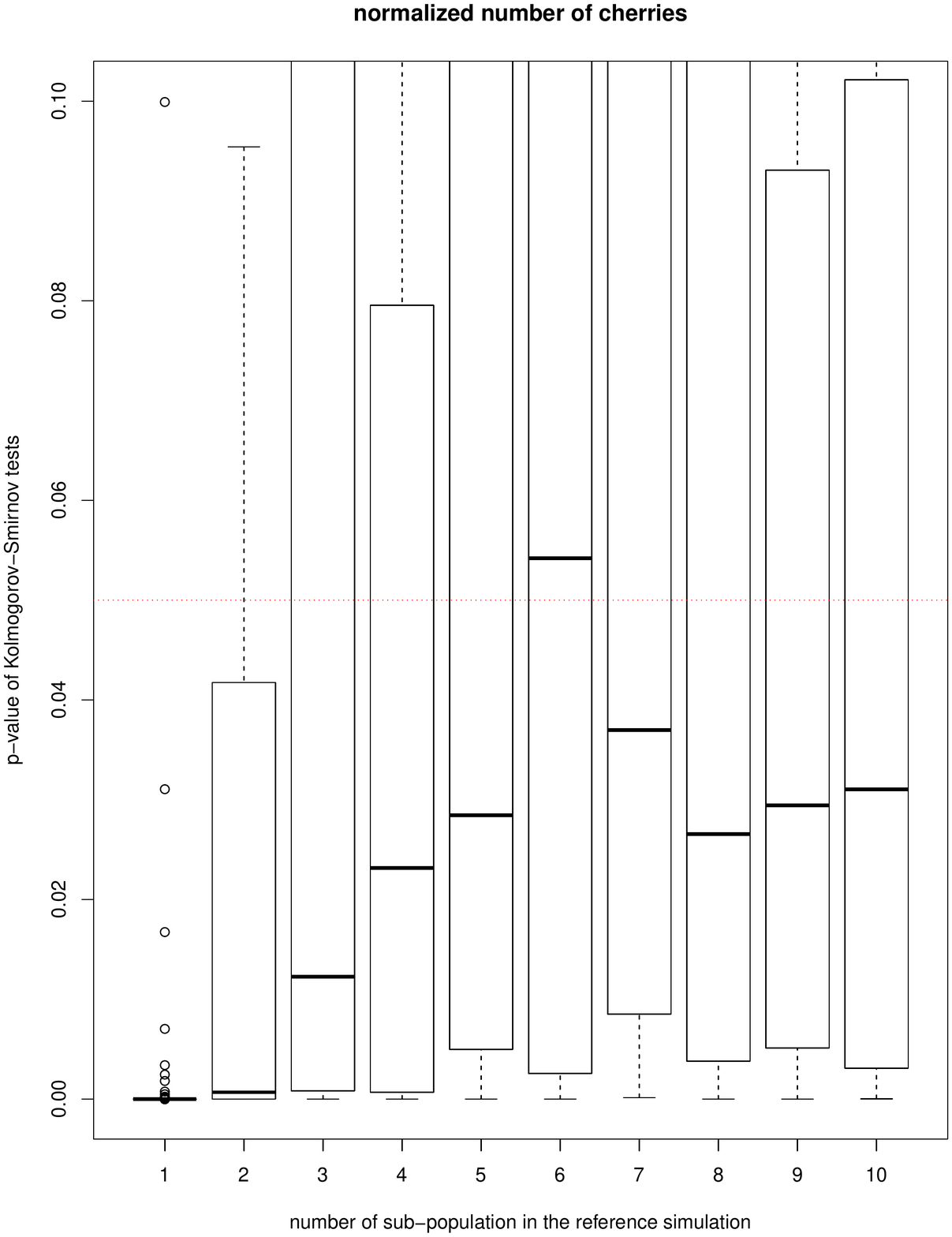}
&
\includegraphics[width=5cm,height=5cm, trim= 0 1cm 0cm 2cm, clip=true]{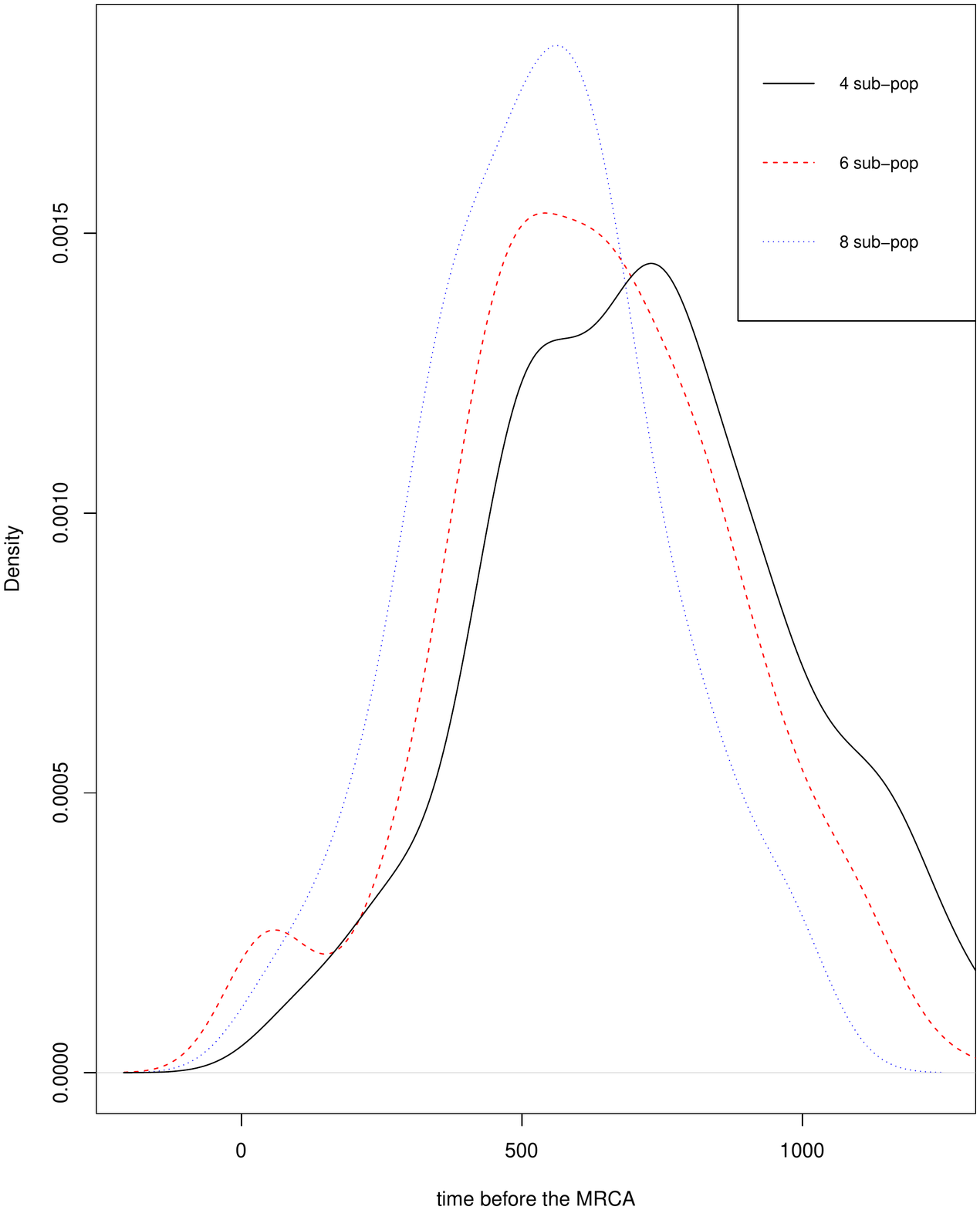} \\
(a) & (b) & (c)
\end{tabular}
\caption{\textit{{\small
(a): External branch length $L_n$: Box-plot of the p-values of the Kolmogorov-Smirnov test, for each value of the number of lineages $m$
at sampling time (in abscissa).
(b): Number of cherries $C_n$: Box-plot of the p-values of the Kolmogorov-Smirnov test, as a function of $m$.
For (a) and (b), $100$ ABC analysis were done for each value of $m$ and we tested if the distribution of the normalized external branch length follows a Gaussian distribution ($H_0$). The threshold value of rejection of $H_0$, $0.05$, is represented by the dashed red line. If the p-values are lower than this threshold, the distribution of the statistics ($L_n$ or $C_n$) of the forward-backward coalescent trees generated by a Doebeli-Dieckmann model is significantly different than the one under a Kingman's coalescent. (c): Compared distributions of the age of the MRCA for the forward-backward coalescent (plain line) and for the Kingman's coalescent (dotted line).}}}\label{Fig:neutral-test}
\end{center}
\end{figure}

Fig. \ref{Fig:Branch-cherries-lines3-6} shows further comparison between Kingman's coalescent and the trees under our model. The distribution of external branch lengths under our model follows an asymmetrical leptokurtic distribution and it tends to be much shorter than under a Kingman's coalescent. The time to the MRCA is also much longer under our model than the Kingman's coalescent. The distribution of the number of cherries follows a symmetrical bell-shaped distribution flattened around the mode.
\begin{figure}[ht!]
\begin{center}
\begin{tabular}{ccc}
\includegraphics[width=6cm, trim= 0 0cm 0cm 2cm, clip=true]{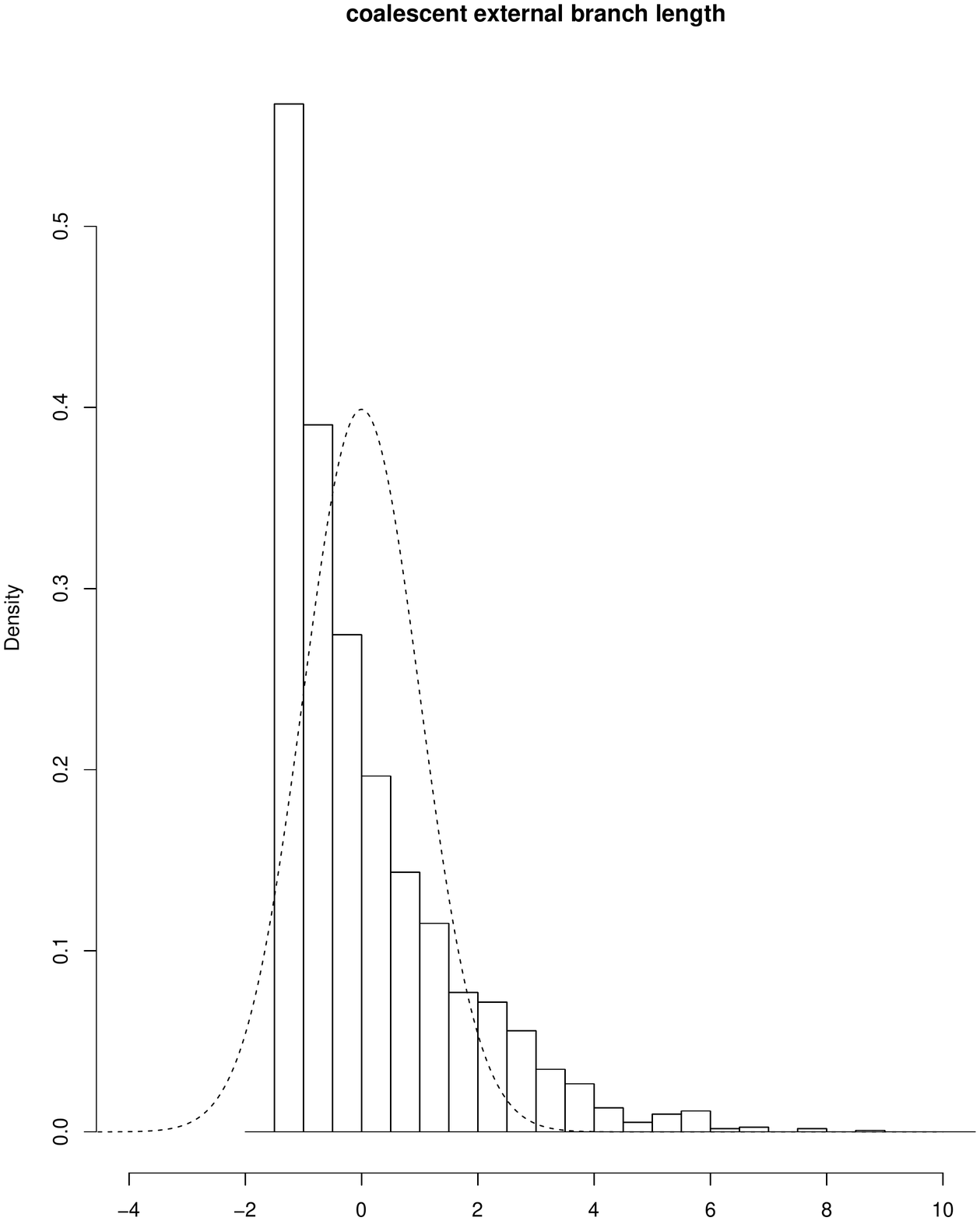}
&\includegraphics[width=6cm, trim= 0 0cm 0cm 2cm, clip=true]{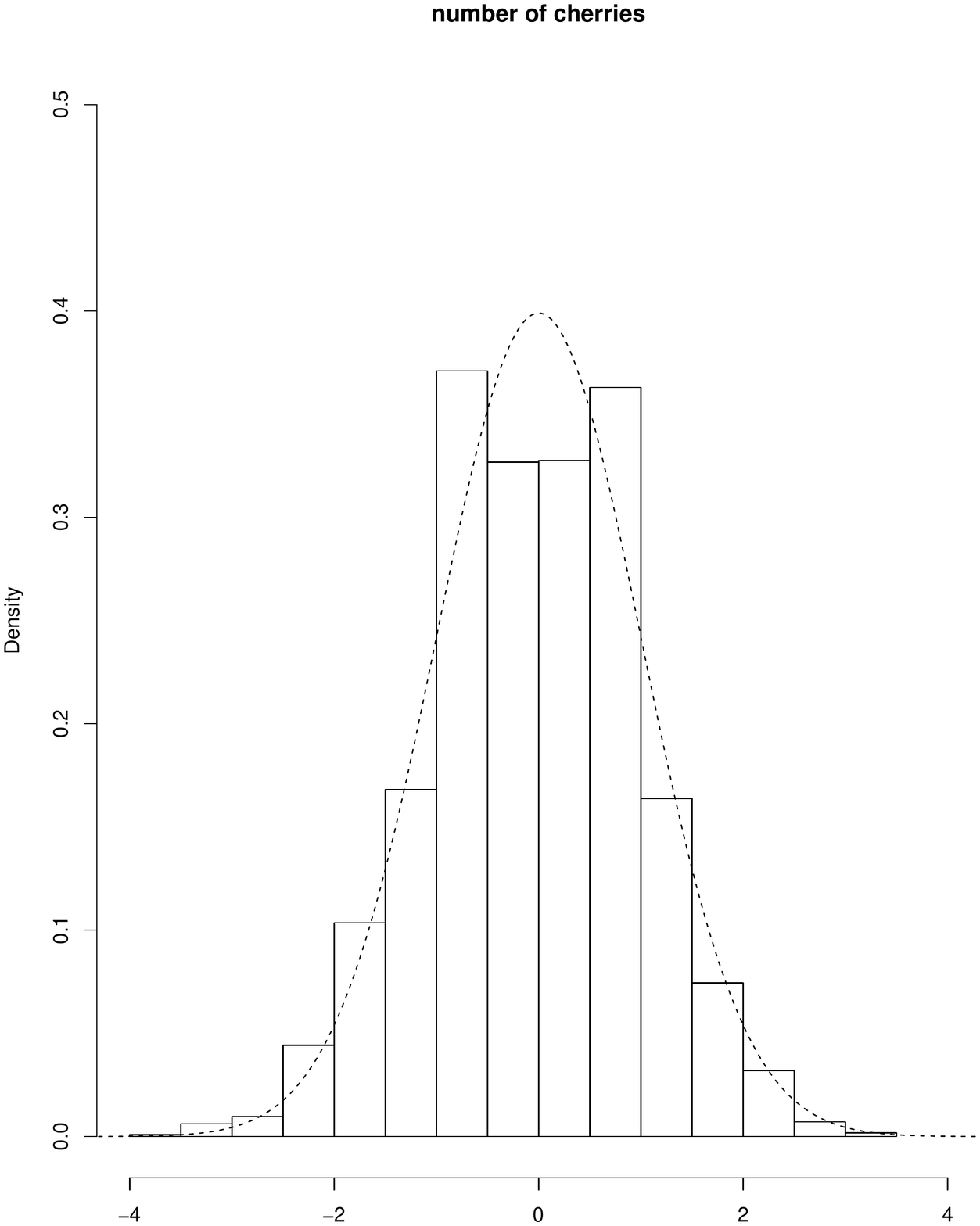}
& \includegraphics[width=6cm, trim= 0 0cm 0cm 2cm, clip=true]{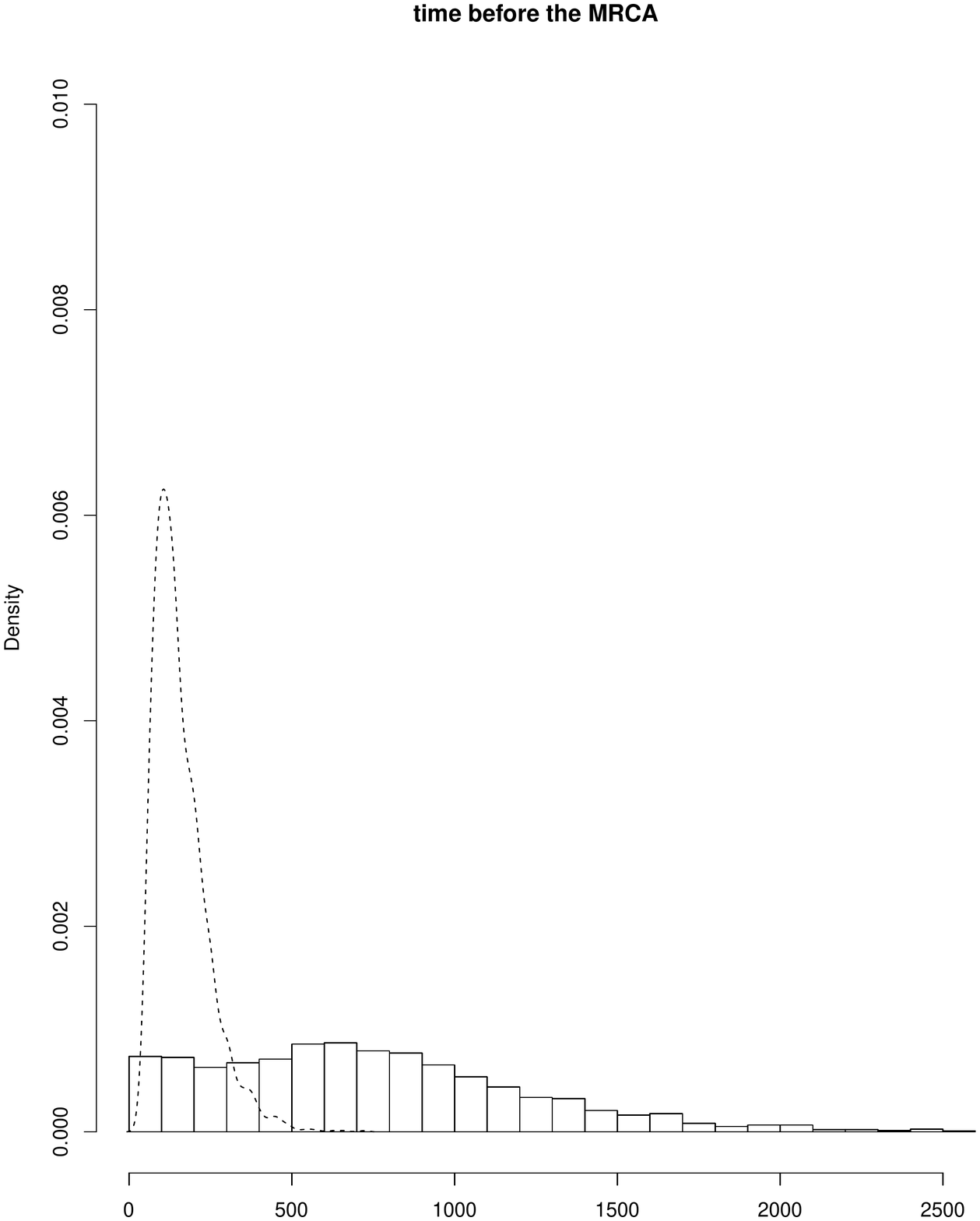}
\\
(a) & (b) & (c)
\end{tabular}
\caption{\textit{{\small Histograms of (a) the renormalized external branch lengths, (b) the renormalized number of cherries, (c) the time to the MRCA. The simulations are shown for $p=0.0076$, $q=0.7503$, $\sigma_b=1.186$, $\sigma_c=0.4951$, $\sigma_m=0.1448$, $\eta_c=0.0211$ and $t_{sim}=1025.619$ (set of parameter A in Table \ref{Table:data_parameters}. Results for three other `reference' sets are given in App. \ref{sec:neutrality-test}). The dashed line represents the distribution followed by a Kingman's coalescent ( Gaussian distribution for (a) and (b), simulations for (c)). }}}\label{Fig:Branch-cherries-lines3-6}
\end{center}
\end{figure}


Overall, we found that the coalescent trees generated by a Doebeli-Dieckmann model significantly differ from a Kingman's coalescent. In particular, we found that using a Kingman's coalescent model and ignoring the trait structure of a population tend to overestimate the recent coalescent times. The genealogies generated by the forward-backward coalescent under a Doebeli-Dieckmann's model are expected to differ from a standard or renormalized Kingman's coalescent for various reasons: i) there are multiple instantaneous coalescence events when a new lineage appears; ii) coalescence rates differ among lineages, creating asymmetries in the phylogenetic tree (trees can therefore be imbalanced); iii) coalescence rates vary in time since they depend on the structure of the population and the traits present at a given time; and iv) eco-evolutionary feedbacks and competitive interactions between lineages affect coalescent rates in the whole population.  \\

\subsection{Application 2: correlations between genetic and social structures in Central Asia}\label{sec:chaix}

In Anthropology, a common question is whether or not socio-cultural changes can affect demographic parameters, such as fertility rates. For instance, it is hypothesized that agriculturalists have a higher fertility than foragers  \cite[\eg][]{sellenandmace}, which is supported by several studies \cite[\eg][]{bentleyetal, rossetal}. In this section, we analyze genetic data in order to test whether populations with two different lifestyles and social organizations show different fertility rates. Nineteen human populations from Central Asia have been sampled in previous studies (Fig. \ref{Fig:carte}(a), \cite{chaixetal, heyeretal}). Two types of socio-cultural organizations are encountered: Indo-iranian populations are patrilineal, \text{i.e.} mostly pastoral and organized into descent groups (tribes, clans...); Turkic populations are cognatic, \text{i.e.} mostly sedentary farmers organized in nuclear families. 631 individuals have been sampled (310 from a cognatic population, 321 from a patrilineal one). Ten microsatellite loci have been genotyped on the Y-chromosome. Since there is no recombination on the sexual chromosomes in humans, it is appropriate to use our model which assumes clonal reproduction. Hence, we will perform ABC analysis on the genetic diversity following the paternal lineages.

\begin{figure}[!ht]
\begin{center}
\begin{tabular}{cc}
\includegraphics[width=8cm,height=6cm, trim=0 1cm 0 0, clip=true]{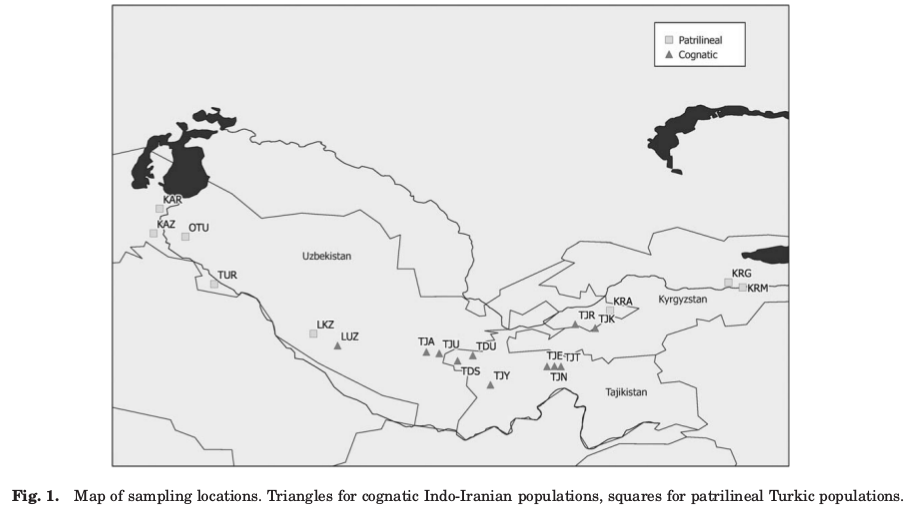}
 &
 \includegraphics[width=8cm,height=6cm]{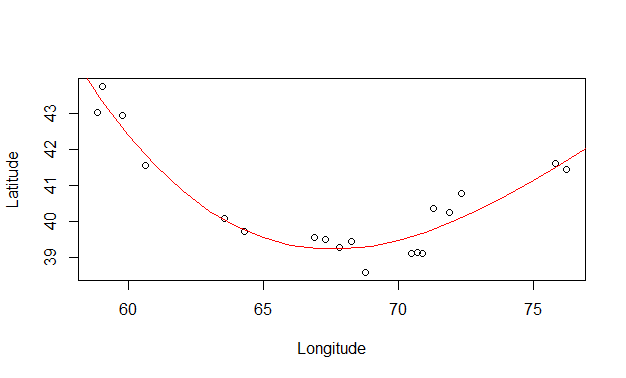}\\
 (a) & (b)
\end{tabular}
\caption{\textit{{\small
(a): Map of sampling locations from \cite{heyeretal}. Triangles correspond to cognatic Indo-Iranian populations, quares to patrilineal Turkic populations. (b): Regression of the data to a 1-dimensional problem.}}}\label{Fig:carte}
\end{center}
\end{figure}

We considered that the trait $x$ in the model is a vector containing the geographic location of the population and the social organization (cognatic or patrilineal). For geographical positions, given the Fig. \ref{Fig:carte}(a), we consider that geographic location is 1-dimensional: we can fit a polynomial curve through the geographical positions of the tribes:
\[P(x)=673.4-25.13 \ x+0.327\ x^2-1.39\ 10^{-3}\ x^3 \ (R^2= 0.92).\] Hence the location of each population is given by the coordinates $(x,P(x))$ (Fig. \ref{Fig:carte}(b)). The distance between populations is computed thanks to the line integral along the interpolated curve (see details in App. \ref{append:chaix-simul}). The neutral marker $u$ is a vector containing the genotype at the ten microsatellites. Here we assume that the neutral marker is fully linked with the trait corresponding to the social organization.

Our aim is to use our ABC procedure on the genetic data to estimate the parameters $\theta=(p_{\mbox{{\scriptsize xb01}}},b_0,b_1,p_{\mbox{{\scriptsize loc}}},q, \sigma_{\mbox{{\scriptsize loc}}},\eta_0,\eta_1,\sigma_c,t_{\sc{sim}})$ of our model. The individual birth rates is assumed to depend on social organization only and not on geographic location: $b_0$ for the patrilineal populations and $b_1$ for the cognatic ones. Death rates are supposed to be due to density-dependent competition for the sake of simplicity: the competitive effect of an individual located at coordinate $y$ on an individual in a patrilineal (resp. cognatic) population at location $y'$ is supposed $C(y,y')=\eta_0 \exp\big(-(y-y')^2/2\sigma_c^2\big)$ (resp. $C(y,y')=\eta_1 \exp\big(-(y-y')^2/2\sigma_c^2\big)$). The individual death rate at location $y$ is given by the sum of the competitive effects of all individuals. We supposed that individual can found new population after dispersal (corresponding to a mutation on the trait $x$ at birth), with probability $p_{\mbox{{\scriptsize loc}}}$, and/or change of social organization, with probability $p_{\mbox{{\scriptsize xb01}}}$. The location of the new population is randomly drawn in a centered Gaussian with standard deviation $\sigma_{\mbox{{\scriptsize loc}}}$. Following anthropological data, we assumed that social organization changes are unidirectional only from patrilineal pastoral to cognatic farmers populations \citep{chaixetal}. $t_{\sc{sim}}$ and $q$ respectively are the duration of the coalescent and the marker mutation probability.\\

Estimating the parameter $\theta$ and using the ABC procedure to select between alternative models will allow us to test whether the null hypothesis
\begin{equation}
H_0\ : \ b_0=b_1\label{test:H_0}
\end{equation} is acceptable, compared to the alternative hypothesis $H_a\ :\ b_0<b_1$ \cite[see \eg][]{grelaudrobertmarinrodolphetaly,pranglefearnheadcoxbiggsfrench,stoehrpudlocucala}. We generated a set of data with the \textit{a priori} probability $1/2$ of having $b_0=b_1$ and the \textit{a priori} probability $1/2$ of having $b_0<b_1$ (see details in App. \ref{append:chaix-simul}). {To do this, we generated 10,000 datasets with $b_0=b_1$ and 10,000 datasets with $b_0<b_1$. The ABC estimation provides weights $W_i$ for each of these 20,000 simulations (see Eq. \ref{def:posterior}) yielding the posterior distribution of the parameters (see  Fig. \ref{fig:chaix-posterior}). These weights $W_i$ also allow to compute the posterior probabilities of each hypothesis: $H_0\:\ \{b_0=b_1\}$ or $H_a\ :\ \{b_0<b_1\}$. When the estimated posterior probability for $\{b_0<b_1\}$ is larger than a certain threshold $\alpha$, the null hypothesis $H_0$ is rejected.} \\

We first checked the quality of the ABC estimation and of the test \eqref{test:H_0} on simulated data. {Among the 20,000 simulations presented in the above paragraph, we chose 200 simulations to play in turn the role of the true dataset, 100 among those with $b_0=b_1$ and 100 among those with $b_0<b_1$. }We obtained that parameters estimates were generally close to the true values (App. \ref{append:chaix-simul}). We then use these 200 datasets to perform 200 tests (using for each of them the 19,999 other simulations). Since we know for each of these 200 tests whether the data are obtain under $H_0\:\ \{b_0=b_1\}$ or $H_a\ :\ \{b_0<b_1\}$, this provides insight on the power of our test and allows us to set the threshold defining the critical region of the test. Here we can choose this threshold $\alpha=0.5$ which is very natural (see App. \ref{append:chaix-simul}). We can then conclude the test for the dataset from Central Asia populations .\\

For the ABC test, we obtained an estimated posterior probability for $\{b_0<b_1\}$ equal to 0.4518, below the threshold $\alpha=0.5$, so that the null hypothesis $H_0$ \eqref{test:H_0} can not be rejected. The p-value of the test, estimated as the proportion of these simulations where $\widehat{\P}(H_a\ |\ S_{\sc{obs}})\geq 0.4518$, can be estimated to 47\%. Hence there is no significantly higher fecundity in cognatic populations compared with patrilineal ones.



\begin{figure}[ht!]
\begin{center}
\includegraphics[width=18cm,height=10cm]{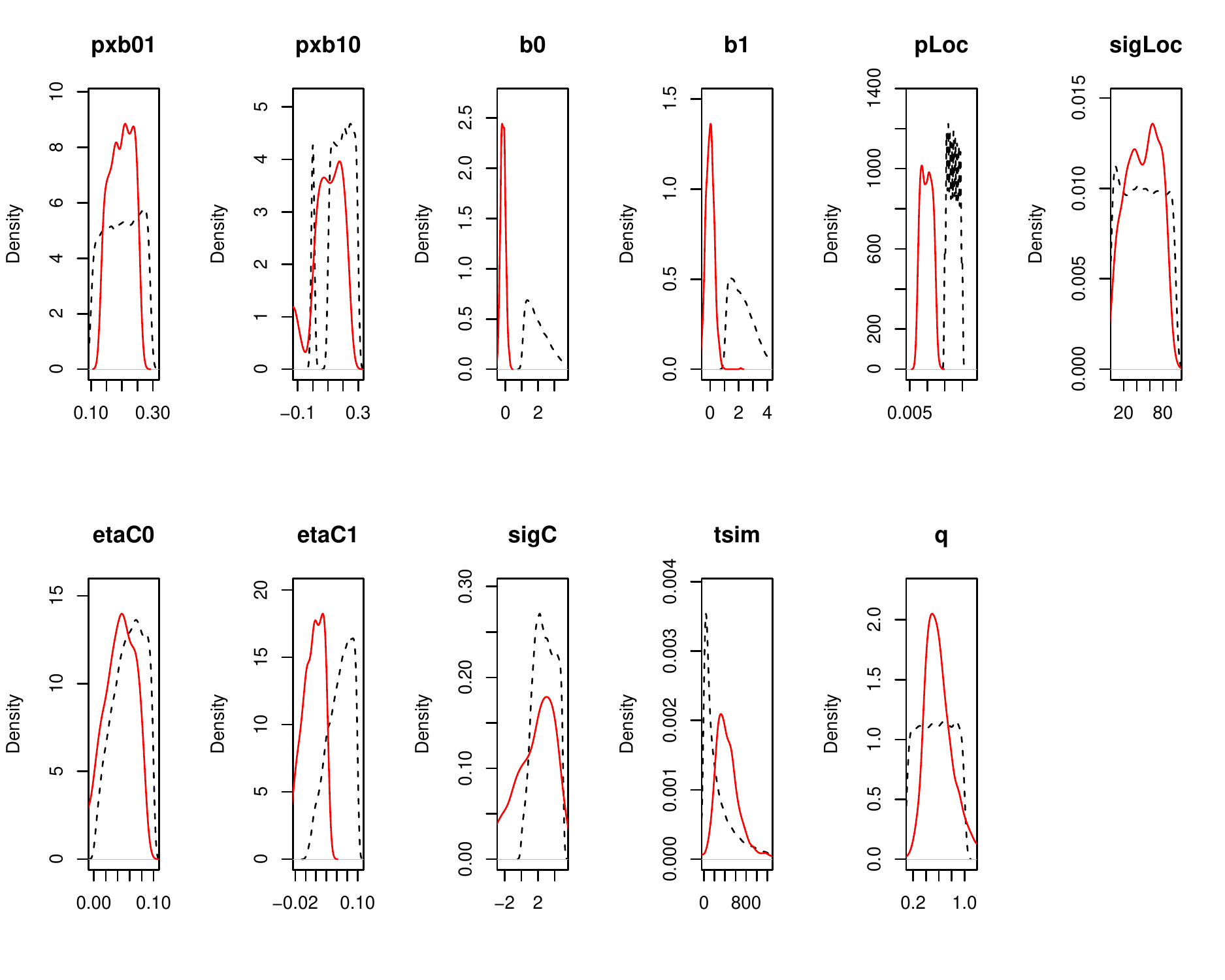}
\end{center}
\caption{{\small \textit{Results of the ABC estimation for the dataset of Heyer et al. \cite{heyeretal} for Central Asia human populations. The prior distributions are plotted in dashed lines and the posterior densities in plain red lines.}}}\label{fig:chaix-posterior}
\end{figure}

\begin{figure}[!ht]
\begin{center}
\begin{tabular}{cc}
\includegraphics[width=12cm,angle=0,height=9cm]{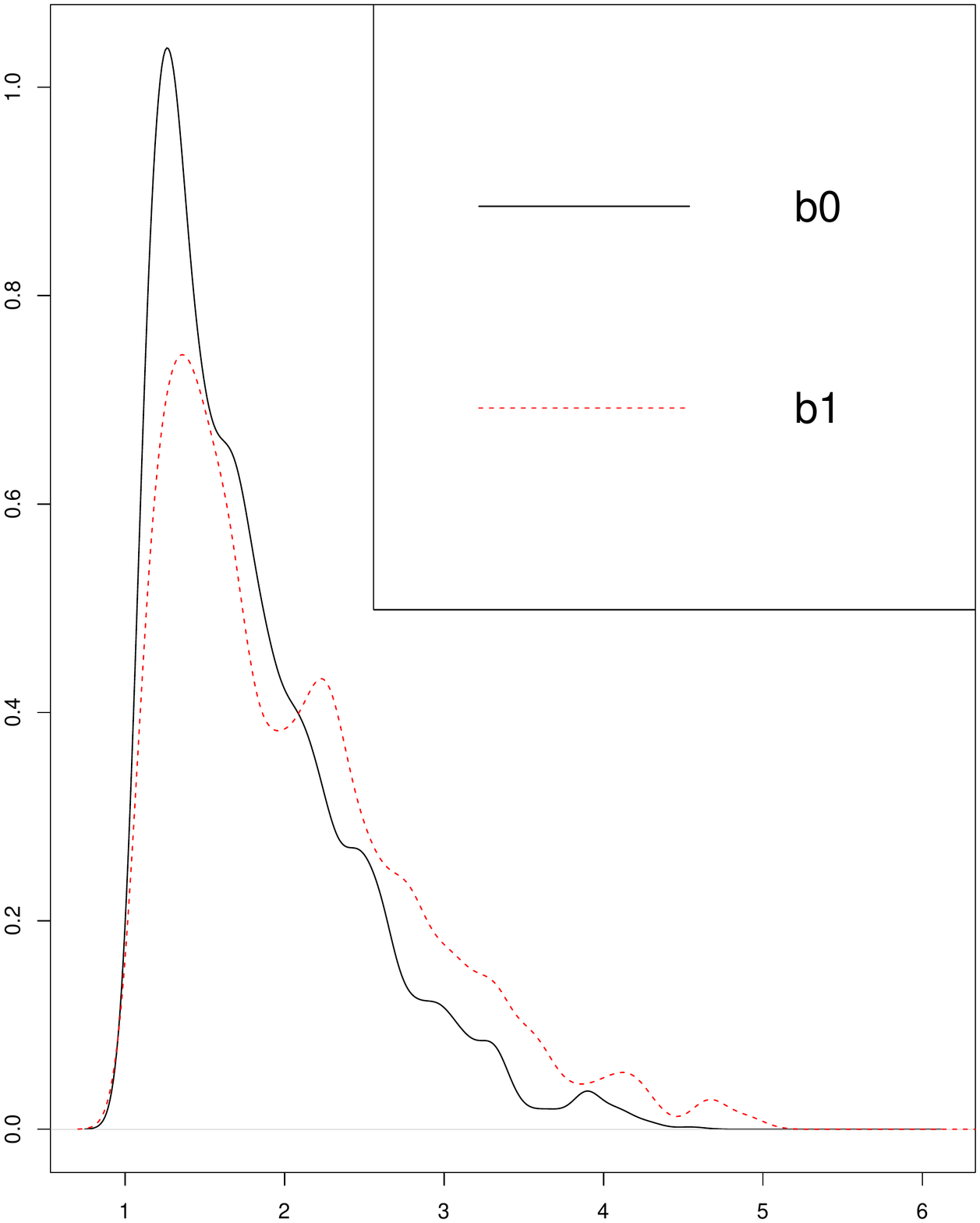} &
\includegraphics[width=12cm,angle=0,height=9cm]{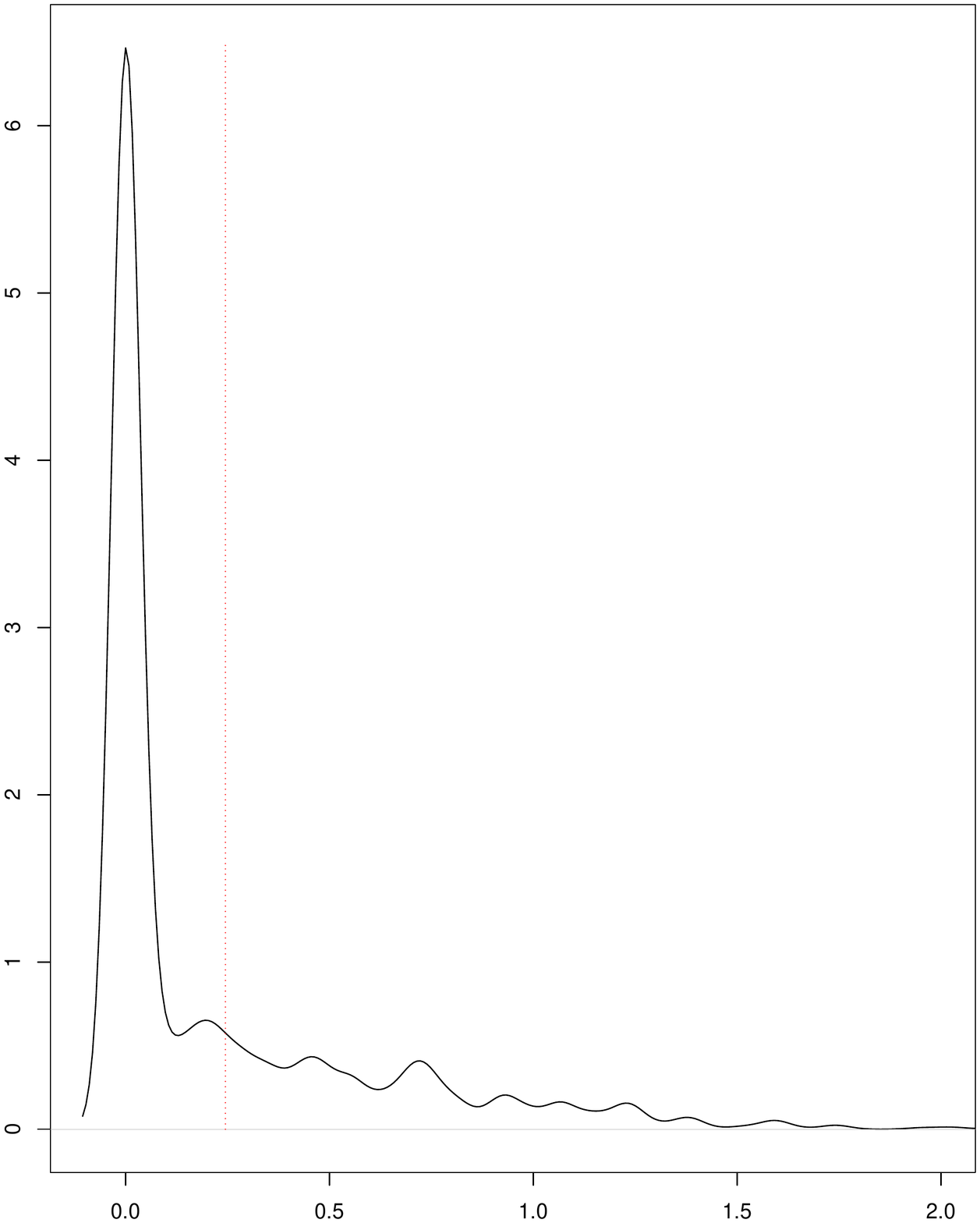}\\
(a) & (b)
 \end{tabular}
\caption{{\small \textit{(a): Approximate posterior distributions for $b_0$ and $b_1$, obtained by ABC on the Central Asian database with 40,000 simulations. (b): Approximate posterior distribution of $b_1-b_0$. The posterior mean of $b_1-b_0$, equal to 0.25, is indicated as the vertical dashed red line.}}}\label{fig:chaixb0b1}
\end{center}
\end{figure}

\clearpage

\section{Discussion}

Inferences from genetic data are most often performed under three important assumptions in the existing literature. First, the population size and structure are known parameters: either it is fixed or it follows a deterministic evolution, according to a given scenario (\eg expansion or bottleneck, or a fixed structure with known migration rates between sub-populations). Second, mutations are supposed to not affect the genealogical trees, \ie models are supposed neutral. Selection is rarely explicitly taken into account in inference methods (yet see for instance, background selection can bias the estimation of demographic variations  \citep{charlesworthetal03, johrietal2020}). Third, there is no feedback between the evolution of the population and its demography: a selected mutation is supposed not to affect population size, or the population structure. The most frequent models used in inference, the Kingman's coalescent and the Wright-Fisher model, make the three assumptions altogether. The goal of the present paper was to present a model and an inference method which allow to relax all these assumptions. We showed that by using an ABC procedure, it was possible to estimate ecological, demographic and genetic parameters from genotypic and phenotypic data. \\
Recently, \cite{rasmussenandstadler} proposed a birth-death model without interactions where mutations can affect the birth and death rates of individuals in a strain, which in return affect the genealogies. They showed how it was possible to use phylogenies to estimate the effect of mutations on fitness in some viruses. In our paper, we go a step further by allowing interactions between individuals, and population structure and demography that depend on the evolution of the population. Our model assumes two genetics traits, a selected trait which governs the structure of the population, and a marker linked to the trait which is neutral and used to infer the genealogy. We first showed how genetic diversity at the neutral marker is related to the evolution at the selected trait, and to the size and structure of the population. We then used this relationship by developing ABC procedure which allows to estimate ecological parameters based on genetic diversity at the neutral marker and on the partial or total knowledge of the population structure. We showed on simulated data that the ABC procedure gives accurate estimates of ecological parameters such as the birth, death and interactions rates, and genetic parameters such as the mutation rate. Our results also showed that non-neutral genealogies can easily be detected under our framework.\\
The ABC procedure is well fitted to deal with complex models if we can simulate the latter easily, which has become increasingly common for most ecological models \cite[\eg][]{legendreclobert,hallermesser}. Here, we applied our model and its ABC procedure to reanalyze the genetic diversity of microsatellites on Y chromosomes in Central Asia human populations. The genetic diversity are compared between two social organizations and lifestyles: patrilineal vs. cognatic. Previous studies showed a significantly different genetic diversity and coalescent trees topologies, which was interpreted as evidence of the effect of socio-cultural traits on biological reproduction, due to how wealth is transmitted within families \citep{chaixetal, heyeretal}. However, these conclusions were obtained under simplifying assumptions: genealogies followed a modified Wright-Fisher model, and the genetic diversity and coalescent trees topologies were compared independently, \ie there was no interaction between populations and between social organization. Such assumptions dismissed the possibility that socio-cultural traits and social organization could change, that new populations can be founded, and that competitive interactions between individuals within and between social organizations might affect demography and evolution. We relaxed all these limitations by applying our model. We supposed that the trait under selection can affect the birth rate. Contrarily to \cite{heyeretal}, we did not test whether wealth transmission could explain differences in genetic diversity and coalescent trees topologies. Rather, we addressed a long-standing question in anthropology: can fertility be affected by a change in a social organization, in particular with a change in the agricultural mode. We found no evidence of a fertility difference between both kinds of social organization. Our findings then ask the question why human populations can adopt new socio-cultural traits without any strong evidence of a biological advantage. Further analyses and data would be necessary to confirm our results, especially regarding the number of children per females. In the data, this information is based on a few interviews that are not at all precise  \citep[see Table S3 in the suppl. mat.,][]{chaixetal}. However, since the genetic diversity sampled in contemporaneous population is due to long historical process, it seems difficult to estimate fertility for several dozens or hundreds generations. Our results only suggest that there is, on average, no evidence of an effect of a social trait on fertility all along the history of Central Asia human populations. \\
Finally, our paper illustrates that it is actually possible to merge ecological and genetic data and models. Our model is based on classical competitive Lotka-Volterra equations, under the assumptions of rare mutations relatively to ecological processes. The genealogies and genetic diversity produced under such a model are then used to infer ecological and demographic parameters. We showed that relaxing strong assumptions of genetic models is possible, and that it allows to provide new analysis methods. Even though we applied our inferential procedure only to simulated genetic data or microsatellites genetic diversity, our model is general enough to embrace any type of data: SNPs, phenotypic traits, etc.  The development of stochastic birth and death models, with (this paper) or without \citep{rasmussenandstadler} interactions open the way to new methods for analyzing data. As highlighted by \cite{frostpybusgogviboudbonhoefferbedford}, this is particularly important for the study of epidemics and pathogens evolution. These authors give a list of current challenges which can partly addressed thanks to the method and models developed here. For instance, the role of the host structure on the pathogens evolution and genetic diversity, the role of stochasticity, and providing more complex and realistic evolutionary models.\\

\textit{Acknowledgments:} The authors thank
The authors thank Laurent S\'eries and Sylvain Ferrand for their help with the CMAP compute servers. They also thank Frédéric Austerlitz and Rapha\"elle Chaix for discussion and for sharing anthropological data from Central Asia. This research has been supported by the Chair ``Modélisation Mathématique et Biodiversité" of Veolia Environnement-Ecole Polytechnique-Museum National d'Histoire Naturelle-Fondation X. V.C.T. also acknowledges support from Labex CEMPI (ANR-11-LABX-0007-01) and B\'ezout (ANR-10-LABX-58). \\

\textit{Competing interests:} The authors declare no competing financial interests in relation to the current work.\\

\textit{Data archiving:} The genetic data, simulation and the programs developed in the paper will be archived on Dryad and Github.\\

{\footnotesize

\providecommand{\noopsort}[1]{}\providecommand{\noopsort}[1]{}\providecommand{\noopsort}[1]{}\providecommand{\noopsort}[1]{}

}

\clearpage

\appendix

\section{Mathematical construction of the PES and of the forward-backward phylogenies}\label{sec:maths}

Recall the stochastic individual-based model described in Section \ref{sec:micro}. We consider a population of clonal individuals, characterized, on the one hand, by a trait, \ie a vector of genetically determined variables $x \in \mathcal{X}\subset \R^d$, which affects the demographic processes such as birth, death and competitive interactions between individuals, and, on the other hand, by a vector of genetic markers $u \in \mathcal{U}\subset \R^q$, supposed neutral (\ie $u$ does not affect demographic processes). $\mathcal{X}$ and $\mathcal{U}$ respectively represent the sets of possible values of the trait and the neutral markers. The number of individuals in the population at time $t$ is denoted $N_t^K$, where $K$ is a scaling parameter which controls the relationship between the demographic and genetical parameters and variables. For instance, regarding competition, one can consider that each individual has a weight $1/K$ such that when $K$ increases, competition between individuals decreases and the population size increases.

The evolution of the population is a stochastic process in continuous time which depends on the rates of all possible demographic and genetic processes. Individuals with trait $x$ give birth at rate $b(x)$ and die at rate $d(x)+\frac{1}{K} C(x,y)$, where $d(x)$ is the intrinsic death rate, and $\frac{1}{K}C(x,y)$ is the additional death rate due to the competitive effect of a single individual with trait $y$. At birth, offspring' traits and markers can change by mutation with probability $p_K$ and $q_K$, respectively. The trait and marker mutation rates are respectively supposed such that $p_K=1/K^2$ and
\begin{equation}\label{hyp:tauxmutations}
q_K=p_K r_K,\quad \mbox{ with }\lim_{K\rightarrow +\infty}r_K= +\infty,\quad \mbox{ and } \lim_{K\rightarrow +\infty} q_K (\log K)^2=0.
\end{equation}
After mutation on the trait $x$, the offspring's trait is $x+h$ with $h$ randomly drawn in a distribution $m(x,h) dh$. The effect of mutations on the neutral marker are supposed to follow a Gaussian distribution with mean 0 and variance $1/K$ (but alternatives are possible, see \cite{billiardferrieremeleardtran}). {Notice that we do not need to assume small mutations for the trait.}\\

Because the population has variable size, it is conveniently represented at time $t$ by the following point measure:
\begin{equation}
\nu_t^K=\frac{1}{K}\sum_{i=1}^{N^K_t} \delta_{x_i,u_i},\label{def:nuK}
\end{equation}
where each individual is represented by a Dirac mass weighting its trait and marker values (individuals being ranked in the lexicographical order, for instance). \\

In this paper, we denote by $\mathcal{M}_F(E)$ the space of finite (non-negative) measures on $E$, and by $\mathcal{M}_1(E)$ the set of probability measures on $E$.
For a non-negative measurable or integrable function $f$ on $E$ and measure $\mu\in \mathcal{M}_F(E)$, we define $\langle \mu,f\rangle=\int_E f d\mu$.\\

In the sequel, we assume that
\begin{equation}\label{hyp:CI}
\sup_{K\in \N^*} \E\big(\langle \nu^K_0,1\rangle^3\big)<+\infty.
\end{equation}

\subsection{Large population limit in the ecological time-scale}

First, when $K\rightarrow +\infty$, the mutation rates vanish and we recover in the limit a system of ordinary differential equations.

\begin{prop}\label{prop:lotka-volterra}Let us assume that the initial condition $\nu_0^K$ converges, when $K\rightarrow +\infty$, to a trait-monomorphic initial condition of the form $\sum_{i=1}^p n_0(x_i) \delta_{x_i} \otimes \pi(x_i,du)$, where $p\in \N$, $x_1,\dots x_p\in \mathcal{X}$ and for all $i\in \{1,\dots p\},$ $\pi(x_i,du)$ is the marker-distribution conditional to the trait $x_i$. Then the marginal trait-distribution of $\nu^K$ converges to a limit of the form $\sum_{i=1}^p n_t(x_i) \delta_{x_i}$ where $(n_t(x_1),\dots n_t(x_p))$ are solution of the following system of ordinary equations \eqref{eq:lotka-volterra}. The convergence is uniform on every compact time intervals, and holds in probability.
\end{prop}

\begin{proof}The proof is a standard proof of tightness-uniqueness argument (see \cite{ethierkurtz} or \cite{fourniermeleard}).
\end{proof}

For $p=1$, we recover the classical logistic equation
$$\frac{dn_t(x_0)}{dt}= \big(b(x_0)-d(x_0)-\eta(x_0) C(x_0,x_0)n_t(x_0)\big)n_t(x_0)$$
whose solutions started at any non-zero initial condition all converge to the unique stationary stable solution, in case $b(x_0)-d(x_0)>0$:
$$\widehat{n}(x_0)=\frac{b(x_0)-d(x_0)}{\eta(x_0) C(x_0,x_0)}.$$
For $p=2$, we recover a competitive 2 species Lotka-Volterra system, whose solution converges to a stationary stable solution, that can correspond to the extinction of one or both species or coexistence.
For polymorphic populations with $p>2$, the dynamics becomes much more complicated (see \cite{zeeman} for $p=3$). However the following criterion from Champagnat et al. \cite{champagnatjabinraoul} ensures the convergence to a stationary stable point:

\begin{prop}[Champagnat Jabin Raoul]\label{prop:coexistenceLV}
Assume that for all $i,j\in \{1,\dots p\}$, $\eta(x_i)C(x_i,x_j)=\eta(x_j)C(x_j, x_i)$ and that the matrix $\big(\eta(x_i)C(x_i,x_j)\big)_{1\leq i,j\leq p}$ is positive definite. Then, starting from any initial condition $(n(x_i ; x_1,\dots x_p), i\in \{1,\dots p\})$ in the positive quadrant, the Lotka-Volterra system  \eqref{eq:lotka-volterra} converges to a unique stationary stable point. In the sequel, we denote this equilibrium $\big(\widehat{n}(x_1; x_1,\dots x_p),\dots \widehat{n}(x_p; x_1,\dots x_p)\big)$.
\end{prop}

Notice that the conditions of Prop. \ref{prop:coexistenceLV} are satisfied for instance when $\eta(x)\equiv \eta$ is a constant function and when $C(x,y)=C(x-y)$ is symmetric with positive Fourier transform. This is in particular the case for Gaussian kernels $C(x,y)=\frac{1}{\sqrt{2\pi}\sigma}\exp\big(-\frac{(x-y)^2}{2\sigma^2}\big)$. When the competition kernel is symmetric with positive Fourier transform, the equilibria of \eqref{eq:lotka-volterra} do not depend any more on the initial conditions.

\subsection{Substitution Fleming-Viot limit in the trait-mutation time-scale}

Now let us return to the microscopic population \eqref{def:nuK}. There are three timescales underlying its dynamics: 1) the ecological timescale of births and deaths is the more rapid, the global birth and death rate being of order $K$;  2) the timescale of marker mutations is the intermediate timescale, when the population size is of order $K$, marker mutations appear at a global rate $r_K /K$; 3) the trait mutations happen on the slower timescale; for a population size of order $K$, the global trait mutation rate is of order $1/K$. \\

We now consider the dynamics in the timescale of trait mutations and consider $(\nu^K_{Kt})_{t\geq 0}$. The dynamics of this process, when $K\rightarrow +\infty$, has been studied in Billiard et al. \cite{billiardferrieremeleardtran}. It is shown that the sequence $(\nu^K_{K.})_{K\in \N^*}$ converges to a limit called Substitution Fleming-Viot Process (SFVP) and described as follows. The trait distribution evolves as the Polymorphic Evolution Sequence (PES) introduced by Champagnat and M\'el\'eard \cite{champagnatmeleard2011}. Between trait mutations, the population stabilizes at the equilibrium of the ODE system \eqref{eq:lotka-volterra} corresponding to the traits which are present in the population. Transitions, whose durations are of order $\log K$ disappear in the limit. The limiting trait distribution thus jumps from one (possible polymorphic) equilibrium to another one when successful trait mutations arise. Under our assumptions on the mutation probabilities $p_K$ and $q_K$, when a new mutant trait appears and invades into the population, the neutral marker that is linked to it benefits from a hitchhiking phenomenon (see e.g. \cite{barton98,barton,durrettschweinsberg2004,durrettschweinsberg,etheridgepfaffelhuberwakolbinger}). Between jumps of the trait distribution, the neutral distribution follows in the limit a diffusive Fleming-Viot process, that boils down to a Wright-Fisher diffusion with mutations in the case where $\mathcal{U}=\{a,A\}$ has only two elements.
 \par To give a more precise description, let us define the different ingredients appearing in the expression of the SFVP. In a trait polymorphic population at equilibrium, of the form $\sum_{i=1}^p \widehat{n}(x_i ; x_1,\dots x_p)\delta_{x_i}$, the fitness function of a new small mutant population of trait $y$ is defined by \eqref{def:fitness}. This fitness function appears in the PES process introduced by Champagnat and M\'el\'eard \cite{champagnatmeleard2011}, defined as follows.

\begin{definition}\label{def:PES}
Let us work under Assumptions of Proposition \ref{prop:coexistenceLV}. The PES process $\Lambda_.(dy)$ is a pure-jump process with values in $\mathcal{M}_F$, the set of point measures on $\mathcal{X}$. It jumps from
$$ \sum_{i=1}^p \widehat{n}(x_i ; x_1,\dots x_p)\delta_{x_i}\mbox{ to }  \sum_{i=1}^p n^*(x_i ; x_1,\dots x_{p+1})\delta_{x_i} + n^*(x_{p+1} ; x_1,\dots x_{p+1})$$with rate
$$\sum_{j=1}^p b(x_j)\widehat{n}(x_j;x_1,\dots x_p) \frac{\big[f(x_{p+1} ; x_1,\dots x_p)\big]_+}{b(x_{p+1})} m(x_j,x_{p+1}-x_j) dx_{p+1}$$
where the fitness function $f$ has been defined in \eqref{def:fitness}.
\end{definition}
By assumptions of Proposition \ref{prop:coexistenceLV}, the sizes $n^*(x_i ; x_1,\dots x_{p+1})$ are well defined.

\noindent The term $b(x_j)\widehat{n}(x_j;x_1,\dots x_p) m(x_j,x_{p+1}-x_j)$ in the jump rate appearing in the definition says that every subpopulation of trait $x_j$ can generate the mutant $x_{p+1}$, with a probability that depend on the size and birth rate of this subpopulation, and of the mutation kernel. The term $\big[f(x_{p+1} ; x_1,\dots x_p)\big]_+/b(x_{p+1})$ describes the probability that a mutant trait $x_{p+1}$ is not wiped out by the stochasticity ruling the births, deaths and competition events.\\

\me Finally, let us define the Fleming-Viot process (see \cite{dawson,dawsonhochberg,donnellykurtz_96} or \cite{etheridgebook}) that appears in our study.
\begin{definition}\label{def:FV} Let us fix $x\in {\cal X}$, $u \in {\cal U}$ and consider a polymorphic population with traits $x_1,\dots x_p \in \mathcal{X}$ that is described by the trait-distribution $\Lambda$. Let us assume that for every continuous bounded test function $\phi$ on $\mathcal{U}$,
\begin{equation}
   \label{hyp-D}
\lim_{K\rightarrow +\infty}\sup_{u\in \U} \bigg | {r_K\over K} \int_{ \U} (\phi(u+h)-\phi(u)) G_{K}(u,dh) - A\phi(u)\bigg | =0,
  \end{equation}
where  $(A, {\cal D}(A))$ is  the generator of a Feller semigroup and $\phi\in {\cal D}(A)\subseteq \mathcal{C}_b(\U,\R)$, the set of continuous bounded real functions on $\U$. \\
The Fleming-Viot process $(F_{t}^{u,\Lambda}(x,.), t\geq 0)$ indexed by $x$, started at time $0$ with initial condition $\delta_{u}$, associated with the mutation operator $A$ and evolving in a population whose trait-marginal distribution is $\Lambda(dx)$, is the  ${\cal P}({\cal U})$-valued process whose law is characterized as the unique solution of the following martingale problem. For any $\phi \in {\cal D}(A)$,
\begin{equation}
M^x_{t}(\phi)
= \langle F^{u,\Lambda}_{t}(x,.),\phi\rangle -  \phi(u) - b(x) \int_{0}^t  \langle F^{u,\Lambda}_{s}(x,.),  A\phi\rangle ds\label{PBM-FV}
\end{equation}
is a continuous square integrable  martingale with quadratic variation process
\begin{align}
\langle M^x(\phi)\rangle_{t}=  &\frac{b(x)+d(x)+\eta(x) \int_{\mathcal{X}} C(x,y) \Lambda(dy) } {\langle \Lambda, 1\rangle}
\int_{0}^t  \left(\langle F^{u,\Lambda}_{s}(x,.), \phi^2\rangle -  \langle F^{u,\Lambda}_{s}(x,.), \phi\rangle^2 \right)ds.\label{crochet-PMB-FV}
\end{align}
\end{definition}

\noindent Now we can enonciate the convergence result established in \cite{billiardferrieremeleardtran}:

\begin{thm}[Billiard, Ferri\`ere, M\'el\'eard and Tran]\label{th:SFVP}
Let us assume that \eqref{hyp:CI}, \eqref{hyp-D} and Hypothesis of  Proposition \ref{prop:coexistenceLV} hold.

The sequence $(\nu^K_{K.})_{K\in \N^*}$ converges in distribution to the superprocess $(V_t(dx,du))_{t\in \R_+}$ such that:\\
(i) $\forall t\in \R_+,\ V_t(dx,\mathcal{U})=\Lambda_t(dx)$, the PES process of Definition \ref{def:PES}. We write $V_t(dx,du)=\Lambda_t(dx)\pi_t(x,du)$, where $\pi_t(x,du)$ is the conditional probability distribution of the marker conditionally to the trait $x$.\\
(ii) At time $t$, the process $V$ jumps from a state $V_{t_-}(dx,du)=\Lambda_{t_-}(dx)\pi_{t_-}(x,du)$ with $\Lambda_{t_-}(dx)=\sum_{i=1}^{\card(\supp(\Lambda_{t_-}))} \widehat{n}\big(x_i ; \supp(\Lambda_{t_-})\big)\delta_{x_i}$ to
$$\sum_{i=1}^{\card(\supp(\Lambda_{t_-}))} n^*\big(x_i ; \supp(\Lambda_{t_-}) \cup \{y\}\big) \delta_{x_i} \pi_{t_-}(x_i,du) + n^*\big( y ; \supp(\Lambda_{t_-}) \cup \{y\}\big) \delta_{(y,v)}$$
with rate
$$\int_{\mathcal{X}\times \mathcal{U}} \Big[ b(x) \int_{\mathcal{X}}\frac{\big[f\big(y ; \supp(\Lambda_{t_-})\big)\big]_+}{b(y)} m(x,y-x) dy  \Big] \Lambda_{t_-}(dx)\pi_{t_-}(x,dv).$$
Let us denote by $\tau_0=0$ and $\tau_1,\tau_2,\dots$ the successive jump times of $V_.$ (which are also the jump times of the PES $\Lambda$).\\
(iii) Between jumps, and conditionally to the trait-distribution, the marker-probability distributions $\pi_t(x,du)$ evolve as independent Fleming-Viot superprocesses, characterized by the following martingale problem: for $t\geq \tau_k$ ($k\geq 0$), for $x\in \supp(\Lambda_t)$ and for $\phi(u)\in \mathcal{D}(A)$:
\begin{equation}
M^x_{t\wedge \tau_{k+1}}(\phi)
= \int_{\mathcal{U}} \phi(u) \pi_{t\wedge \tau_{k+1}}(x,du) -   \int_{\mathcal{U}} \phi(u) \pi_{\tau_k}(x,du) - b(x) \int_{\tau_k}^{t\wedge \tau_{k+1}}   A\phi(u) \pi_s(x,du)\  ds\label{PBM-FV-SFVP}
\end{equation}is a square integrable martingale with quadratic variation process
\begin{multline}
\langle M^x(\phi)\rangle_{t\wedge \tau_{k+1}} =  \frac{b(x)+d(x)+\eta(x) \int_{\mathcal{X}}C(x,y)\Lambda_s(dy)}{\widehat{n}\big(x ; \supp(\Lambda_{\tau_k})\big)} \\
\int_{\tau_k}^{t\wedge \tau_{k+1}} \Big(\int_{\mathcal{U}} \phi^2(u) \pi_s(x,du) -\big(\int_{\mathcal{U}} \phi(u) \pi_s(x,du)\big)^2\Big)  ds.
\end{multline}
The convergence holds in the sense of finite dimensional distributions and in the sense of occupation measures.
\end{thm}

Point (ii) tells us that when a new mutant trait appears, there is a hitchhiking phenomenon and the marker that is physically linked with this trait can invade the population (see \cite{billiardferrieremeleardtran}). With a similar proof, it can be proved that when there is coexistence of other traits, the marker distributions of the sub-populations corresponding to the traits that coexist after the invasion of the new mutant trait are the same than the marker distributions before the jump.\\

Heuristically, we can think of the trait $x\in \R^d$ as characterizing the species to which the individuals belong. Put as this way, the model defined below is a speciation extinction model and we are interested in the phylogenies of a neutral marker in this framework.

\subsection{Phylogeny model for the SFVP: the PES-based phylogenies}\label{sec:phylo}

In the SFVP limit, the subpopulation of trait $x$ living at time $t$ say, can be considered of constant ``size'' $\widehat{n}(x ; \supp(\Lambda_t))$, with birth and death rates $b(x)=d(x)+\int_{\mathcal{U}}  C(x,y)\Lambda_t(dy)$. The competition term implies that the death rate in the species $x$ - or alternatively the subpopulation size $\widehat{n}(x ; \supp(\Lambda_t)) $ - depends on the composition of the whole population: each time a new species appears or each time a species gets extinct, the sizes of the species populations change. \\
Between mutations, the marker distribution evolves as a Fleming-Viot superprocess and thus we can expect that the phylogenies of individuals sampled from the species $x$ at time $t$ are distributed as a Kingman coalescent with rate $b(x)$ and efficient size $\widehat{n}(x ; \Lambda_t)$ (that depends a priori on the complete set of traits in the population), with a bottleneck at the time $\tau_x$ of appearance of the mutant trait $x$.\\

Let us first recall some results on the phylogenies of Fleming-Viot processes. The links between the processes forward and backward in time requires the notion of duality, and we refer to \cite{etheridgebook,jansenkurt} for detailed presentations. Two Markov processes $(X_t)_{t\in \R_+}$ and $(Y_t)_{t\in \R_+}$ with values in $E$ and $F$ respectively are in duality with respect to $f\in \Co(E\times F,\R)$, if for every $x_0\in E$, $y_0\in F$, $t\in \R_+$,
$\E_{x_0}\big(f(X_t,y_0)\big)=\E_{y_0}\big(f(x_0,Y_t)\big)$.
Taking the derivative, this implies that for all $x\in E$ and $y\in F$, $\mathcal{L}_X f(.,y)(x)=\mathcal{L}_Yf(x,.)(y)$, where $\mathcal{L}_X$ and $\mathcal{L}_Y$ are the generators of $X$ and $Y$. Applying duality to the Fleming-Viot process, we can obtain the distribution of the phylogenies of a sample on individuals chosen uniformly in the population at time $t$.

\begin{lem}Let $A$ be a generator and $n\in \N^*$. Let $\zeta_0\in \Co(\mathcal{U}^n,\R)$ be a function such that for every $i\in \{1,\dots, n\}$, and for every $(u_1,\dots u_{i-1}, u_{i+1},\dots u_n)\in \mathcal{U}^{n-1}$, $u_i\mapsto \zeta_0(u_1,\dots u_n)$ belongs to $\mathcal{D}(A)$. Let $(U_t^1,\dots U_t^n)_{t\in \R_+}$ be a $\mathcal{U}^n$-valued process whose generator $G$ is defined for a function as $\zeta_0$ by:
\begin{align}
\mathcal{G}_U \zeta_0(u_1,\dots u_n)= & \frac{1}{2}\sum_{{\scriptsize \begin{array}{c}i\not= j\\i\in \{1,\dots n\}\\
j\in \{1,\dots,n\}\end{array}}} \Big(\zeta_0\big(u_1,\dots u_{j-1},u_i,u_{j+1},\dots u_n\big)-\zeta_0\big(u_1,\dots u_n\big)\Big)\nonumber\\
+  & \sum_{i=1}^n A\big(x_i\mapsto \zeta_0(x_1,\dots x_n)\big)(x_i).\label{def:G}
\end{align}
The genealogies of the process $(U^1_t, \dots U^n_t)_{t\geq 0}$ is a Kingman coalescent with parameter 1.
\end{lem}

\begin{proof}The generator $G$ says that
\begin{itemize}
\item for all $i,j\in \{1,\dots,n\}$, the particle $j$ is replaced by a particle at the same state as the particle $i$ with rate $1/2$.
\item between jumps, each component $U^i_t$ of the vector process evolves independently following the generator $A$.
\end{itemize}
Then, that the genealogies of the process $(U^1_t, \dots U^n_t)_{t\geq 0}$ is a Kingman coalescent is straightforward, if we recall that a Kingman coalescent on $\{1,\dots n\}$ is a process on the set of partitions of $\{1,\dots n\}$ where two different elements of the partition merge after an independent exponential time of parameter 1.
\end{proof}

\begin{prop}Let us consider the Fleming-Viot process $(\pi_t(du))_{t\geq 0}$ started from $\pi_0$ and associated with the generator $A$ (we omit here the trait parameters $x$ and the support $\Lambda$ of the trait distribution) that is defined in Th. \ref{th:SFVP} (iii). Let $n\in \N^*$ be the number of individuals drawn independently in $\pi_t(du)$ at time $t>0$. Let $\zeta_0\in \Co(\mathcal{U}^n,\R)$ such that for every $i\in \{1,\dots, n\}$, and for every $(u_1,\dots u_{i-1}, u_{i+1},\dots u_n)\in \mathcal{U}^{n-1}$, $u_i\mapsto \zeta_0(u_1,\dots u_n)$ belongs to $\mathcal{D}(A)$. Let $(U_t^1,\dots U_t^n)_{t\in \R_+}$ be a $\mathcal{U}^n$-valued process with generator $G$ defined in \eqref{def:G}.
Then,
\begin{align}
\E_{\pi_0}\big(\langle \pi_t^{\otimes n},\zeta_0\rangle\big)=\int\dots \int \E_{(u_1,\dots u_n)}\Big(\zeta_0\big(U^1_t,\dots U^n_t\big)\Big) \pi_0(du_1)\dots \pi_0(du_n).\label{identif:loisphylo}
\end{align}
\end{prop}

The relation \ref{identif:loisphylo} shows that the distribution of $n$ individuals sampled independently from $\pi_t$ at time $t$ is the same as the distribution of the diffusion $(U_t^1,\dots U^n_t)$ started at independent positions drawn in $\pi_0$.

\begin{proof}First, let us consider the distribution of $n$ individuals sampled uniformly in $\pi_t(x,du)$, for $x\in \mathcal{X}$. To use the duality techniques, we compute $\E_{\pi_0}\big(F(\zeta_0,\pi_t(x,.)^{\otimes n})\big)$ for $\zeta_0\in \Co(\mathcal{U}^n,\R_+)$ and for the test function $F\in \Co\big(\Co(\mathcal{U}^n)\times \mathcal{M}_1(\mathcal{U}^n),\R\big)$ such that $$\forall \zeta_0\in \Co(\mathcal{U}^n,\R), \forall \mu\in \mathcal{M}_1(\mathcal{U}^n), F(\zeta_0,\mu)=\int_{\mathcal{U}}\dots \int_{\mathcal{U}} \zeta_0(u_1,\dots u_n) d\mu(u_1,\dots u_n).$$
For $\zeta_0(u_1,\dots u_n)=\prod_{i=1}^n \phi_i(u_i)$ where $\phi_1,\dots \phi_n\in \Co(\mathcal{U},\R)$ bounded, we have for any measure $\pi\in \mathcal{M}_1(\mathcal{U})$,
$F(\zeta_0,\pi^{\otimes n})=\prod_{i=1}^n \langle \pi,\phi_i\rangle.$
Then, by Itô's formula for such a function $\zeta_0$,
\begin{align*}
F\big(\zeta_0,\pi_t^{\otimes n}(x,.)\big)= & F\big(\zeta_0,\pi_0^{\otimes}(x,.)\big) +M_t^F
+  \sum_{i=1}^n \int_0^t \Big(\prod_{j\not= i} \langle \pi_s(x,.),\phi_j\rangle\Big) \int_{\mathcal{U}} A\phi_i(u)\pi_s(x,du)\ ds\\
+ & \frac{1}{2} \sum_{i\not= j } \int_0^t \Big(\langle \pi_s(x,.),\phi_i \phi_j\rangle - \langle \pi_s(x,.),\phi_i\rangle \langle \pi_s(x,.),\phi_j\rangle\Big) \Big(\prod_{k\notin \{i,j\}} \langle \pi_s(x,.),\phi_k\rangle \Big) \ ds\\
= & F\big(\zeta_0,\pi_0^{\otimes n}(x,.)\big)+M_t^F
+  \sum_{i=1}^n \int_0^t F\big(A_i \zeta_0,\pi_s^{\otimes n}(x,du)\big)ds\\
+ & \frac{1}{2} \sum_{i\not=j} \int_0^t \Big(F\big(\Phi_{ij}\zeta_0,\pi_s^{\otimes n}(x,du)\big)-F\big(\zeta_0,\pi_s^{\otimes n}(x,du)\big)\Big) ds
\end{align*}where $(M_t^F)_{t\geq 0}$ is a square integrable martingale, where
 $A_i \zeta_0 (u_1,\dots u_n)= A\big(u_i\mapsto \zeta_0 (u_1,\dots u_n)\big)(u_i)$,
 and where $\Phi_{ij}\zeta_0(u_1,\dots u_n)=\zeta_0(u_1,\dots ,u_{j-1},u_i,u_{j+1},\dots u_n)$. Taking the expectation:
 \begin{align*}
\E_{\pi_0(x,.)}\Big(F\big(\zeta_0,\pi_t^{\otimes n}(x,.)\big)\Big)= & F\big(\zeta_0,\pi_0^{\otimes n}(x,.)\big)
+  \sum_{i=1}^n \int_0^t \E_{\pi_0(x,du)} \Big(F\big(A_i \zeta_0,\pi_s^{\otimes n}(x,du)\big)\Big) ds\\
+ & \frac{1}{2} \sum_{i\not=j} \int_0^t \E_{\pi_0(x,du)}\Big(F\big(\Phi_{ij}\zeta_0,\pi_s^{\otimes n}(x,du)\big)-F\big(\zeta_0,\pi_s^{\otimes n}(x,du)\big)\Big) ds
\end{align*}
Looking at the generator appearing in the above equation, we introduce the Markov process $(\zeta_t)_{t\in \R_+}$ with values in $\Co(\mathcal{U}^n,\R)$ that jumps from $\zeta$ to $\Phi_{ij}\zeta$ with rate $1/2$ for all $i,j\in \{1,\dots n\}$ and whose evolution between jumps is given by the generator $\sum_{i=1}^n A_i$.
The generator of $(\zeta_t)_{t\in \R_+}$ thus satisfies, for a measure $\pi\in \mathcal{M}_1(\mathcal{U}^n)$,
\begin{align}
\mathcal{G}_\zeta F(.,\pi)(\zeta)=
\frac{1}{2}\sum_{{\scriptsize \begin{array}{c}i\not= j\\i\in \{1,\dots n\}\\
j\in \{1,\dots,n\}\end{array}}} \Big(F\big(\Phi_{ij}\zeta,\pi\big)-F\big(\zeta,\pi\big)\Big)+ F\Big(\sum_{i=1}^n A_i\zeta,\pi\Big).\label{def:Gzeta}
\end{align}Thus, $(\zeta_t)_{t\in \R_+}$ and $(\pi^{\otimes n}_t(x,.))_{t\in \R_+}$ are in duality through $F$, and:
\begin{align*}
\E_{\pi_0(x,.)}\Big(F\big(\zeta_0,\pi_t^{\otimes n}(x,.)\big)\Big)
= & \E_{\zeta_0}\Big(F\big(\zeta_t,\pi_0^{\otimes n}(x,.)\big)\Big)\\
= & \int_{\mathcal{U}}\dots \int_{\mathcal{U}} \E_{\zeta_0} \Big(\zeta_t(u_1,\dots u_n)\Big) \pi_0(x,du_1)\dots \pi_0(x,du_n),
\end{align*}by Fubini's theorem.\\

Now, comparing their generators \eqref{def:G} and \eqref{def:Gzeta}, we see that $(\zeta_t)_{t\in \R_+}$ is itself in duality with $(U_t^1,\dots U_t^n)_{t\in \R_+}$ through the function $f(\zeta, (u_1,\dots u_n))=\zeta(u_1,\dots u_n)$. Thus:
\begin{align*}
\E_{\zeta_0} \Big(\zeta_t(u_1,\dots u_n)\Big)= & \E_{(u_0,\dots u_n)} \Big(\zeta_0(U^1_t,\dots U_t^n)\Big),
\end{align*}from which we deduce \eqref{identif:loisphylo}.
\end{proof}


\bigskip
Thus, it appears that if we draw $n$ individuals independently in $\pi_t$, the distribution of their markers is the same as if we consider a Kingman coalescent started from $n$ individuals at time $t$, draw the marker value of the individuals of time $0$ independently in $\pi_0$ and let these values evolve along the branches according to the generator $G$. In another words, the distribution $\pi_t^{\otimes n}$ is the same as the distribution of the values of diffusions $G$ along the branches of a Kingman coalescent.\\

\bigskip

From these results, we deduce the forward-backward coalescent: conditionally on the species tree obtained by the PES, we have along each branches a Kingman coalescent whose parameters depend on the sizes of each species, and hence on the entire trait distribution.

\subsection{Simulation of phylogenies with the forward-backward coalescent process}\label{sec:simulations}

We describe how to simulate the phylogenies of $n$ individuals sampled at time $0$. Because the simulation starts with obtaining a path of the PES in the forward sense, a parameter $-t_{\sc{sim}}$ ($t_{\sc{sim}}>0$) corresponding to the starting point in the past of the PES is needed. By translation, it is equivalent to say that we start the simulation of the PES at time $t_0=0$ and we reconstruct the lineages of individuals at the present time $t_{\sc{sim}}$. Once the PES has been simulated, phylogenies of individuals sampled at time $t_{\sc{\sc{sim}}}$ can be simulated conditionally to the PES, backward time, as described in Section \ref{sec:def-FBcoal}. We emphasize again that because of competition and interactions between individuals, the phylogenies can not be obtained by considering a sole process that would be backward in time. It would be possible to simulate the phylogenies of the whole population forward time and then to sample within these simulated phylogenies, but that would be time-consuming. Instead, the forward-backward coalescent process is based on the (forward in time) PES, which is a macroscopic process at the population level, and then reconstitutes the phylogenies of sampled individuals at $t_{\sc{sim}}$ (backward time).

To keep it intuitive, we describe the algorithm for the Dieckmann-Doebeli model described in Section \ref{sec:Dieckmann-Doebeli}.

\subsubsection{Simulation of the PES}

At $t_0=0$, the population is assumed monomorphic. The trait $x_0$ of the first population is drawn in a normal distribution (mean $0$, standard deviation $\sigma_m$). Its density is the equilibrium of the corresponding isolated population: $\widehat{n_0} = \frac{b(x_0) - d(x_0)}{ \eta(x_0) C(x_0,x_0)}$.\\

We then simulate mutation events recursively until the end of the simulation (when time $t_{\sc{sim}}$ is reached). Assume that at the current step, the population consists of $p$ traits $x_1,\dots x_i, \dots x_p$. \\
For each sub-population, the duration before the emergence of a possible new mutant in this sub-population is randomly drawn in an exponential distribution of rate $\widehat{n}(x_i ; x_1\dots x_p) b(x_i)$. The shortest of these durations determines the mutation that occurs (if its occurs before the end of the simulation). The mutant trait $x_{\sc{mut}} \in [-2,2]$ is drawn in a truncated normal distribution with mean the trait value of the sub-population from which it emerges and standard deviation $\sigma_m$.\\
Ecological drift can cause the extinction of the mutant before its invasion: its survival probability depends on its initial growth rate in the resident population at equilibrium (the $p$ existing populations of trait $x_1,\dots x_p$ with their densities). If the mutant does not reach a deterministic threshold, which happens with probability
$ [f(x_{\sc{mut}} ; x_1,\dots x_p)]_+/b(x_{\sc{mut}})$,
nothing happens, and the simulation keeps going with the same traits and densities. Another mutant emergence is then simulated.\\
Once a mutant emerges, the new Lotka-Volterra system \eqref{eq:lotka-volterra} must be solved. The mutant density is initially low. Numerically, we took the density $\widehat{n}(x_{\sc{mut}} ; x_1, \dots x_p, x_{\sc{mut}})$ of an isolated population with trait $x_{\sc{mut}}$ divided by 100. The system is integrated until it approaches the equilibrium (all growth rates close of 0) and then a Newton method for Lotka-Volterra ODE is used to refine the solution. If a sub-population density is below a certain fixed threshold, it is considered extinct and suppressed from the population.\\

We store the polymorphic evolutionary sequence (PES): a list of population structures (sub-population traits, densities, emergence time and origin). Example of PES creation and emergence of a mutant is given in Fig. \ref{Fig:PES}.

\begin{figure}[ht!]\hspace{0cm}
\begin{tabular}{cc}
(a) & \includegraphics[width=16cm]{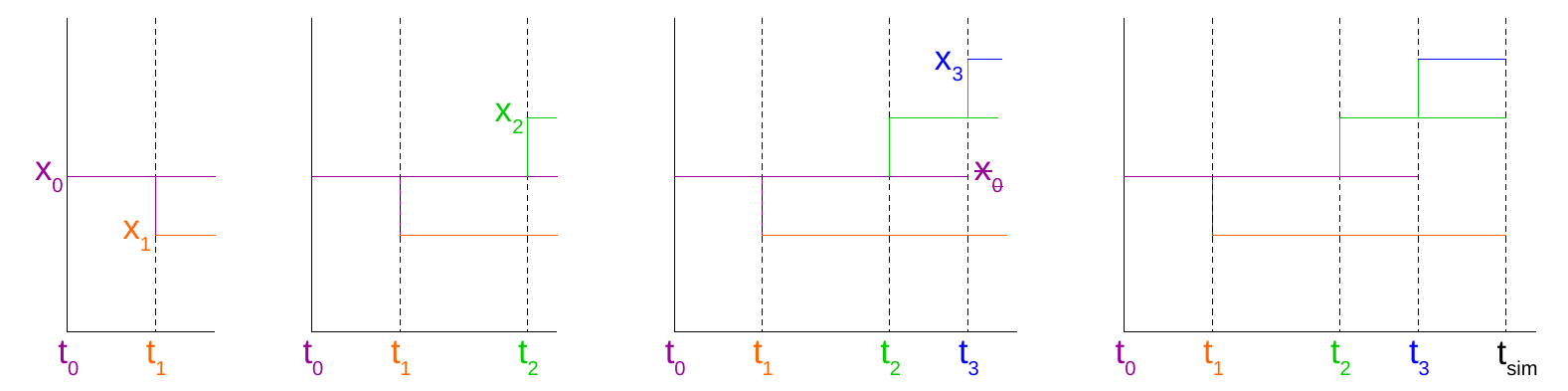}\\
(b) & \includegraphics[width=16cm]{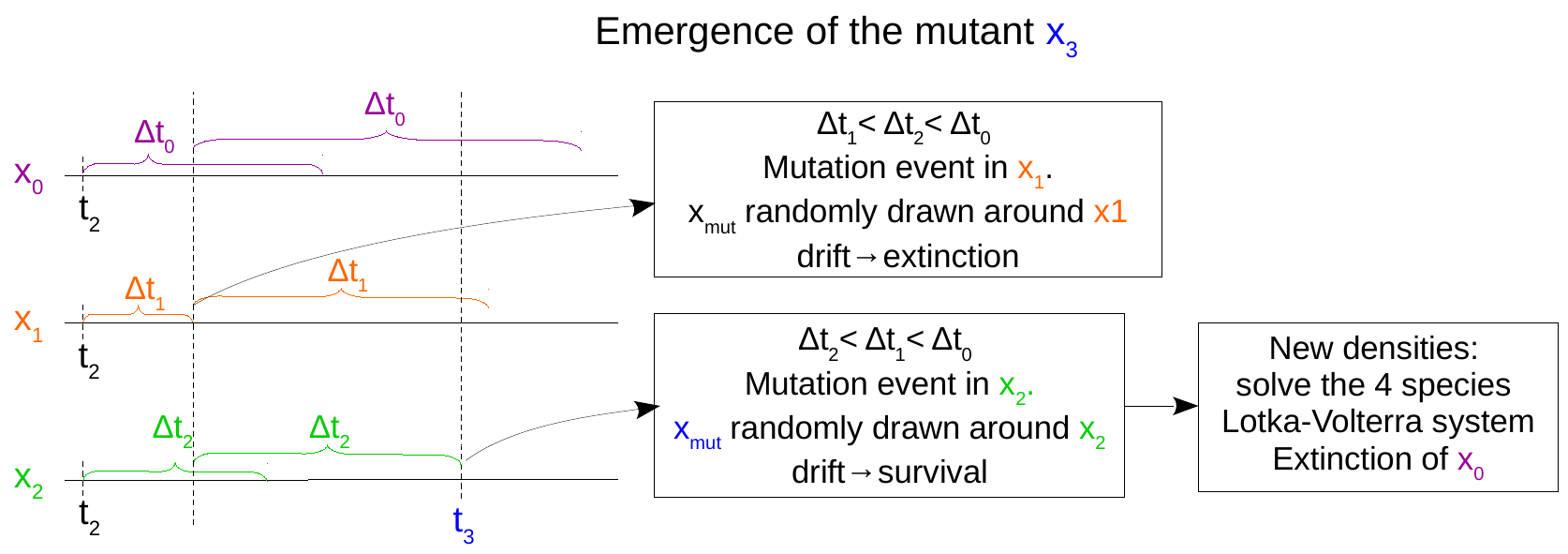}
\end{tabular}
\caption{\textit{{\small (a) Example of the creation of a PES. From left: the first mutation time $t_1$ and the first mutant $x_1$, which invades, are simulated. Second, a second mutation time $t_2$ and a second mutant which invade are simulated. Third: a third mutant $x_3$ invades and the trait $x_0$ is eliminated by the appearance of this new species. Finally, no more mutation occur before $t_{\sc{sim}}$. (b) Example of the emergence of the mutant in a population with three traits $x_0, x_1, x_2$. $\Delta t_i$ corresponds here to the exponential time before the occurrence of a mutant in the species with trait $x_i$. Here the first mutant appears in the population of trait $x_1$ but fails to invade. Three new clocks are then defined at the time of the failed invasion. It is finally in the species of trait $x_2$ that the new successful mutant occurs.}}}\label{Fig:PES}
\end{figure}

\subsubsection{Simulation of sampled phylogenies}

First, $n=1000$ individuals are sampled in the population at $t_{\sc{sim}}$. Their traits depend on the relative abundances of the sub-populations at $t_{\sc{sim}}$. If the population consists of $p$ species with traits $x_1,\dots x_p$, the sample is drawn from a multinomial distribution of probabilities \[ \frac{\widehat{n}(x_i ; x_1, \dots x_p)}{\sum_{j=1}^k \widehat{n}(x_j ; x_1, \dots x_p)}.\]

Then, we consider the coalescence of the sampled individuals. Individuals coalesce within their species following the Kingman coalescence model. In a population of $p$ species with traits $x_1,\dots x_p$, to each pair of individual in the same species, say of trait $x_i$ for $i\in \{1,\dots p\}$, is associated an independent random clock with parameter \[\frac{b(x_i)}{\widehat{n}(x_i ; x_1, \dots x_p)}.\]The next coalescence corresponds to the smallest of these exponential random times. The simulation can be done directly at the level of species, which spares some simulation of random exponential variables. In the species of trait $x_i$, the next coalescence hence occurs after an exponential random time of parameter $b(x_i) n_i (n_i - 1) / (2 \widehat{n}(x_i ; x_1, \dots x_p))$; with $n_i$ the number of distinct lineages in the species of trait $x_i$. The two individuals that coalesce are then drawn randomly in the species of trait $x_i$ and are replaced by their common ancestor. Notice that the coalescence rate must be modified at each trait-mutation event.\\

While going backward in time, we may encounter trait-mutation events. In the sub-population emerging from this mutation event, all individual phylogenies that are still distinct coalesce to produce the common ancestor of the sub-population. This common ancestor is then assigned to its parental sub-population (the sub-population in which the mutant emerged). Then coalescence within sub-population starts again until the next mutation event (or $t_0$).
Coalescence steps are described in Fig \ref{Fig:Coalescent}.


\begin{figure}[ht!]\hspace{0cm}
\parbox{16cm}{
\includegraphics[width=16cm]{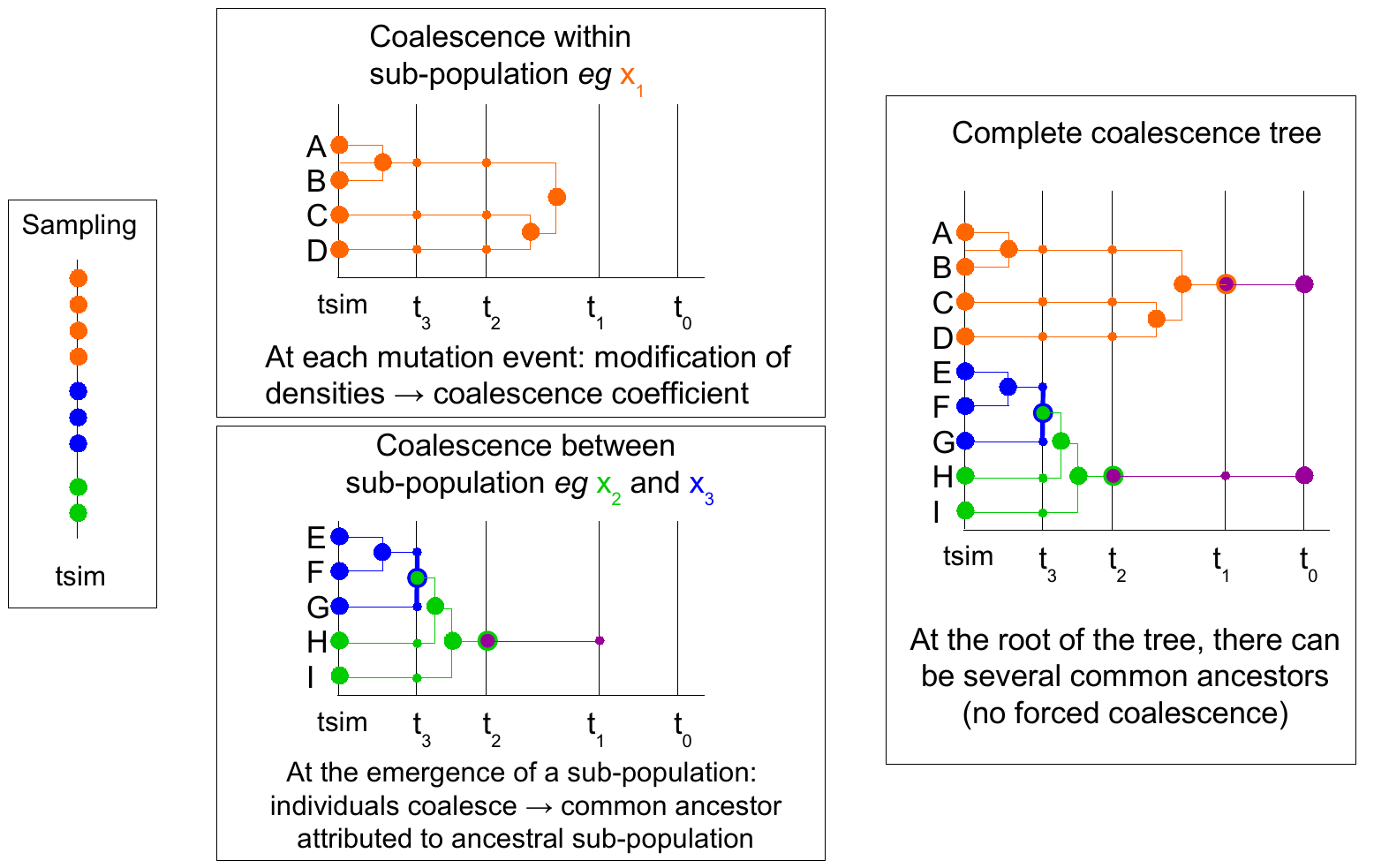}
\caption{\textit{{\small Example of the emergence of a mutant in a population with 3 traits.}}}\label{Fig:Coalescent}}
\end{figure}

\subsubsection{Neutral locus mutation along the phylogeny}

If we want to simulate the marker process along lineage, once the phylogenies have been obtained, then we simulate forward in time Brownian motions along each branch of the phylogenetic tree, with every Brownian motion starting from the value at which the Brownian motion of the mother branch stopped.\\

However, many descriptive statistics of population structure depend on neutral markers rather than on selected traits. This can be handled with a generator slightly different from the one presented in \eqref{eq:FV}. We simulated neutral microsatellite markers mutations along the phylogenetic tree: at each mutation event, the marker gain or loss one motif repetition. So here $\U=\N$ and the transitions are $\pm 1$ except if $u=0$ in which case it can only increase by 1. The mutation events occur in forward time: the mutations are transmitted from the common ancestor to its line of descent until the sampled individuals. When the individual separates into two individuals, they inherit the markers of their parent and then mutations accumulate separately in the two lines.\\

For each lineage, the time step between each neutral mutation is randomly drawn in an exponential distribution of rate $b(x_i) \theta$ (if the trait of the individual changes, the coefficient too). The mutation event results in the loss or in the gain of a motif repeat, with equal probability. Notice that the mechanism generating such marker mutation does not correspond to assumption \eqref{hyp:tauxmutations}. They are much slower than expected. However, the results can be carried to this case: the fact that phylogenies are Kingman coalescent processes can be adapted in this case by adding a color that mutates with the right probability $q_K$.

\section{Principle and framework of the ABC estimation}\label{sec:ABC-recap}


Because the competition and interactions involve the whole population, one is confronted, when trying to write the likelihood of sampled individuals' phylogenies to formulas where many terms are unobserved. Since we are in continuous time with sub-populations of varying densities, numerical integration over the missing events whose number is unknown is impossible. The ABC method (see \cite{beaumontzhangbalding,marinpudlorobertryder} for a presentation) appears as an alternative which does not use likelihood functions but rather relies on numerical simulations and comparisons between simulated and observed summary statistics. We conducted the ABC analysis using the package \rm{abc} in R \cite{package-abc}.\\

The ABC method is a Bayesian method. Let us denote by $\theta$ the parameters to estimate and by $\pi(d\theta)$ its prior distribution. The data is denoted by $\mathbf{x}$. Instead of trying to estimate the posterior distribution $\pi(d\theta \ |\ \mathbf{x})$, which might be very intricate when explicit formulas for the likelihoods are not available, the ABC method proposes to approximate the target distribution $\pi(d\theta \ |\ S(\mathbf{x}))$ where $S(\mathbf{x})$ is a vector of descriptive statistics, called summary statistics. Let us denote by $S_{\sc{obs}}=S(\mathbf{x})$ the value of the chosen vector of descriptive statistics computed on the observations. \\

The principle of the ABC method is to simulate new data corresponding to new parameters drawn into the prior distribution. The parameters that generate data leading to summary statistics close to the observed $S_{\sc{obs}}$ are then given a more important weight. The approximation of the target distribution is then the reweighted and smoothed empirical measure of the simulated parameters. The idea of ABC is grounded on non-parametric statistical theory (see e.g. Blum \cite{blum}) and on corrections methods (see \cite{beaumontzhangbalding}). More precisely, the ABC estimation requires the following steps:
\begin{enumerate}
\item Draw $N$ independent parameter sets $\theta_i$, $i\in \{1,\dots N\}$ in the prior distribution $\pi(d\theta)$,
\item For each $i\in \{1,\dots N\}$, simulate a realization of the forward-backward coalescent process following Section \ref{sec:simulations}.
\item Determine for each simulation the corresponding descriptive statistics $S_i$,
\item Compare the descriptive statistics $S_i$, $i\in \{1,\dots N\}$, with $S_{\sc{obs}}$ to determine the set of parameter values ($\theta_i$) leading to the best fit between $S_i$ and $S_{\sc{obs}}$:
\begin{itemize}
\item[(a)] rejection method: retain only the $\delta=0.1\%$ simulations with the lower Euclidean distance between $S_i$ and $S_{\sc{obs}}$, by setting a weight $W_i=1/(\delta N)$ to the latter and $W_i=0$ otherwise.
\item[(b)] smooth re-weighting: to each simulation is computed a weight $W_i=K_\delta(S_i-S_{\sc{obs}})$ where $\delta$ is a tolerance threshold and $K_\delta$ a smoothing kernel.
\end{itemize}
\item the posterior distribution is then estimated thanks to the weights $(W_i)_{1\leq i\leq n}$ defined above:
\begin{equation}\label{def:posterior}
\widehat{\pi}(d\theta | S_{\sc{obs}}):= \sum_{i=1}^N \frac{W_i}{\sum_{j=1}^N W_j} \delta_{\theta_i^*}(d\theta)
\end{equation}where $\theta^*_i$ is a correction of $\theta_i$ accounting for the difference between $S_i$ and $S_{\sc{obs}}$, which we detail in the next paragraph.
\end{enumerate}

In the \rm{abc} package, several correction methods are available for computing $\theta^*_i$. These corrections have been proposed in \cite{beaumontzhangbalding}. The idea is that if locally around $S_{\sc{obs}}$ the parameter $\theta$ is a function of the summary statistics, say $\theta=f(S,\varepsilon)$ where $\varepsilon$ is a noise component, then when $S_i$ is close to $S_{\sc{obs}}$, it is possible to correct $\theta_i=m(S_i)+\varepsilon_i$ into $\theta_i^*=m(S{\sc{obs}})+\varepsilon_i$ (see also the remark in \cite[Eq. (4.1)]{blumtran}). Using the neural network method in the \rm{abc} package, 
and based on the couples $(\theta_i,S_i)$, $i\in \{1,\dots N\}$, we can estimate the (possibly non-linear) regression function $m$ in the neighborhood of $S_{\sc{obs}}$. If $m$ is assumed locally linear at $S_{\sc{obs}}$, a  local  linear  model  can be chosen (\rm{loclinear} in the
\rm{abc} function), else, neural networks (\rm{neuralnet}) can be used to account  for  the  non-linearity  of $f$. Once the regression is performed, using the estimator $\widehat{m}$ of $m$, the corrected values $\theta_i^*$ are defined as
$$\theta_i^* = \widehat{m}(S_{\sc{obs}}) + \frac{\widehat{\sigma}(S_{\sc{obs}})}{\widehat{\sigma}(S_i)} \big(\theta_i-\widehat{m}(S_i)\big) $$
here the bracket in the r.h.s. is an estimator of the residuals $\varepsilon_i$.
Additionally, a correction for heteroscedasticity is applied (by default)
where $\widehat{\sigma}(.)$ is an estimator of the conditional standard deviation (see \cite{blumfrancois} for details).

\section{Parameters and priors for the estimation of the Dieckmann-Doebeli model}\label{sec:prior}

In order to use the ABC method, $N=400 000$ simulations were run using independent random parameter values that are sampled from their prior distributions. The prior distributions of each parameter of the model should reflect the expected values of the parameters. The parameter sets for the four \textit{pseudo-data} sets used in the ABC procedure are given in Tab. \ref{Table:data_parameters}. Notice that the pseudo-data A and B (resp. C and D) in Tab. \ref{Table:data_parameters} have been obtained by simulations with the same sets of parameters. For most parameters, we do not have insights on the expected distribution, thus we choose uniform distributions (details in table \ref{Table:prior_distrib}).

\begin{table}[ht!]
\begin{center}
\begin{tabular}{|c|c|c|c|c|c|c|c|c|}
\hline
parameter & ID & $p$ & q &$\sigma_b$	 & $\sigma_c$ & $\sigma_m$ & $\eta_c$ & $t_{\sc{sim}}$ \\
\hline
value & A & $0.0076$ & $0.7503$ &$1.186$ & $0.4951$ & $0.1448$ & $0.0211$ & $1025.619$ \\
\hline
value & B & $0.0076$ & $0.7503$ &$1.186$ & $0.4951$ & $0.1448$ & $0.0211$ & $1025.619$ \\
\hline
value & C & $0.0071$ & $0.6815$ &$1.7725$ & $0.3866$ & $0.1372$ & $0.0306$ & $335.2116$ \\
\hline
value & D & $0.0071$ & $0.6815$ & $1.7725$ & $0.3866$ & $0.1372$ & $0.0306$ & $335.2116$\\
\hline
\end{tabular}
\caption{\textit{{\small Parameters of the simulations used as data set (referred with their ID in the text).}}}\label{Table:data_parameters}
\end{center}
\end{table}

\begin{figure}[!ht]
\begin{center}
\includegraphics[width=4cm, angle=-90]{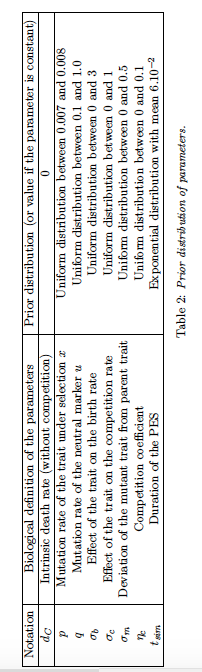}
\caption{\textit{{\small Prior distribution of parameters.}}}\label{Table:prior_distrib}
\end{center}
\end{figure}

\section{Summary statistics for the ABC estimation based on phylogenetic trees}\label{sec:summarystatistics}

We detail here several summary statistics that can be used for the ABC estimation. Depending on how much information we have in the data, we can use all these statistics for the ABC or only a subset of them.
The three first families (Sections \ref{sec:sumstat-famille1}-\ref{sec:sumstat-famille3}) provide description of the trait and marker distribution in the population at the final time. These statistics can be computed only at the sample level or also at the population level if additional knowledge is available. Many of these statistics are listed in Pudlo et al. \cite{pudlomarinestoupcornuetgautierrobert}. The family \ref{sec:sumstat-famille4} contains statistics describing the shape of the phylogenetic tree and can be used if some additional statistics, for instance related to fossils or datation of past speciation events, are  available.

\subsection{Population structure from the total population}\label{sec:sumstat-famille1}

Statistics measured on the last step of the PES (at $t_{\sc{sim}}$):
\begin{itemize}
\item number of coexisting species (number of traits in the population),
\item mean abundance of the coexisting species,
\item variance of the abundance of the coexisting species,
\item mean trait $x$ of the species,
\item trait variance between species,
\item mean trait $x$ of the individuals (depends on sub-population densities),
\item variance between the traits of individuals (depends on  sub-population densities).
\end{itemize}

\subsection{Population structure from the sampled individuals}\label{sec:sumstat-famille2}

Statistics measured on  the $n$ sampled individuals:
\begin{itemize}
\item number of sampled sub-populations,
\item relative abundance variance between sampled sub-populations,
\item mean trait of sampled sub-populations,
\item trait variance between sampled sub-populations,
\item mean trait of sampled individuals (depends on relative abundance),
\item trait variance between sampled individuals (depends on relative abundance).
\end{itemize}

\subsection{Description of the allelic distribution of the trait and marker}\label{sec:sumstat-famille3}

\begin{itemize}
\item number of distinct alleles for the marker,
\item variance of the marker's allele distribution,
\item genetic diversity (measuring the width of the marker's allele support). In case the marker corresponds to microsatellites, the genetic diversity is expressed as
\[\sum_{i} \sum_{j} \left((i\ h_{i}  - j\ h_{j} \right)^2,\]
with $h_i$ and $h_j$ the frequencies of the populations with $i$ and $j$ repetitions. This index corresponds to the expected difference between each couple of individuals in the population. 
\item Unbiased Gene Diversity, which in the case of microsatellites rewrites as
\[\frac{n}{n-1}\sum_{i} \sum_{j} \left((i\ h_{i}  - j\ h_{j} \right)^2.\]
It is the average difference between two individuals in the population.
\item M-index: ratio of the number of alleles for the marker on the width of the marker's allele support,
\[M=\frac{\mbox{number of alleles}}{\mbox{max(number of repetitions)} - \mbox{min(number of repetitions)}}.\]This index is the average percentage of intermediate allelic states that are occupied. The lower the M index, the more the population has lost possible alleles.
\item number of traits in the population,
\item mean and variance of traits
\item abundance distribution for traits.
\end{itemize}




The allelic diversity is useful to infer population size through time, because drift leads to the loss of alleles in small populations. The M-index is useful to detect reductions in population size (their intensity and duration): during a bottleneck, alleles are lost randomly (not specifically at the end of the range size), thus the percentage of occupied intermediate allelic states decreases (in huge population, alleles should exist throughout the range size). Low allele diversity and low M-index indicate recent reduction of size while low allele diversity but high M-index corresponds to populations that have been small for a long time.\\


Usually markers that are used correspond to several loci. In this case, additional indices can be added  to describe the joint distribution of the marker on these loci. Let us denote by $r$ the number of loci. Let us consider two species of traits $x$ and $y$ respectively. We denote by $x_{ij}$ and $y_{ij}$ the frequencies at the locus $j$ of the allele $i$ in the populations $x$ and $y$:
\begin{itemize}
\item $F_{ST} = \frac{\dfrac{\sum_{j}^r \sum_{i} \frac{x_{ij}^2}{r} + \sum_{j}^r \sum_{i} \frac{y_{ij}^2}{r}}{2} - \sum_{j}^r \sum_{i} \frac{x_{ij} y_{ij}}{r}}{1 - \sum_{j}^r \sum_{i} \frac{x_{ij} y_{ij}}{r}}$
\item $\delta\mu^2 = \frac{1}{r} \sum_{j}^r \left( \sum_{i} i\ x_{ij}  - \sum_{i} i\ y_{ij} \right)^2 $. The alleles ($i$) must be expressed as a number of motif repeats.
\item Nei's $D_A$ distance: $D_A = 1 - \sum_{j}^r \sum_{i} \sqrt{x_{ij} y_{ij}}$
\item Nei's standard genetic distance $D_S$: $D_S = - \ln \left( \frac{\sum_{j}^r \sum_{i} \frac{x_{ij} y_{ij}}{r}}{\sqrt{(\sum_{j}^r \sum_{i} \frac{x_{ij}^2}{r} ) (\sum_{j}^r \sum_{i} \frac{y_{ij}^2}{r} )}} \right)$
\end{itemize}

$F_{ST}$, $\delta\mu^2$, $D_A$ and $D_S$ are calculated with equal weight given to each population (densities of the populations are not taken into account), or weighted by the densities of the populations (for each pair of population $x$ and $y$, the statistic is weighted with $\frac{n_x \times n_y}{\sum_{z\in \mathcal{X}} n_{z}}$.

\subsection{Statistics measured on the PES and on the coalescent tree}\label{sec:sumstat-famille4}

Statistics measured on the whole PES:
\begin{itemize}
\item number of mutation events along the PES
\item Sackin's index $I_S^n = \sum_{i=1}^n N_i$ with $n$ the number of leaves of the tree and $N_i$ the number of internal nodes from the leave to the root.
\item sum of the branch length
\item sum of the external branch length\\
\end{itemize}

Statistics measured on the coalescent tree:
\begin{itemize}
\item sum of the branch length
\item sum of the external branch length
\item number of cherries (\textit{i.e.} leaves that coalesce together)
\item time before the most recent common ancestor
\item time before the  most recent common ancestor if time of coalescence was neutral
\end{itemize}

\section{Results of the ABC estimation for the sets of pseudo-data}\label{sec:posterior-distrib}

Prior and posterior distribution of parameters estimation determined with the ABC method, for each of the models A-D corresponding to Table \ref{Table:data_parameters}. All descriptive statistics (blue), statistics with complete knowledge of the population (pink), statistics on sampled individuals, with a knowledge of the trait (red) and statistics on sampled individuals, without any knowledge on the trait (orange).\\

For the pseudo-dataset $A$, we also investigate the impact of the number of microsatellites. As can be seen in Fig. \ref{Fig:nb_microsat_all}, the posterior square deviations for each of the parameters remains quite stable. More precision is achieved on $t_{sim}$ when we have more microsatellite data. For the parameters $p\times \sigma_m$ or $q$, we see that a higher number of microsatellites allows to reach a better precision under scenario 4.

\begin{figure}[ht!]
\begin{center}
\includegraphics[width=14cm]{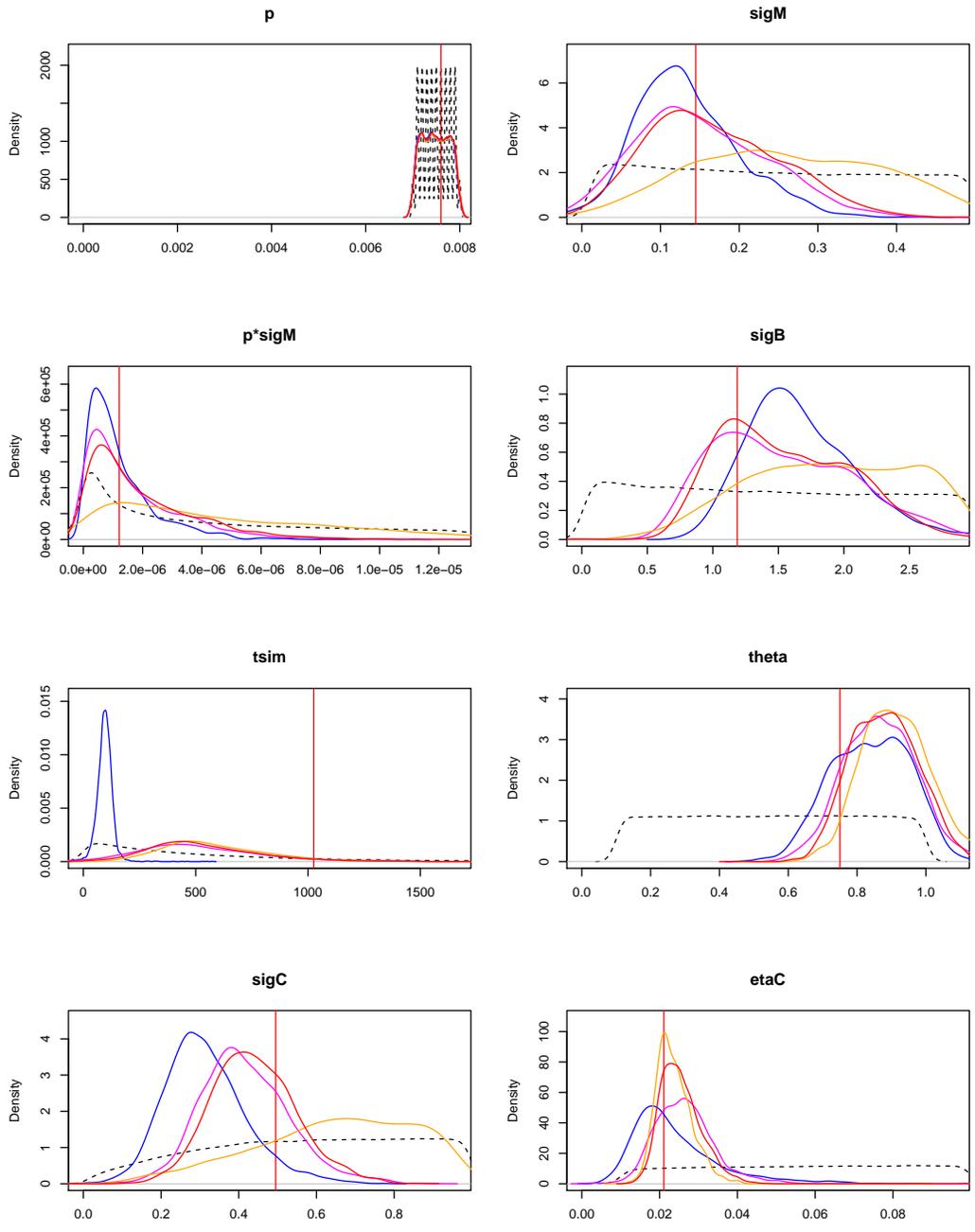}\\
\caption{{\small \textit{Results of the ABC estimation for the pseudo-data B (see \ref{Table:data_parameters})}}}
\end{center}
\end{figure}

\begin{figure}[ht!]
\begin{center}
\includegraphics[width=14cm]{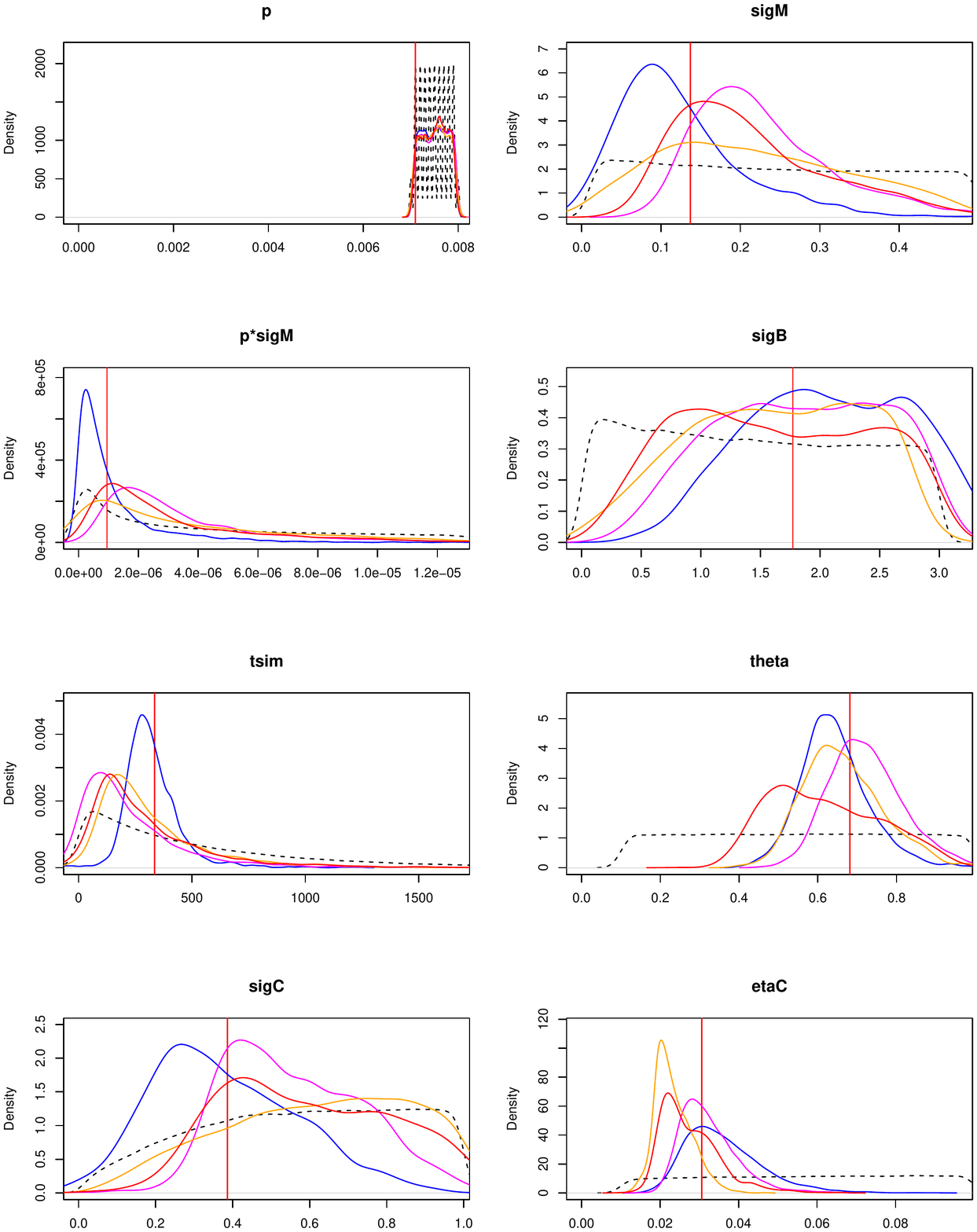}\\
\caption{{\small \textit{Results of the ABC estimation for the pseudo-data C (see \ref{Table:data_parameters})}}}
\end{center}
\end{figure}

\begin{figure}[ht!]
\begin{center}
\includegraphics[width=14cm]{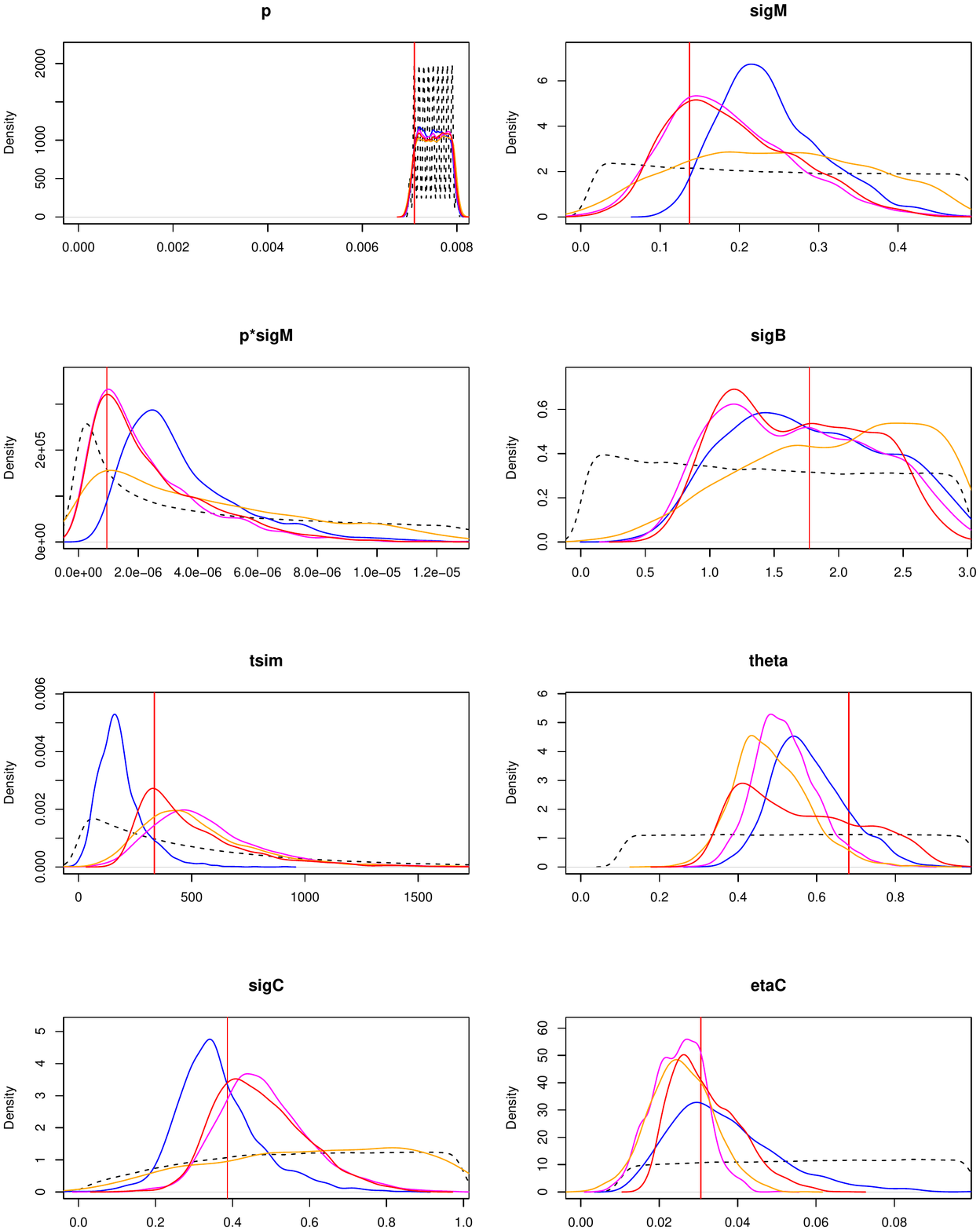}\\
\caption{{\small \textit{Results of the ABC estimation for the pseudo-data D (see \ref{Table:data_parameters})}}}
\end{center}
\end{figure}

\begin{figure}[ht!]
\begin{center}
\includegraphics[width=14cm]{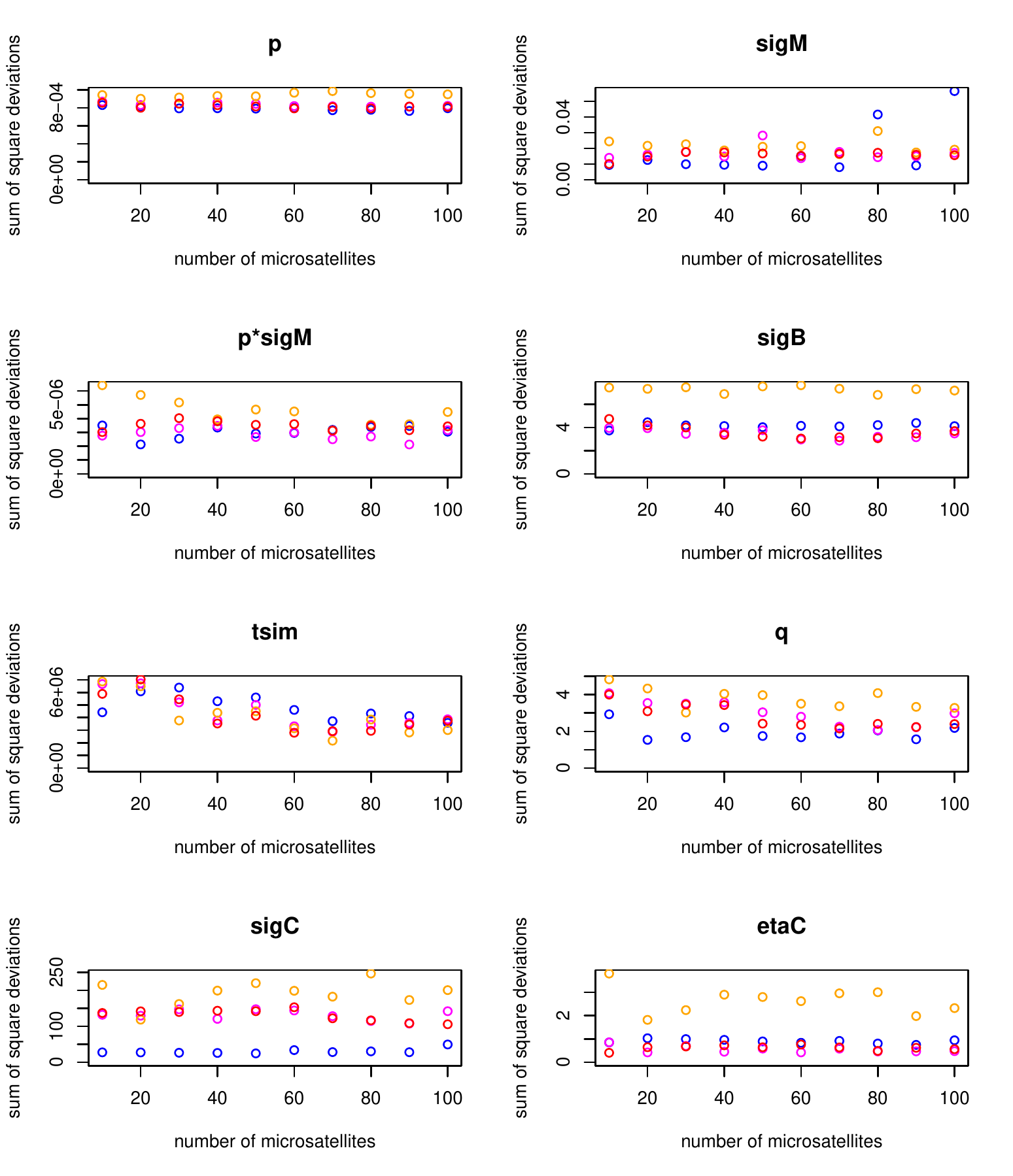}
\caption{\textit{\footnotesize Posterior square deviations for the different parameters under the various estimation scenarios (pseudo-data set $A$, see App. Tab. \ref{Table:data_parameters}): Blue, Scenario 1 (All descriptive statistics are available); Pink, Scenario 2 (data from the totality of the population); Red, Scenario 3 (data from a sample of the population); Orange, Scenario 4 (data from a sample of the population, the traits $x$ is not known). }}\label{Fig:nb_microsat_all}
\end{center}\end{figure}

\clearpage

\section{Does the Doebeli-Dieckmann coalescent follow a Kingman's coalescent?}\label{sec:neutrality-test}
\subsection{Neutrality tests: statistics and procedures}\label{sec:neutrality}

When there is no interaction and when the species play no role, the phylogenies of $n$ individuals are described by Kingman coalescent processes. In absence of interaction but when species still have to be taken into account, we have Kingman coalescent processes nested in the species tree. In the latter case, there is no coalescence between phylogenies belonging to distinct species until the speciation event where both species are reunited into a mother species.\\
In many studies, phylogenies are assumed to stem from neutral models (Kingman coalescent, branching Brownian motions). Then likelihood methods based on this assumption are used for example. The forward-backward coalescent process that is used here brings more complexity but proposes a way to take into account the interactions between species and the fact that the phylogenetic trees might be imbalanced. A first natural step is to check that the new model produces trees that could not be seen as generated by usual models (our choice goes here to Kingman coalescent). This shows that structure has been taken into account.

\textit{Test statistics.} Following the path opened by Fu and Li \cite{fuli}, we can consider several statistics that are usually chosen as test statistics to test the neutrality of a phylogenetic tree:
\begin{itemize}
\item the distribution of cherries, i.e. the number $C_n$ of internal nodes of the tree having two leaves as descendants. The normalized distribution of the number of cherries for a Kingman coalescent follows a Gaussian distribution (see \cite{blumfrancois2005-Sackin}):
\begin{equation}\label{test-stat:Cn}
\lim_{n\rightarrow +\infty}\frac{C_n - \frac{n}{3}}{\sqrt{\frac{2n}{45}}} \stackrel{(d)}{=}\mathcal{N}(0,1)
\end{equation}
\item the length $L_n$ of external branches, i.e. edges of the phylogenetic tree of the $n$ sampled individuals admitting one of the $n$ leaves as extremity. A beautiful result by Janson and Kersting \cite{jansonkersting} shows that when $n$ converges to infinity, the distribution of the normalized total length of external branches of a Kingman coalescent converges to a Gaussian distribution:
\begin{equation}\label{test-stat:Ln}
\lim_{n\rightarrow +\infty}\frac{L_n - \E(L_n)}{\sqrt{\frac{ \Var(L_n)}{2}}} \stackrel{(d)}{=} \mathcal{N}(0,1).
\end{equation}
\item the time $\TnMRCA$ to the most recent common ancestor (MRCA). In a Kingman coalescent, this time is distributed as the sum of $n$ independent exponential random variables with respective rates $\frac{b}{\widehat{n}}\frac{i (i-1)}{2}$, $i\in \{1,\dots n\}$, where $\widehat{n}$ is the density at equilibrium and $b$ is the natural growth rate when there is no species structure in the population.
\end{itemize}

Using the three test statistics that we have presented, we tested whether the observed phylogenies could be described by a Kingman coalescent or not. For the number of cherries and for the external branch length, we computed the renormalized statistics are performed normality tests: a Mann-Whitney test (using the function wilcox.test in R, with unpaired distributions), a Shapiro test and a Kolmogorov-Smirnov test. For the time to the MRCA, we performed an adequation test between the distribution of the MRCA in our model and the distribution of the sum of independent exponential random variables described above. \\
In all the tests, we set the significance level to $\alpha=5\%$ and the null hypothesis $H_0$ is that the observed tree can stem from a Kingman coalescent. If the p-value of the test is lower than this threshold $\alpha$, then we must reject $H_0$ and consider that distributions are different in our model and Kingman's model.\\

To check that these tests perform well, we applied them to simulated Kingman coalescent processes and checked when the null hypotheses were correctly accepted (see App. \ref{sec:test-neutralite-kingman}).\\

\textit{Application to the forward-backward coalescent.}  To test that the forward-backward coalescent produces non-neutral phylogenies, a first idea would be to perform a Monte-Carlo test by simulating several trees from this model and running the neutrality tests on these simulations. However, such result would be dependent on the parameters chosen for the simulations and would also necessitate a Monte-Carlo loop for exploring the set of parameters. Keeping these ideas in mind, we proceed a bit differently. Assume that we had data $\mathbf{x}$ generated by the forward-backward coalescent and \textit{a priori} distributions for the parameters $\theta$. In this Bayesian framework, would the null hypothesis $H_0\ :\ $ ``a phylogeny produced by the forward-backward coalescent is distributed as a Kingman coalescent" be accepted or not?\\

Given the data $\mathbf{x}$, the \textit{a posteriori} distribution of $\theta$ is computed using ABC, yielding an \textit{a posteriori} distribution on the phylogenies conditionally to $\mathbf{x}$. The estimation of the probability $\P\big( H_0 \mbox{ is accepted }\ |\ \mathbf{x}\big)$ can then be approximated by Monte-Carlo using the $N$ simulations performed by the ABC procedure: we use the simulations retained by ABC in order to reject abnormal simulations, and the associated weights produced by the ABC. This provides an approximation of the distribution of the test statistics \eqref{test-stat:Cn} and \eqref{test-stat:Ln} conditionally on $S_{\sc{obs}}=S(\mathbf{x})$ yielding in turn an approximation of $\P\big( H_0 \mbox{ is accepted }\ |\ S(\mathbf{x})\big)$:
\[\widehat{\P}\big( H_0 \mbox{ is accepted }\ |\ S(\mathbf{x})\big)=\sum_{i=1}^N \frac{W_i(\mathbf{x})}{\sum_{j=1}^N W_j(\mathbf{x})} \ind_{\frac{n}{3}-1.96 \sqrt{\frac{2n}{45}}\leq C_n(i)\leq \frac{n}{3}+1.96 \sqrt{\frac{2n}{45}}},\]
where $C_n(i)$ counts the number of cherries in the simulation number $i$.\\

 Since this question is investigated in a general framework without data, we use a Monte-Carlo approach to sum over the possible data $\mathbf{x}$. Averaging this result over datasets $\mathbf{x}$ conditionally on the number of species at the final time of the simulations allows to estimate the probability of accepting $H_0$ conditionally on the number of species at the sampling time. We condition on the number of species at the sampling time as we noticed that this variable impacts the outcome of the test. When there is only one species for example, our model predicts that the phylogenies look like Kingman coalescents (with a possible multiple merge at the first adaptive jump encountered). 
 However, with a growing number of species, we expect a deviation from the Kingman model due to forced coalescence at the creation of each species, and to the different coalescence rates between species.


\begin{itemize}
\item for each number of species ($m$ from $1$ to $10$), we randomly chose a simulation that we used as pseudo-data $\mathbf{x}$ in the ABC an Alyssa. We repeated this action $100$ times (providing $\mathbf{x}=\mathbf{x}_1,\dots \mathbf{x}_{100}$ for each value of $m$).
\item for each $\mathbf{x}$, an ABC is performed and provides an approximated \textit{a posteriori} empirical distribution conditionally on $S(\mathbf{x})$: to each simulations $i\in \{1,\dots N\}$ is associated a weight $W_i(\mathbf{x})$, see \eqref{def:posterior}. For each of these weighted simulations, the statistics $C_n$ and $L_n$ can be computed and the normality tests \eqref{test-stat:Cn} and \eqref{test-stat:Ln} can be performed. We deduce from this the approximated \textit{a posteriori} distribution of the p-values.
    \[\sum_{i=1}^N \frac{W_i(\mathbf{x})}{\sum_{j=1}^N W_j(\mathbf{x})} \delta_{{\scriptsize\mbox{p-value for the simulation }i}}.\]
\item Averaging over $\mathbf{x}$, conditionally on the number of species at the sampling time, we obtain an estimator of the probability that $H_0$ is rejected conditionally to the number of species at the sampling time.
\end{itemize}
The normality of the test statistics \eqref{test-stat:Cn} and \eqref{test-stat:Ln} are performed on the same ABC runs. Fig. \ref{Fig:Branch-cherries-lines3-6} gives the distributions of the normalized external branch length and the number of cherries of the weighted simulations from the ABC analysis presented in the previous section.\\

We set the significance level to $\alpha=5\%$ for all these tests. If the p-value of the test is lower than this threshold, then we must reject $H_0$ and consider our distribution does not follow a Gaussian distribution.
The results for the external branch lengths and cherries are shown in Fig. \ref{Fig:neutral-test}(a) and (b) and reveal that our coalescent trees differ from Kingman coalescent trees.\\

We also wanted to determine whether the coalescence time depended on the number of species. For this, we made a classical test of mean comparison associated with the null hypothesis: $H_0 \ :\ $``the mean time to MRCA in the data is equal to the mean time to MRCA in a Kingman coalescent''. This test is based on the following Student  statistic:
\begin{equation}\label{test:comparaison-moyennes}
\sqrt{n}\frac{\mbox{mean}(\mbox{distrib1}) - \mbox{mean}(\mbox{distrib2})}{\sqrt{\mbox{sd}(\mbox{distrib1}) + \mbox{sd}(\mbox{distrib2})}} \xrightarrow{d} {\cal N}(0,1).
\end{equation}
One can see in Fig \ref{Fig:neutral-test}(c) that the distributions differ, which is confirmed by the test. This confirms the finding obtained with the study of external branch lengths: timescale suffers from ignoring the interaction between species and can lead to false datings.

\subsection{Neutrality of the model without mutation}\label{sec:test-neutralite-kingman}

We first tested the consistency of our model with Kingman coalescent trees when no mutation of the trait can emerge: the population is monomorphic and the trait mutation is set to 0. We simulated trees with no mutation of the trait and tested the normality of the distribution of normalized external branch length and the normalized number of cherries. We also tested the adequation between the distribution of the time to the MRCA in our model and in a Kingman coalescent. The empirical distributions of these three statistics are represented in Fig \ref{Fig:Cn_Ln_tMRCA_p0_all}. Visually, these empirical distributions fit the targeted distributions under $H_0$ that are given in Section \ref{sec:neutrality}.

\begin{figure}[ht!]
\begin{center}
\begin{tabular}{ccc}
\includegraphics[height=5cm]{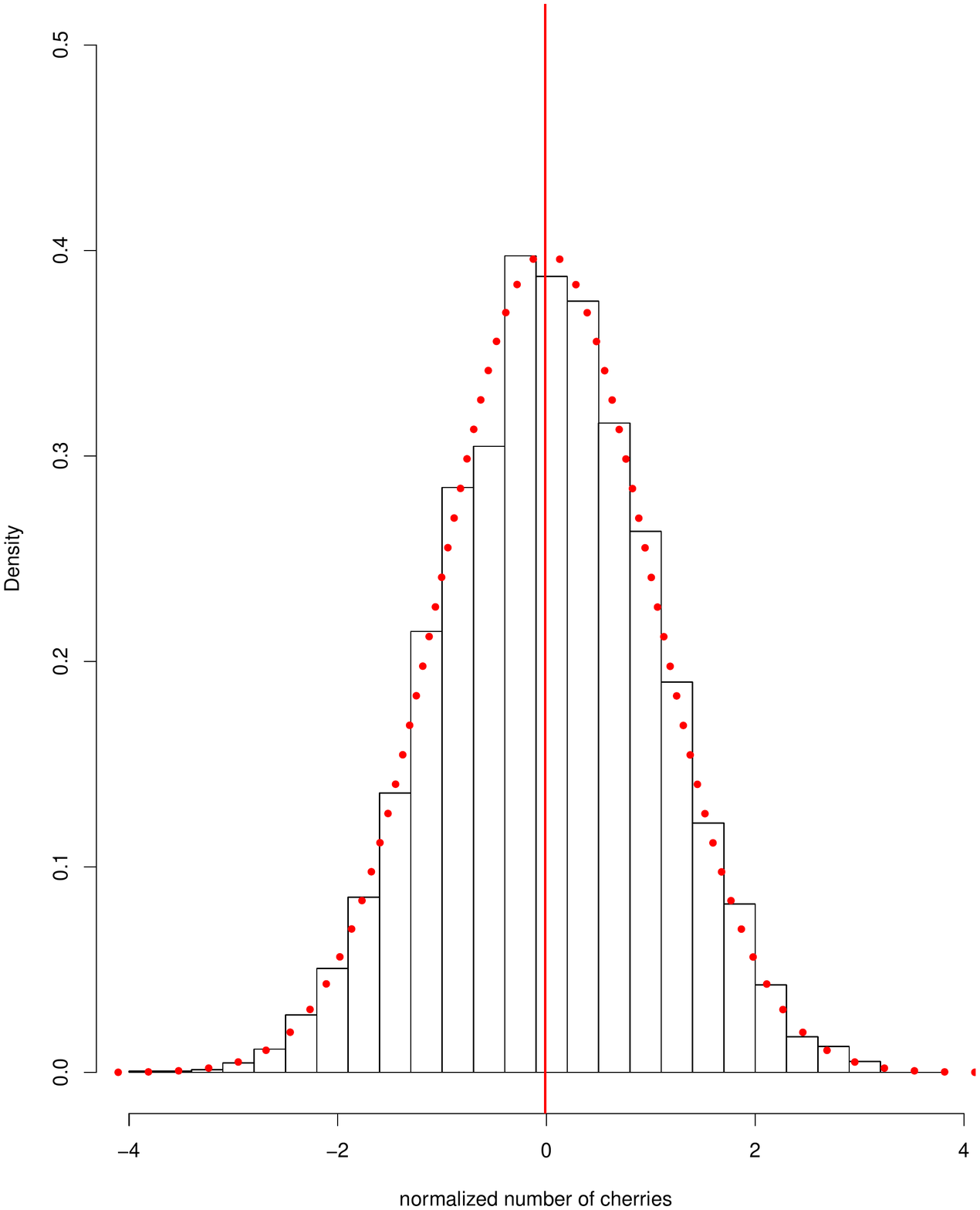} &
\includegraphics[height=5cm]{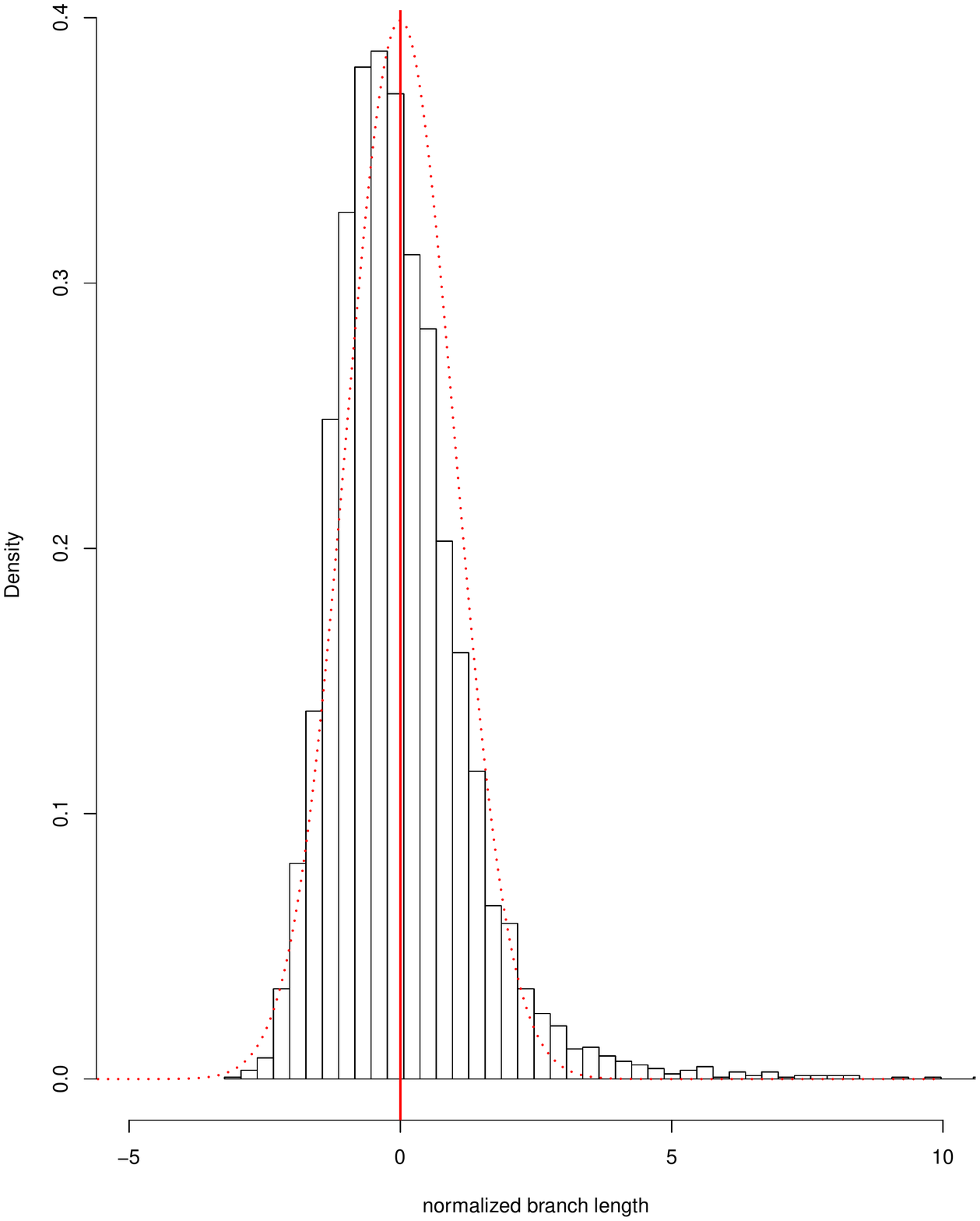} &
\includegraphics[height=5cm]{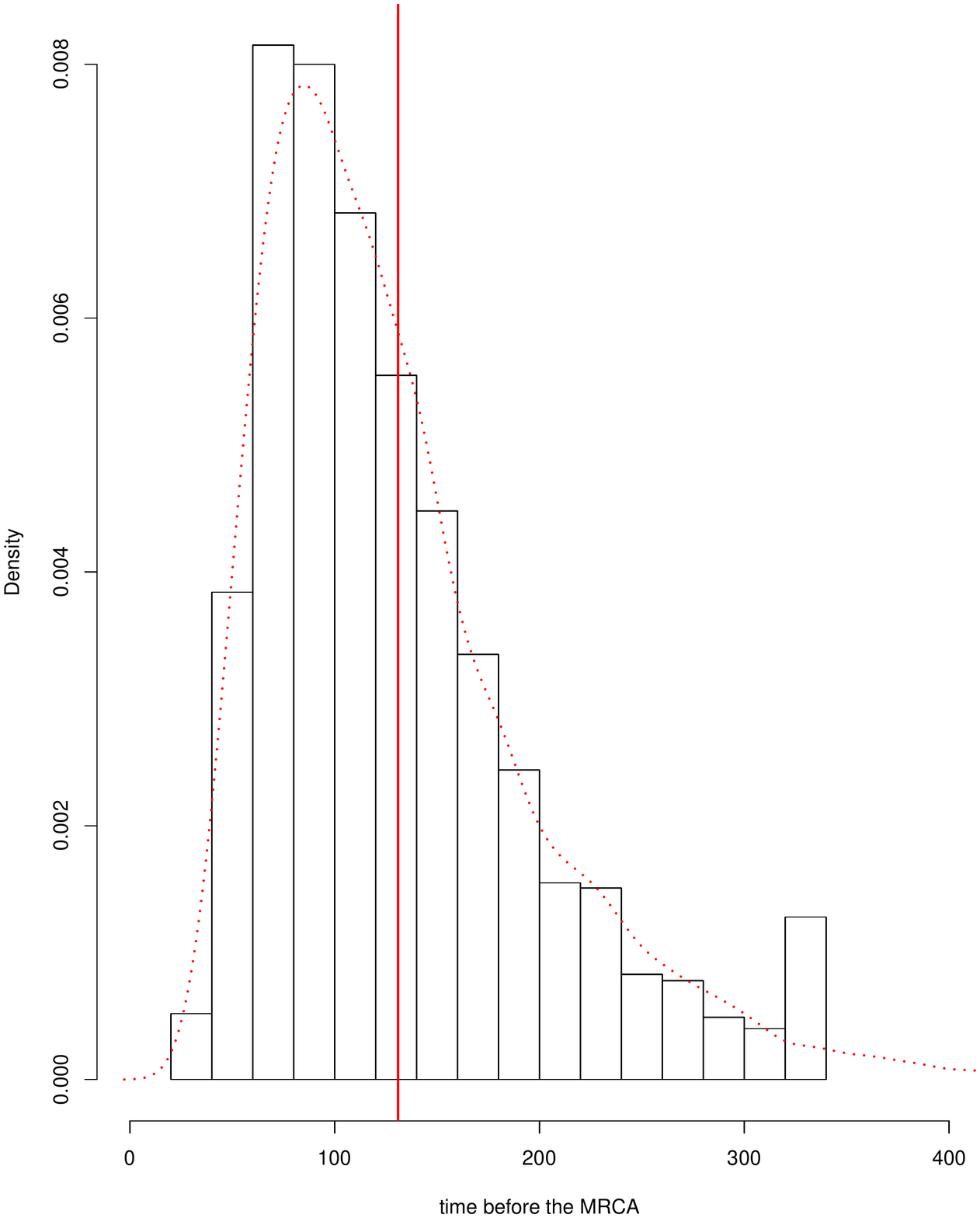}\\
(a) & (b) & (c)
\end{tabular}
\caption{\textit{{\small Distribution of (a) the normalized number of cherries, (b) the normalized external branch length of the coalescent tree and (c) the time before the MRCA of $1000$ simulations with no mutation. The means of the distributions are represented by red lines, and the distribution assuming Kingman's coalescent are represented by dashed red lines.}}}\label{Fig:Cn_Ln_tMRCA_p0_all}
\end{center}
\end{figure}

Fig \ref{Fig:Cn_Ln_tMRCA_p0_sample} represents the results of the tests presented in Section \ref{sec:neutrality} depending on the number of simulations, for samples of 100 simulations (as will be done in Section \ref{sec:docoalescenttreessignificantlydifferfromKingman} in the body of the paper). We see that most tests accept neutrality, except Shapiro-Wilkinson test once, but this test is known to be overly conservative when the size of the sample is a bit large.

\begin{figure}[!ht]
\begin{center}
\begin{tabular}{ccc}
\includegraphics[height=5cm]{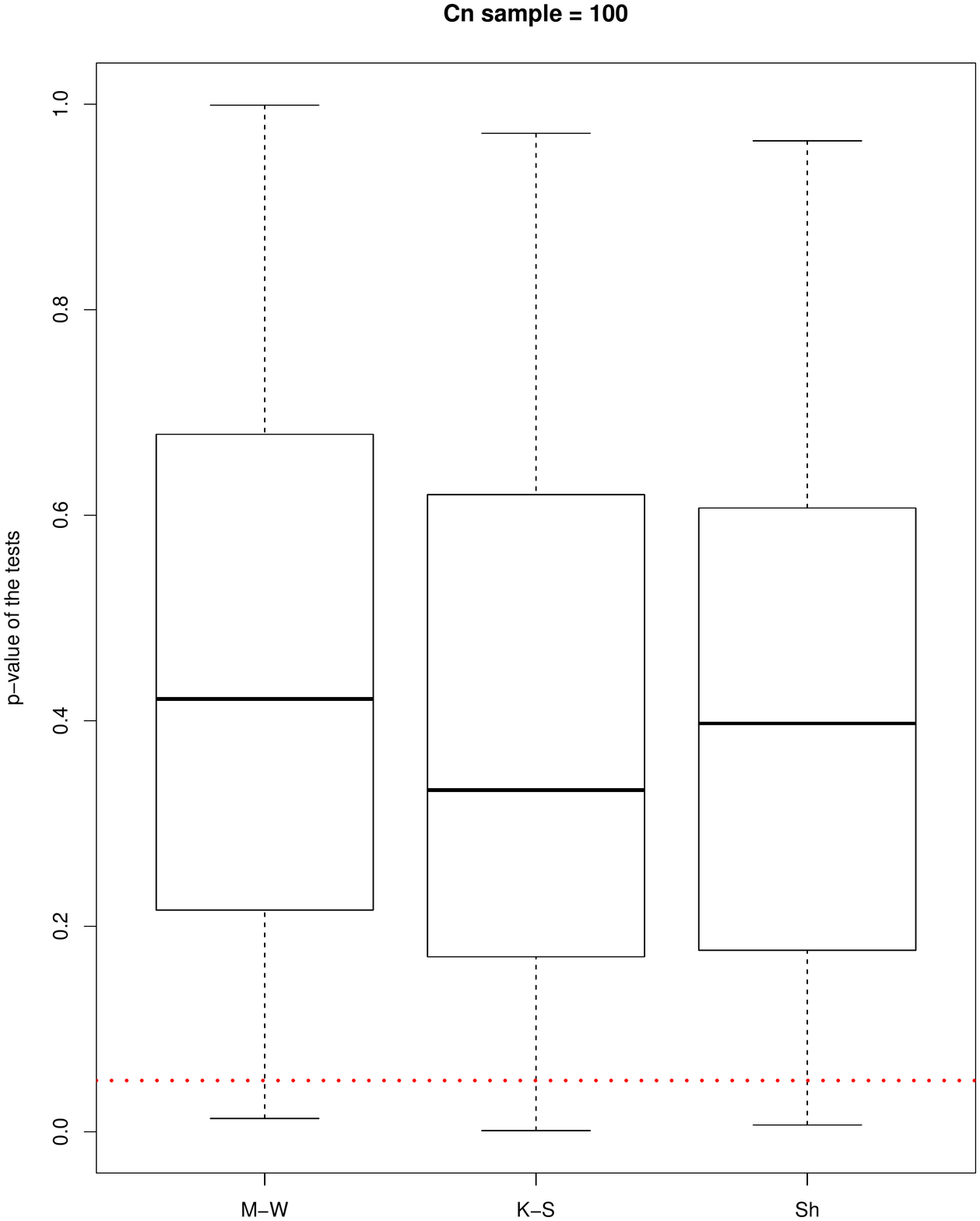} &
 \includegraphics[height=5cm]{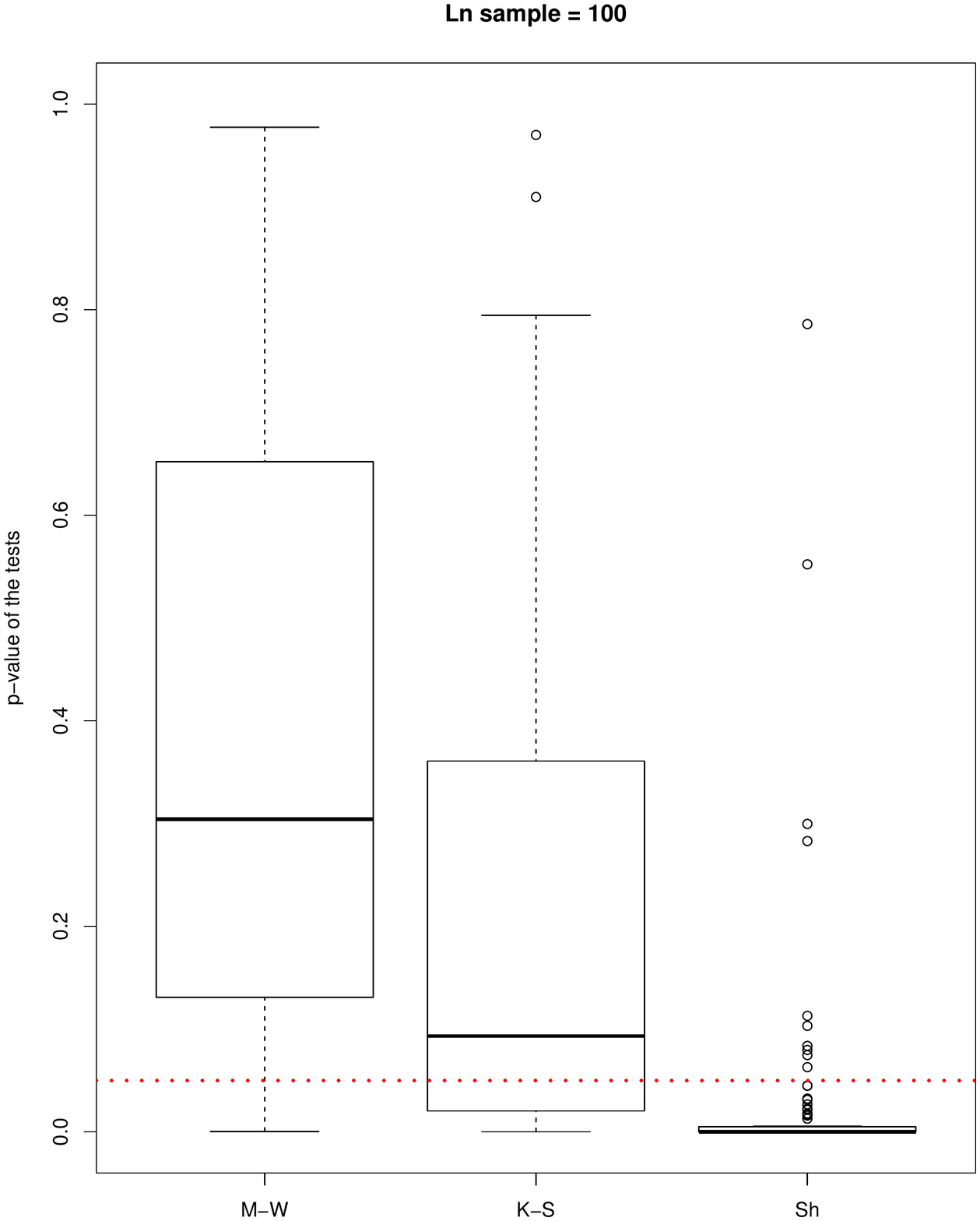} &
 \includegraphics[height=5cm]{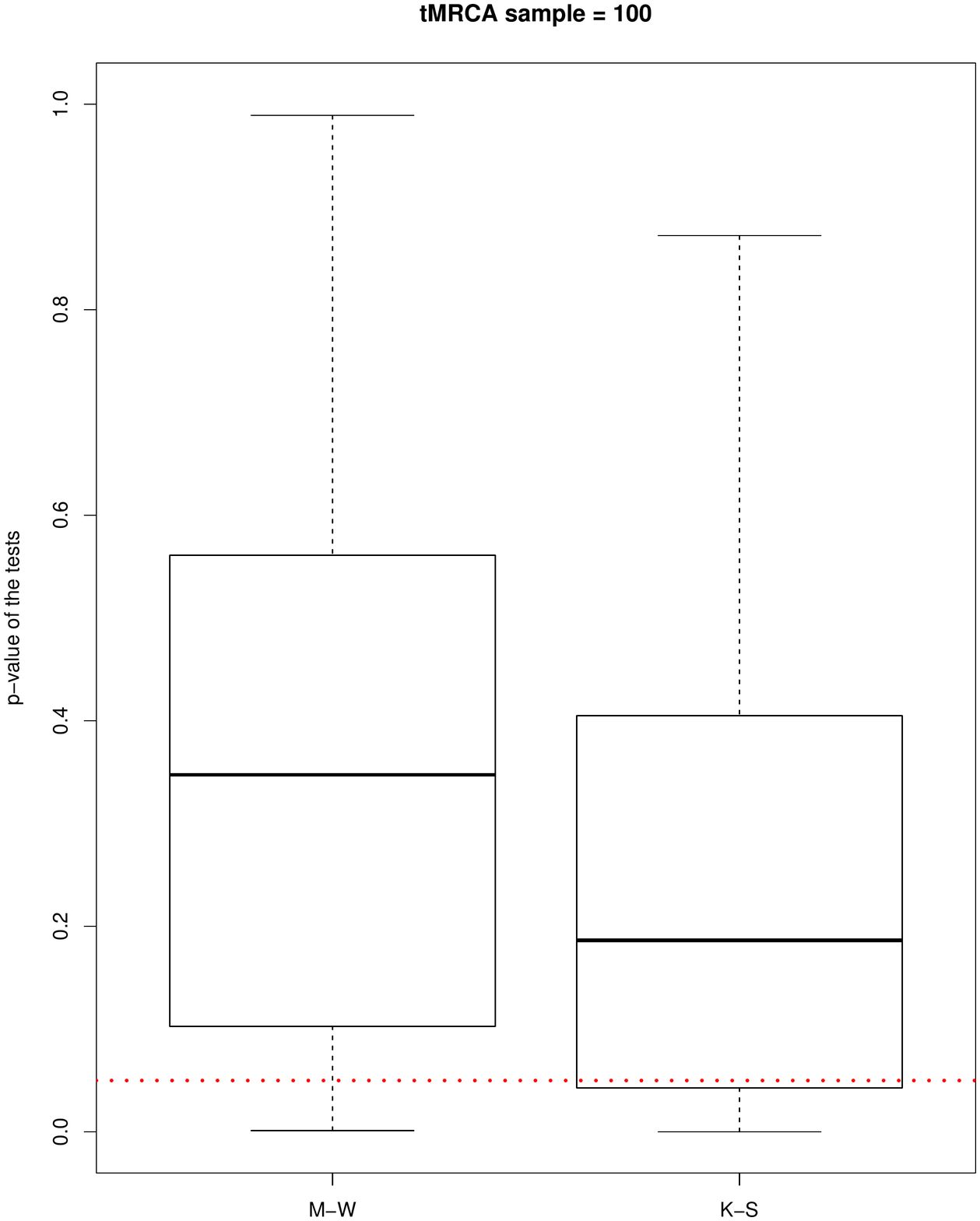}\\
 (a) & (b) & (c)
\end{tabular}
\caption{\textit{{\small
Tests of neutrality of three statistics depending on the number of sampled simulations: (a) distribution of normalized number of cherries ($C_n$)
(b) distribution of normalized external branch length of coalescent trees ($L_n$) and (c) time before the MRCA. Box-plots of the p-values of Mann-Whitney (M-W) and Kolmogorov-Smirnov (K-S) and Shapiro (Sh) tests for $100$ random samples of $100$ simulations are shown. The threshold value of rejection of $H_0$ (0.05) is represented by the dashed red line. If the p-values are inferior to this threshold, the distributions are statistically different from the targeted distribution under $H_0$. The p-value of the tests computed on all simulations are depicted by the blue symbols.}}} \label{Fig:Cn_Ln_tMRCA_p0_sample}
\end{center}
\end{figure}

\clearpage
\subsection{Neutrality test for the model with mutation}

\begin{figure}[ht!]
\begin{tabular}{cccc}
\includegraphics[height=3.5cm,width=4cm]{2019_10_16_Hist_Ln_simulA.pdf}&
\includegraphics[height=3.5cm,width=4cm]{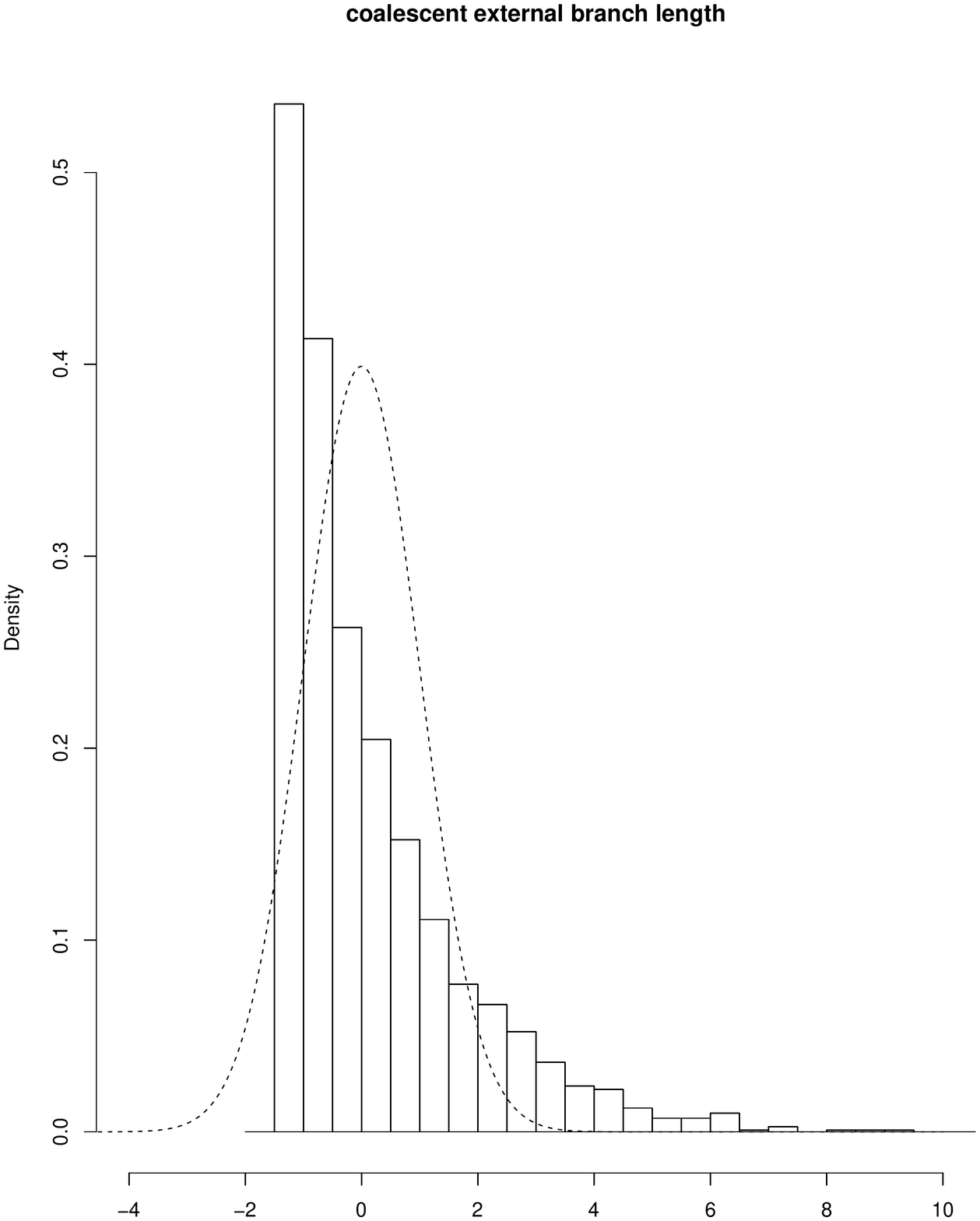} &
\includegraphics[height=3.5cm,width=4cm]{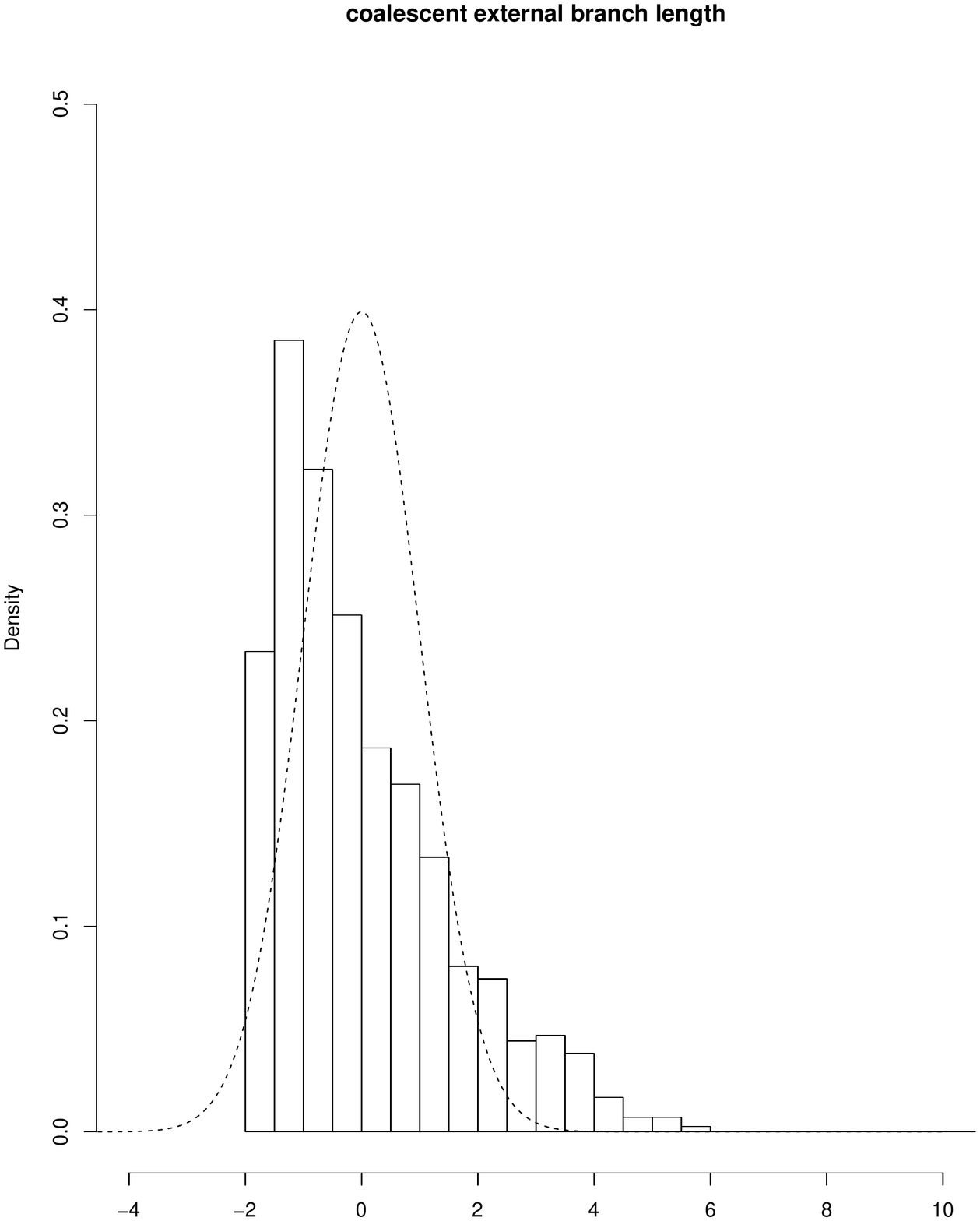} &
 \includegraphics[height=3.5cm,width=4cm]{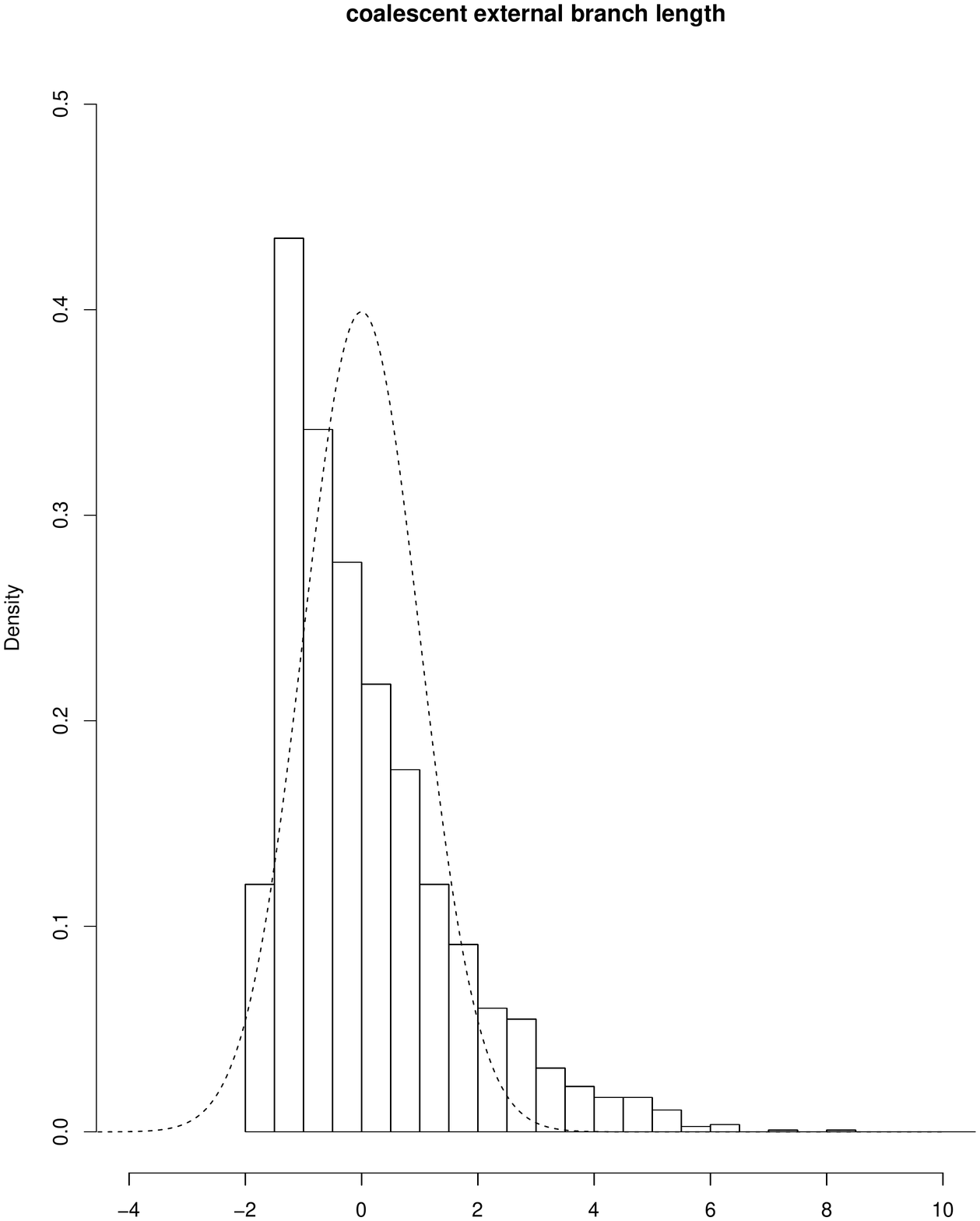}\\
 A & B & C & D
\end{tabular}
\caption{\textit{{\small Histograms of the renormalized external branch lengths produced by the ABC on the pseudo-data A to D. The dashed line represents the distribution followed by Kingman coalescent (Gaussian distribution)}}}\label{Fig:Branch-lines3-6}
\end{figure}

\begin{figure}[!ht]
\begin{tabular}{cccc}
\includegraphics[height=3.5cm,width=4cm]{2019_10_16_Hist_Cn_simulA.pdf} &
\includegraphics[height=3.5cm,width=4cm]{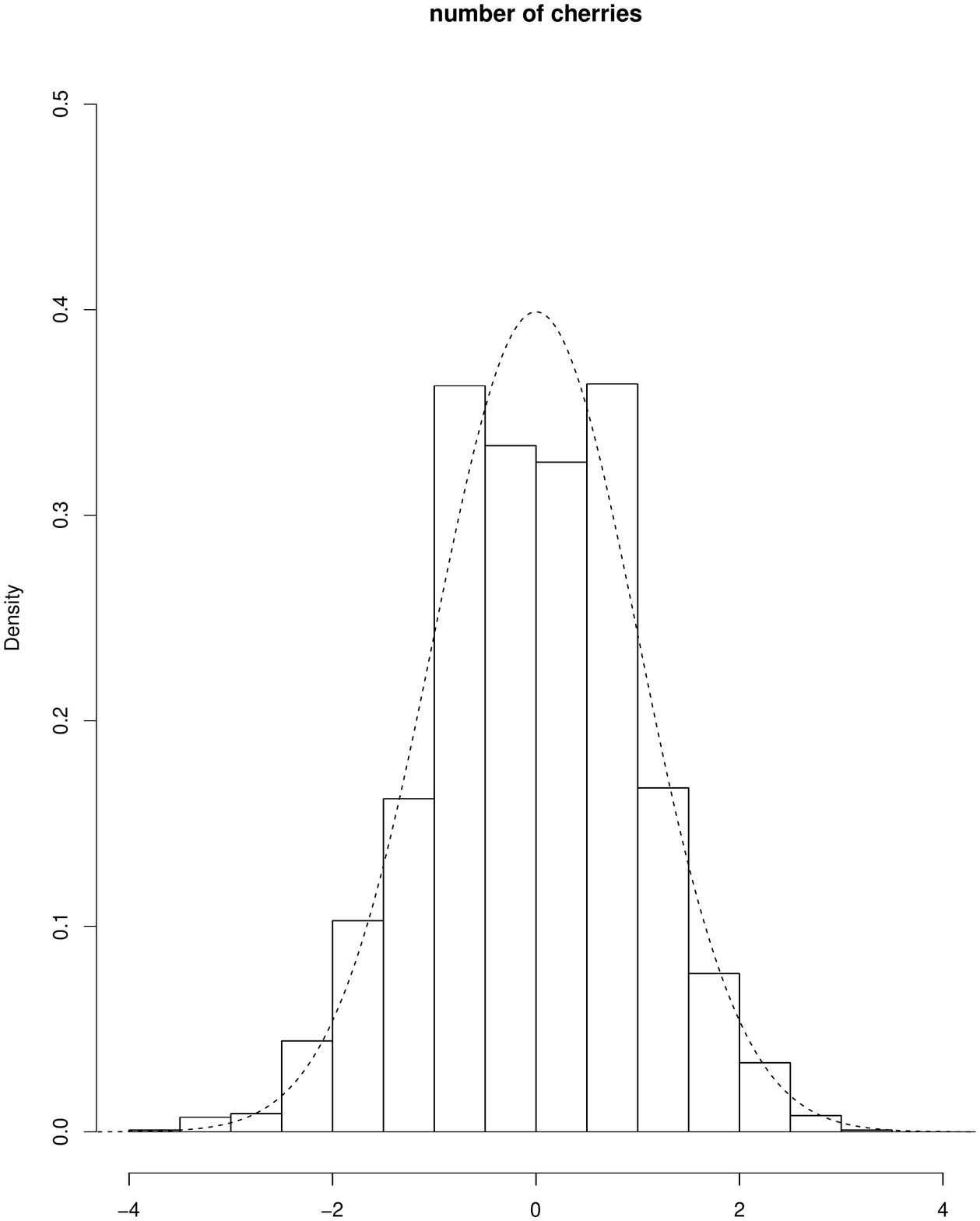} &
\includegraphics[height=3.5cm,width=4cm]{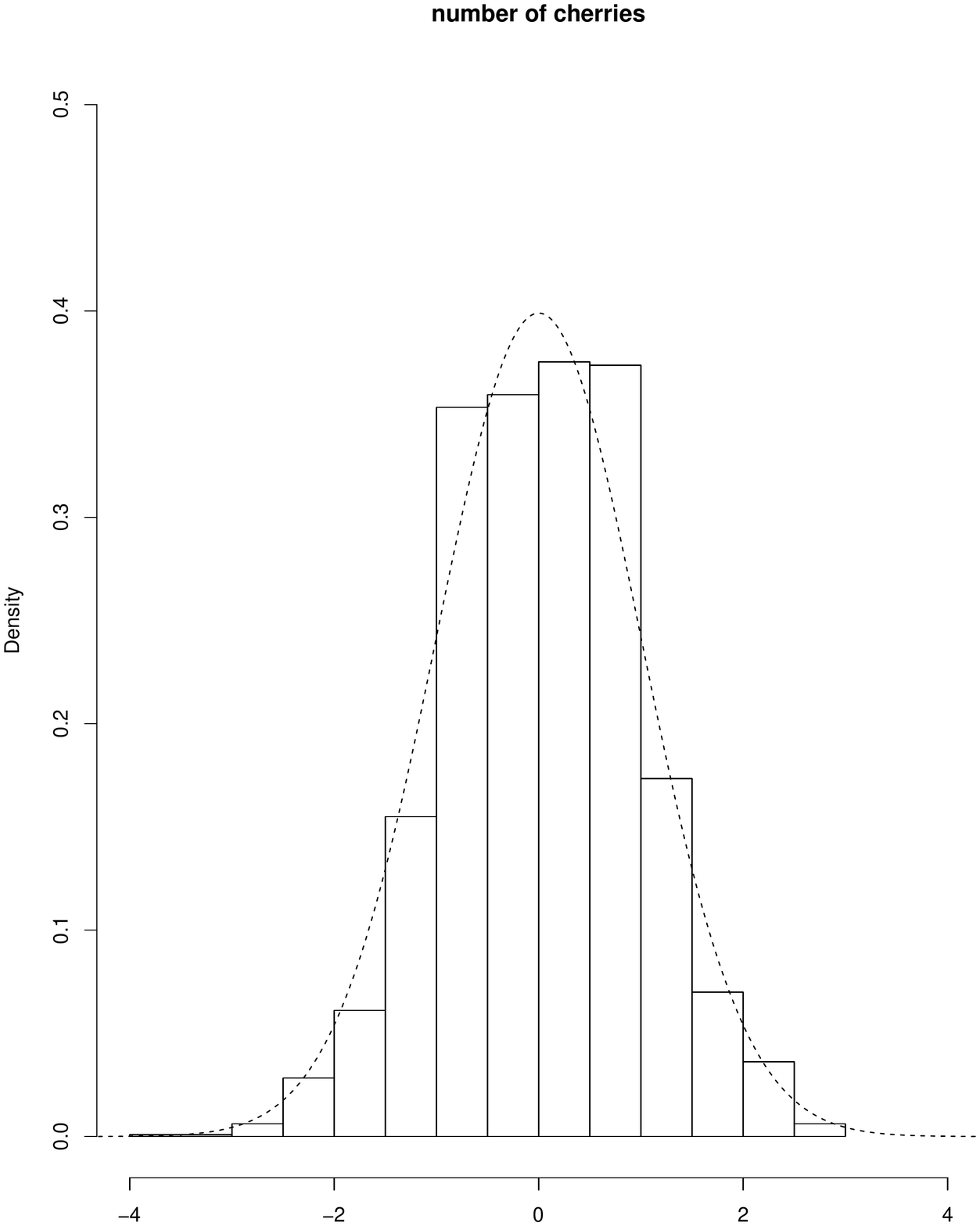} &
\includegraphics[height=3.5cm,width=4cm]{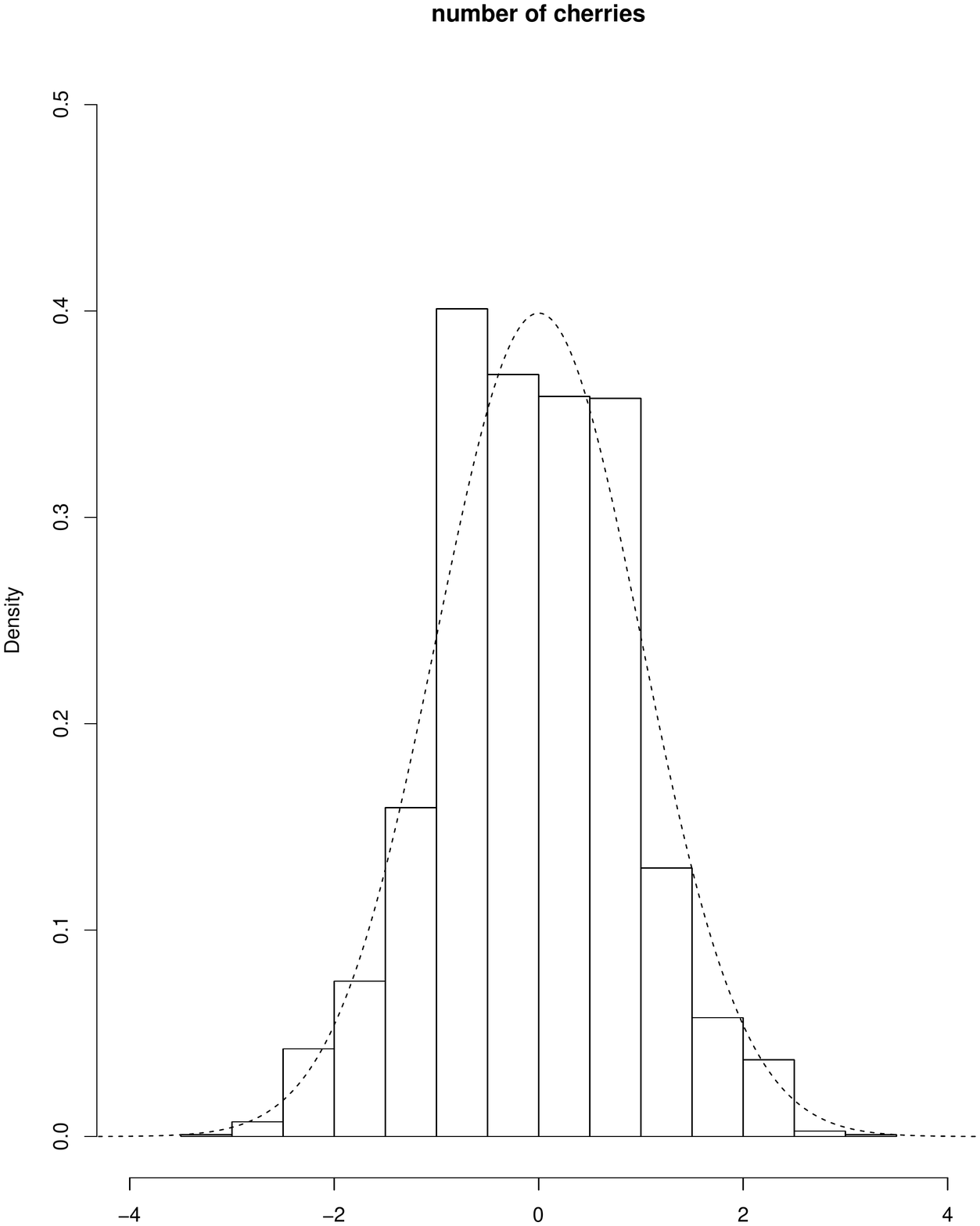}\\
A & B & C & D
\end{tabular}
\caption{\textit{{\small
Histograms of the renormalized number of cherries produced by the ABC on the pseudo-data A to D. The dashed line represents the distribution followed by Kingman coalescent (Gaussian distribution)}}}\label{Fig:Cherries-lines3-6}
\end{figure}

\begin{figure}[!ht]
\begin{tabular}{cccc}
\includegraphics[height=3.5cm,width=4cm]{2019_10_22_Hist_Tmrca_simulA.pdf} &
\includegraphics[height=3.5cm,width=4cm]{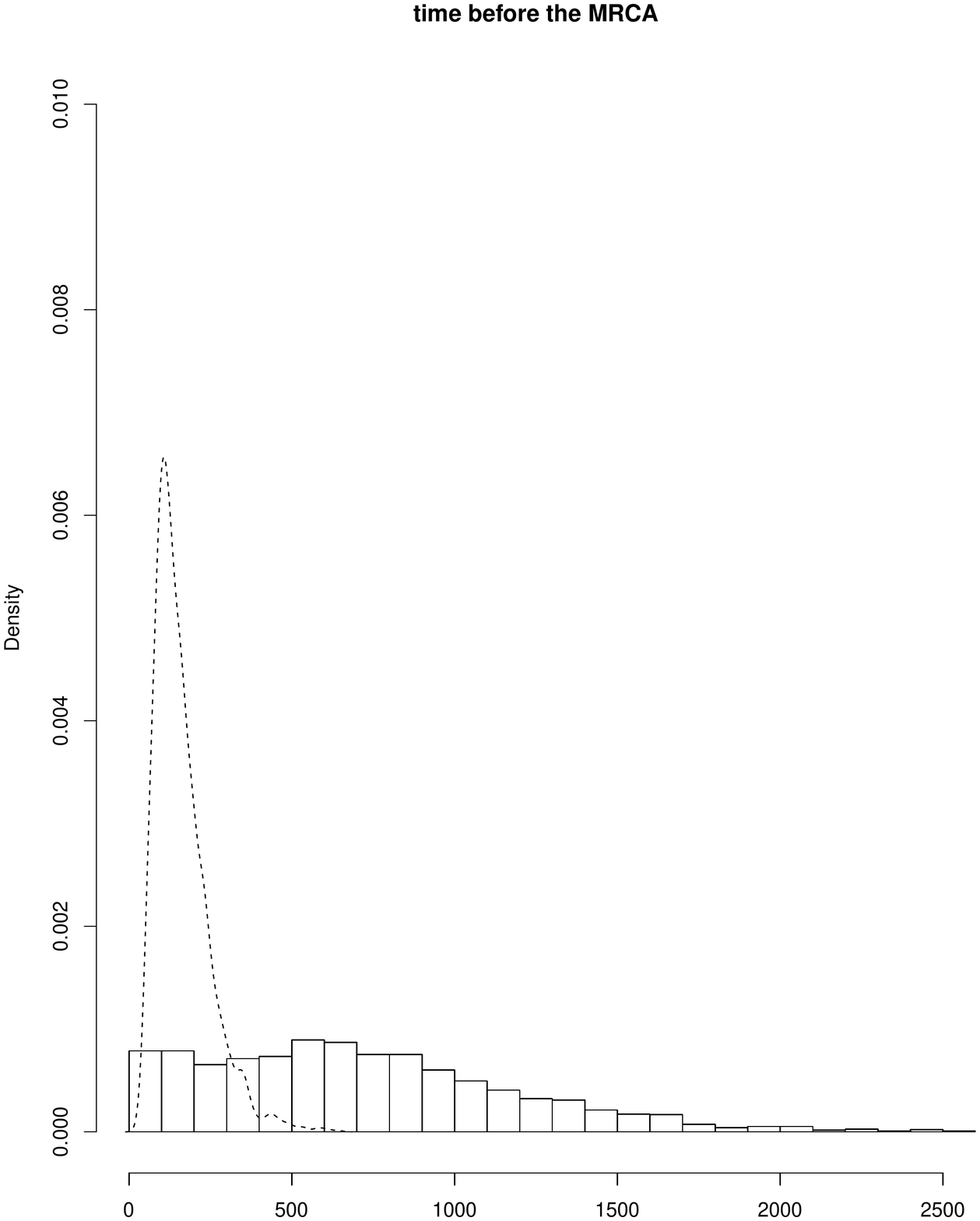} &
\includegraphics[height=3.5cm,width=4cm]{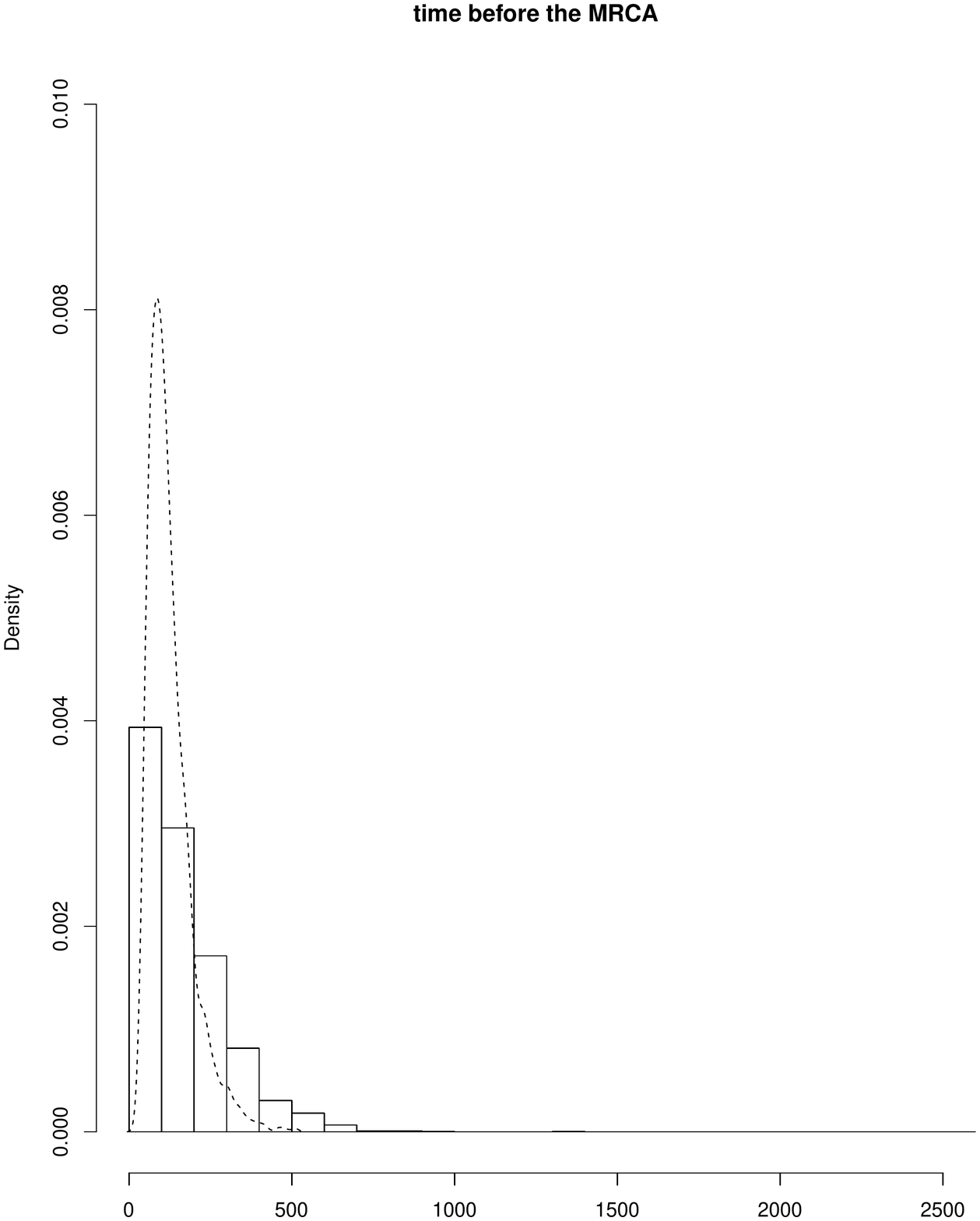} &
\includegraphics[height=3.5cm,width=4cm]{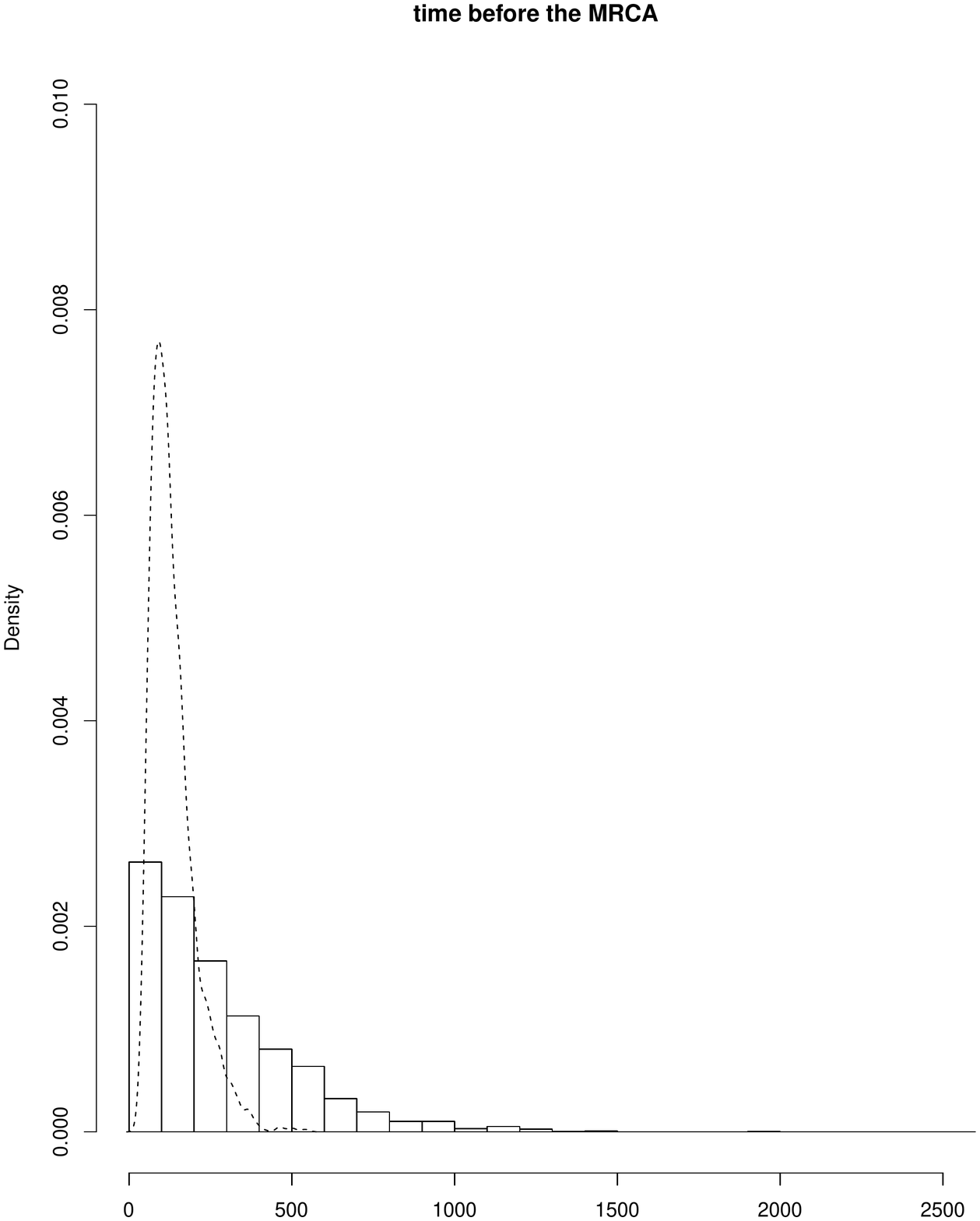}\\
A & B & C & D
\end{tabular}
\caption{\textit{{\small
Histograms of the time to the MRCA produced by the ABC on the pseudo-data A to D. The dashed line represents the distribution followed by Kingman coalescent (obtained by simulations)}}}\label{Fig:tMRCA-lines3-6}
\end{figure}

\section{Application to patrilineal and cognatic populations in Central Asia}

\subsection{Distances between populations}

In Figure \ref{Fig:carte}(a), it appears that the populations that are considered are distributed roughly along a curve that is plotted in Fig. \ref{Fig:carte}(b).
The interpolation curve corresponds to a polynomial of degree 3 giving the latitude $y$ as a function of the longitude $x$, and that is fitted by ordinary least squares:
\[y=P(x)=673.4-25.13 \ x+0.327\ x^2-1.39\ 10^{-3}\ x^3.\]
The $R^2$ associated with this regression is 92.52\%. \\

The populations are then projected on this curve and the distances between two locations are then computed using the line integral. Hence, two populations at locations $z_0=(x_0,y_0)$ and $z_1=(x_1,y_1)$ on the graph of $x\mapsto P(x)$ are considered at distance:
\[d(z_0,z_1)=\int_{x_0}^{x_1} \sqrt{1+P'(x)} dx.\]

\subsection{Test of the ABC procedure }\label{append:chaix-simul}
Our goal is to test the null hypothesis
\begin{equation}
H_0\ : \ b_0=b_1\label{test:H_0-annexe}
\end{equation}
with the alternative hypothesis $H_a\ : b_0<b_1$.
For that 20,000 simulations have been performed with parameters randomly drawn in their prior distribution with an additional constraint that either $b_0=b_1$ or $b_0<b_1$, constituting a total of 20,000 simulations. An additional constraint was that both social models survive, so that cases with $b_0<b_1$ and $b_0=b_1$ clearly have different dynamics. Two hundred pseudo-data sets are chosen among these 20,000 simulations: 100 with $b_0=b_1$ and 100 with $b_0<b_1$. \\

We first used the ABC procedure to estimate the posterior probabilities (see Fig. \ref{fig:ABC-Asian-Simul}). Our results show that for most parameters, the estimate parameters are close to their true values.
\clearpage

\begin{figure}[ht!]
\begin{center}
\begin{tabular}{cc}
(a) & \includegraphics[width=14cm]{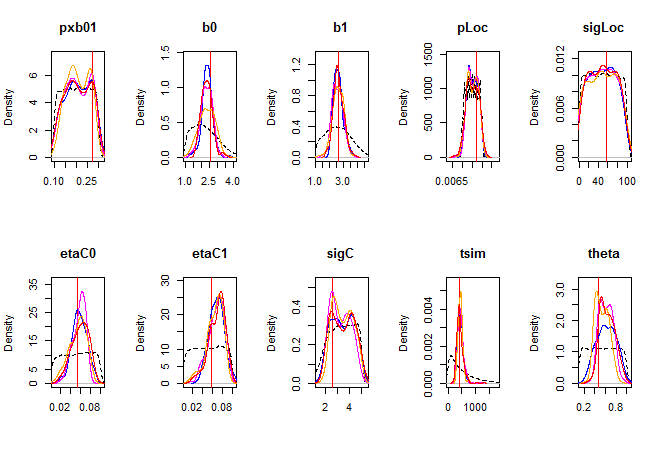}\\
(b) & \includegraphics[width=14cm]{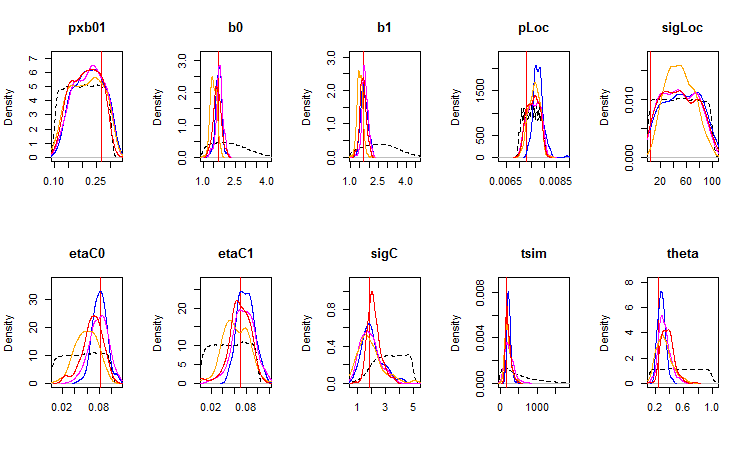}
\end{tabular}
\end{center}
\caption{{\small \textit{Results of the ABC estimation for one of the simulated datasets. (a) when the dataset satisfies $b_1>b_0$, (b) when the dataset satisfies $b_0=b_1$. Legend as in Fig.\ref{Fig:PriorPost_line3}.}}
}\label{fig:ABC-Asian-Simul}
\end{figure}

Second, the approximate posterior probabilities that $b_0=b_1$ and $b_0<b_1$, respectively noted $\widehat{\P}(H_0 \ |\ S_{\sc{obs}})$ and $\widehat{\P}(H_a\ |\ S_{\sc{obs}})=1-\widehat{\P}(H_0 \ |\ S_{\sc{obs}})$, were computed such that
\begin{equation}\label{eq:teststat}
\widehat{\P}(H_a\ |\ S_{\sc{obs}})=\sum_{i=1}^{20,000} \ind_{b_1^{(i)}>b_0^{(i)}} W_i,
\end{equation}
where $b_0^{(i)}$ and $b_1^{(i)}$ are the birth rate in the simulation $i\in \{1,\dots 20,000\}$, and where $W_i$ is the weight defined in App. \ref{sec:ABC-recap}. We compute these quantities \eqref{eq:teststat} for the real data from Central Asia, but also for 200 `training datasets': 100 datasets chosen from the simulations with $b_0=b_1$ presented above and 100 datasets chosen from the simulations with $b_0<b_1$. Computing \eqref{eq:teststat} for these 200 training sets provides an indication of the distribution of this statistics under $H_0$ and $H_a$: for example, we find that the medians of the empirical distributions under $H_0$ and $H_a$ are respectively $0.4335$ and $0.5843$. This is represented in Fig. \ref{fig:test1} (a).\\

The $H_0$ hypothesis is rejected if $\widehat{\P}(H_a\ |\ S_{\sc{obs}})>\alpha$ where $\alpha$ is a threshold parameter chosen in order to minimize the probabilities of Type I and Type II errors. To approximate these, we will use the 200 `training datasets' defined above. A natural choice is $\alpha=0.5$ that lies in the middle of the two medians computed in the preceding paragraph. Another possibility is to minimize the sum of these probabilities for different $\alpha$ (Fig. \ref{fig:test1}(c)), which leads to $\alpha=0.49$, close to $\alpha=0.5$.\\
For $\alpha=0.5$, approximation of the probabilities of Types I and II errors based on the 200 tests performed on simulations are shown in Table \ref{tab:test}.

\begin{table}[!ht]
\begin{center}
\begin{tabular}{|c|c|c|}
\hline
\multicolumn{3}{|c|}{$\alpha=0.5$}\\
\hline
& $b_0=b_1$ & $b_1>b_0$\\
\hline
not reject $H_0$ & 64 \%& 25\%\\
reject $H_0$ & 36\% & 75\%\\
\hline
\end{tabular}
\caption{{\small \textit{Confusion matrix based on the 200 tests performed on simulated data with the threshold $\alpha=0.50$. Each of the 200 tests are based on ABC estimation using 20,000 simulations, and choosing in turn 200 of these to play the role of the dataset: 100 with $b_0=b_1$ and 100 with $b_1>b_0$. We have computed the type I (resp. II) errors by counting the proportion of the 100 tests using data with $b_0=b_1$ that has lead to the choice $H_a$ using the threshold $\alpha=0.50$ (resp. with $b_0<b_1$ that has lead to $H_0$).}}}\label{tab:test}
\end{center}
\end{table}

\begin{figure}[!ht]
\begin{center}
\begin{tabular}{ccc}
\includegraphics[width=7.5cm,angle=0,height=5cm]{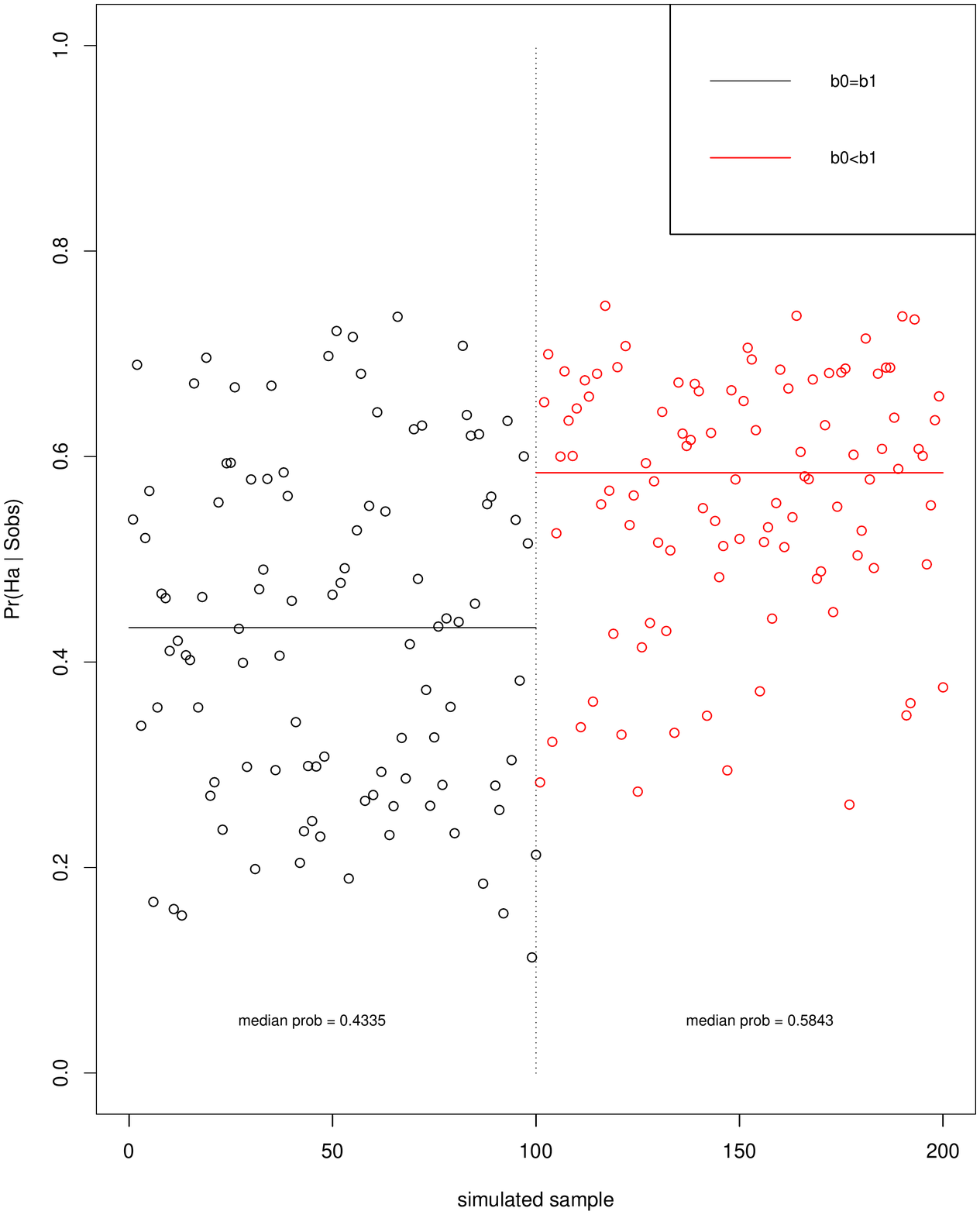} &
\includegraphics[width=7.5cm,angle=0,height=5cm]{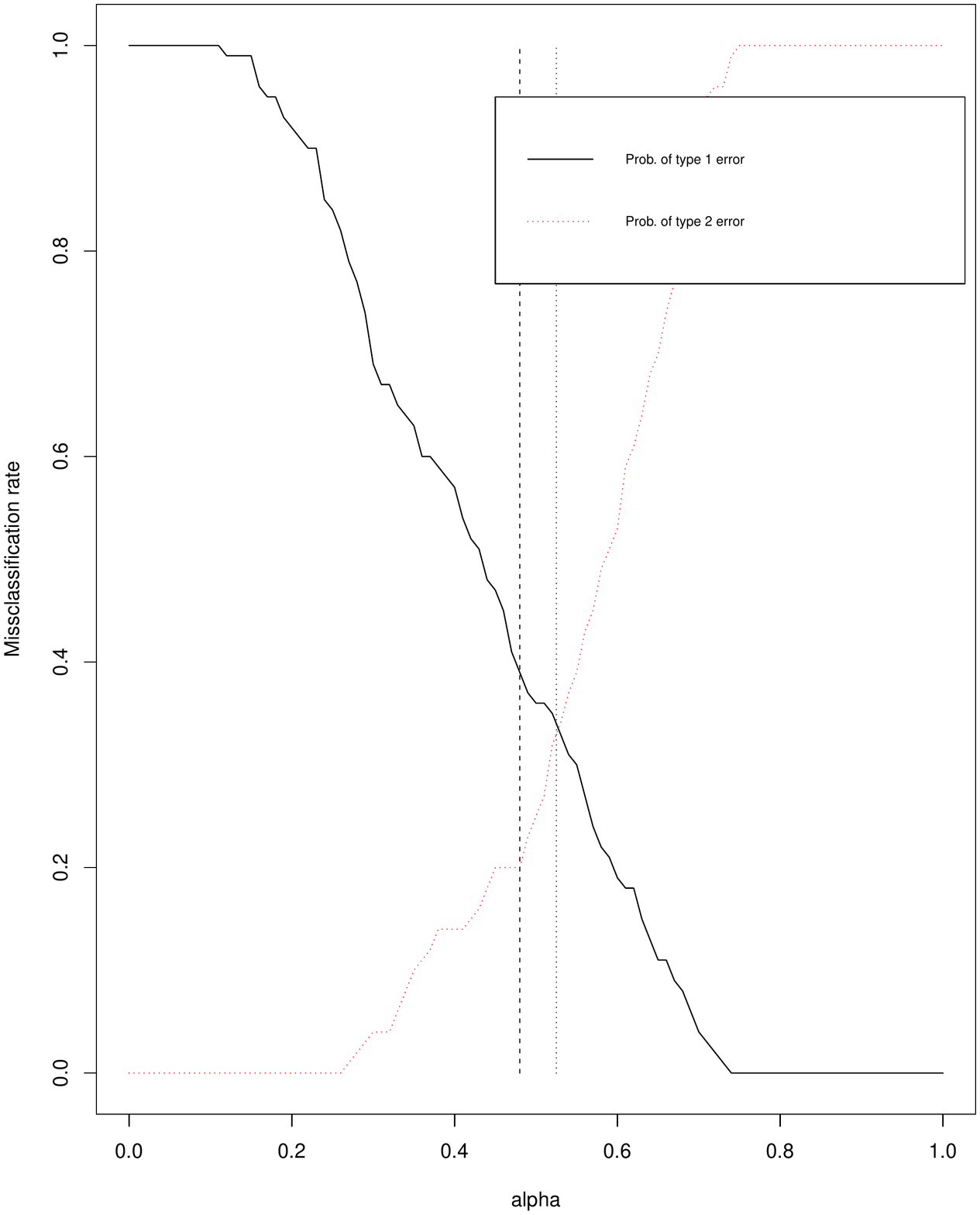}
&
\includegraphics[width=7.5cm,angle=0,height=5cm]{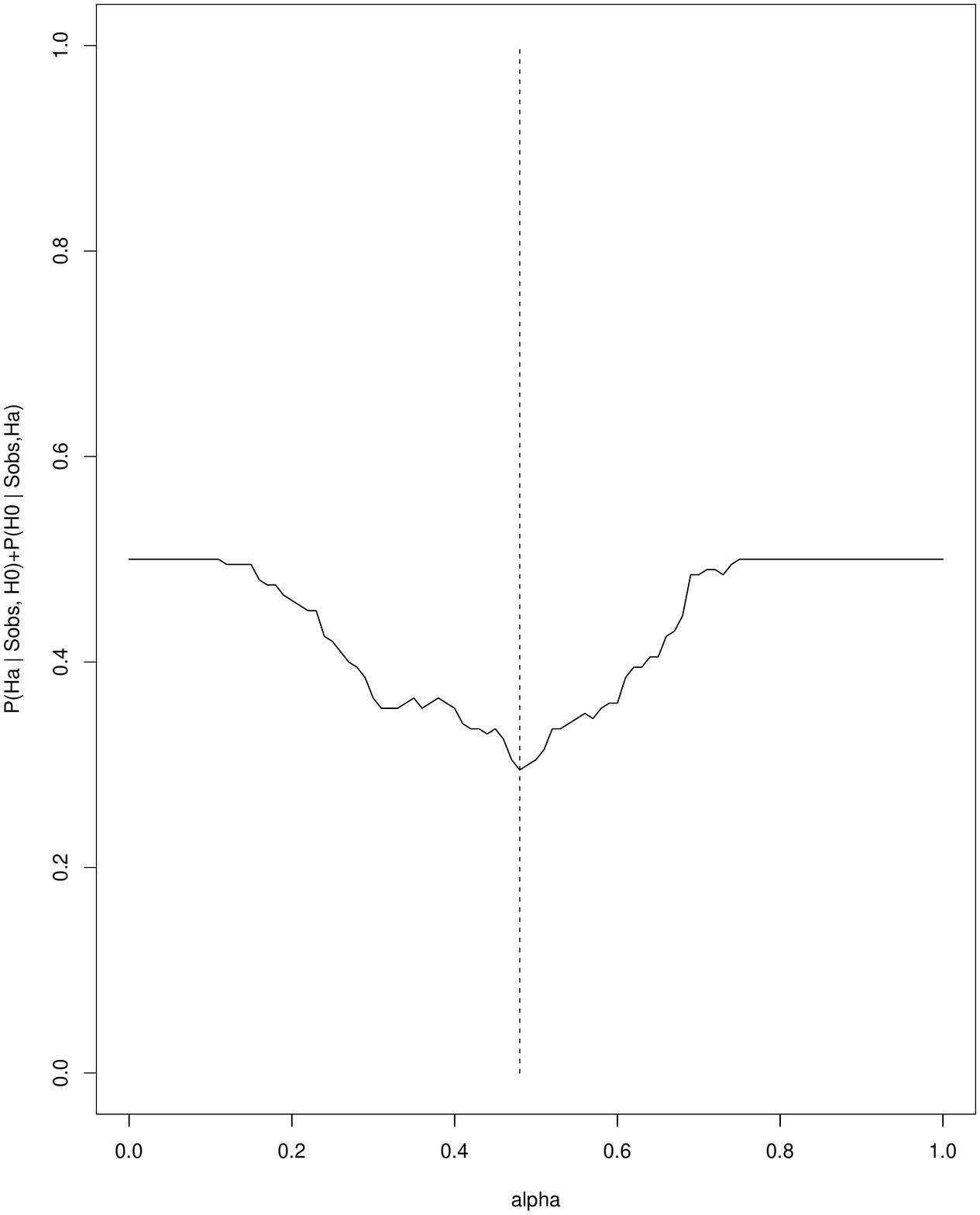}
\\
(a) & (b) & (c)
 \end{tabular}
\caption{{\small \textit{(a): Estimated posterior probability that $b_0<b_1$ conditional on the summary statistics, for each of the 200 simulations chosen in turn to be the `data'. The index of the simulation is in abscissa. The 100 first simulations have been obtained under the constraint $b_0=b_1$ while the 100 last ones are under the constraint $b_0<b_1$. For each simulation, the posterior probability $\widehat{P}(H_a\ |\ S_{\sc{obs}})$ is computed using \eqref{eq:teststat}. The two plain horizontal lines correspond to the medians of these estimations for the 100 simulations where $b_0=b_1$ and for the 100 simulations where $b_0<b_1$: these medians are respectively $0.4335$ and $0.5843$. (b): The estimated probabilities of Type I and Type II errors are plotted as a function of the threshold $\alpha$ defining the critical region: $\{\widehat{\P}(H_a\ |\ S_{\sc{obs}})>\alpha\}$. The intersection of these two curves corresponds to $\alpha=0.53$ (vertical dotted line). The value $\alpha=0.49$ is indicated in dashed line. (c): The sum of the estimated probabilities of Type I and Type II errors is plotted as a function of $\alpha$. The sum is minimal when $\alpha=0.49$ (dashed line).  }}}\label{fig:test1}
\end{center}
\end{figure}

\subsection{ABC on the Central Asian database}\label{append:chaix-donnees}
We performed the statistical test \eqref{test:H_0} on the Central Asian dataset. Using the same ABC framework with the same 20,000 simulations, we computed \eqref{eq:teststat} and performed the ABC test presented as with the pseudo-data sets (App. \ref{append:chaix-simul}).
We find
$\widehat{\P}(H_a\ |\ S_{\sc{obs}})=0.4518$ which is below the threshold $\alpha=0.50$.

We are in the acceptance region and the null hypothesis $H_0\ : \ b_0=b_1$ can not be rejected.
Hence, the test concludes that there is no significant difference of fertility rates between the two social organizations: patrilineal and cognatic.\\

Using the 100 simulations where $b_0=b_1$, the p-value of the test, estimated as the proportion of these simulations where $\widehat{\P}(H_a\ |\ S_{\sc{obs}})\geq 0.4518$, can be estimated to 47\%.

\end{document}